%% file: RTE2RFT-arxiv.tex
\pdfoutput=1
\documentclass[a4paper,UKenglish]{lipics-v2019}
\usepackage{microtype}%if unwanted, comment out or use option "draft"

\nolinenumbers
\hideLIPIcs

\input{macros}

\usepackage[pdflatex,recompilepics=false]{gastex}

\newcounter{todocounter}

\title{Efficient Construction of Reversible Transducers from  Regular Transducer Expressions}

\author{Luc Dartois}{Univ Paris Est Creteil, LACL, F-94010 Creteil, France}{luc.dartois@lacl.fr}{https://orcid.org/0000-0001-9974-1922}{}
\author{Paul Gastin}{Université Paris-Saclay, ENS Paris-Saclay, CNRS, LMF, 91190, Gif-sur-Yvette, France}{paul.gastin@lsv.fr}{https://orcid.org/0000-0002-1313-7722}{}
\author{R. Govind}{IIT Bombay, India} {govindr@cse.iitb.ac.in}{https://orcid.org/0000-0002-1634-5893}{}
\author{Shankara Narayanan Krishna}{IIT Bombay, India} {krishnas@cse.iitb.ac.in}{https://orcid.org/0000-0003-0925-398X}{}

\authorrunning{L. Dartois, P. Gastin, R. Govind and  S. Krishna}
\Copyright{Luc Dartois, Paul Gastin, R. Govind and S. Krishna}



\category{} 

\relatedversion{} 

\supplement{}

\funding{Supported by IRL ReLaX}
\hideLIPIcs

\begin{document}

\maketitle
\begin{abstract}
  The class of regular transformations has several equivalent characterizations such as functional MSO transductions, deterministic two-way transducers, streaming string transducers, as well as regular transducer expressions (RTE).

  For algorithmic applications, it is very common and useful to transform a specification, here, an RTE, to a machine, here, a transducer. In this paper, we give an efficient construction of a two-way reversible transducer (2RFT) equivalent to a given RTE. 2RFTs are a well behaved class of transducers which are deterministic and co-deterministic (hence allows evaluation in linear time \wrt the input word), and where composition has only polynomial complexity.
  
  We show that, for full RTE, the constructed 2RFT has size doubly exponential in the size of the
  expression, while,  if the RTE does not use   Hadamard product or chained-star, the constructed
  2RFT has size exponential in the size of the RTE.  
 \end{abstract}

\section{Introduction}
%%%%%%%%%%%%%%%%%%%%%%%%%%%%%%%%%%%%%%%%%%%%%%%%%%%%%%%%%%%%%%%%%%%%
%%%   FIGURES
%%%%%%%%%%%%%%%%%%%%%%%%%%%%%%%%%%%%%%%%%%%%%%%%%%%%%%%%%%%%%%%%%%%%
\begin{gpicture}[name=Overview,ignore]
  \unitlength=1mm
  \gasset{Nw=8,Nh=8}
  \put(-30,115){\gasset{Nadjust=wh,Nmr=0}
  \node[Nframe=n](A1)(5,0){$w$}
  \node(A2)(29,0){\begin{tabular}{c}
    Parser $\TParse_h$ \small (1NFT) \\ \small $\norm{\TParse_h} \leq |h|^{\w{h}}$
  \end{tabular}}
  \node[Nframe=n](A3)(60,0){\begin{tabular}{c} Parsings \\ of $w$\end{tabular}}
  \node(A4)(89,0){\begin{tabular}{c}
    Evaluator $\T_h$ \small (2RFT) \\ \small $\norm{\T_h} \leq 5|h|\cdot\w{h}$
  \end{tabular}}
  \node[Nframe=n](A5)(123,0){$\gamma\in\rsem{h}(w)$}
  \node[Nframe=n](A7)(57,0){}
  \node[Nframe=n](A8)(63,0){}
  \drawedge(A1,A2){}
  \drawedge(A2,A7){}
  \drawedge(A8,A4){}
  \drawedge(A4,A5){}
  }

  \put(-30,100){\gasset{Nadjust=wh,Nmr=0}
  \node[Nframe=n](A1)(5,0){$w$}
  \node(A2)(30,0){\begin{tabular}{c}
    Parser $\TParse^{U}_h$ \small (2RFT) \\[1ex] \small 
    $\norm{\TParse^{U}_h}=2^{\mathcal{O}(|h|^{2\cdot \w{h}})}$
  \end{tabular}}
  \node[Nframe=n](A3)(61,0){\begin{tabular}{c} Unambiguous \\ parsing of $w$\end{tabular}}
  \node(A4)(92,0){\begin{tabular}{c}
    Evaluator $\T_h$ \small (2RFT) \\ \small $\norm{\T_h} \leq 5|h|\cdot\w{h}$
  \end{tabular}}
  \node[Nframe=n](A5)(124,0){$\usem{h}(w)$}
  \node[Nframe=n](A7)(57,0){}
  \node[Nframe=n](A8)(63,0){}
  \drawedge(A1,A2){}
  \drawedge(A2,A7){}
  \drawedge(A8,A4){}
  \drawedge(A4,A5){}
  }

  \put(-30,85){\gasset{Nadjust=wh,Nmr=0}
  \node[Nframe=n](A1)(15,0){$w$}
  \node(A2)(56,0){\begin{tabular}{c}
    $B$ (1NFA) \qquad $\norm{B} = 2\cdot\norm{\TParse_h}^2$ \\[1ex]
    Non-functionality checker for $\TParse_h$ 
  \end{tabular}}
  \node[Nframe=n](A3)(112,5){$w \in \dom{h}\setminus\udom{h}$}
  \node[Nframe=n](A4)(112,-5){$w \in \udom{h} \cup \overline{\dom{h}}$}
  \drawedge(A1,A2){}
  \drawedge[ELpos=60](A2,A3){Yes}
  \drawedge[ELside=r,ELpos=60](A2,A4){No}
  }

  \put(-30,70){\gasset{Nadjust=wh,Nmr=0}
  \node[Nframe=n](A1)(15,0){$w$}
  \node(A2)(56,0){\begin{tabular}{c}
    $\TParse'_h$ (2RFT) \scriptsize $\norm{\TParse'_h}=2^{\mathcal{O}(\norm{\TParse_h})}$ \\[1ex]
    Uniformizer of $\TParse_h$ 
  \end{tabular}}
  \node[Nframe=n](A3)(105,0){Some parsing of $w$}
  \drawedge(A1,A2){}\drawedge(A2,A3){}
  }

  \put(-30,47){\gasset{Nadjust=wh,Nmr=0}
  \node[Nframe=n](A1)(5,0){$w$}
  \node(A2)(29,0){\begin{tabular}{c}
    $B'$ (2RFA) \\ [1ex] \small 
    Complement of $B$\\ [1ex]
    $\norm{B'}=2^{\mathcal{O}(\norm{\TParse_h})^{2}}$
  \end{tabular}}
  \node[Nframe=n](A3)(60,0){\begin{tabular}{c} $w$ if in \\ \\
    \small $\udom{h} \cup \overline{\dom{h}}$
  \end{tabular}}
  \node(A4)(91,0){\begin{tabular}{c}
    $\TParse'_h$ (2RFT) \\[1ex] \small $\norm{\TParse'_h}=2^{\mathcal{O}(\norm{\TParse_h})}$
  \end{tabular}}
  \node[Nframe=n](A5)(123,0){\begin{tabular}{c} \footnotesize Unambiguous \\ \footnotesize parsing of $w$\end{tabular}}
  \drawedge(A1,A2){}
  \node[Nframe=n](A7)(57,0){}
  \node[Nframe=n](A8)(63,0){}
  \node[Nframe=n](A9)(116,0){}
  \drawedge(A2,A7){}
  \drawedge(A8,A4){}
  \drawedge(A4,A9){}
  \node[Nadjust=n,Nw=100,Nh=27,Nmr=5](P)(60,0){}
  \node[Nframe=n](*)(100,10){$\TParse^U_h$}
  }
\end{gpicture}  
\begin{gpicture}[name=generic 2RFT,ignore]
  \gasset{Nw=5,Nh=5,loopdiam=4,curvedepth=0,ilength=4,flength=4}
  \unitlength=.75mm
  \node[Nmarks=i](1)(0,0){$+$}
  \node(2)(20,0){$+$}
  \node(3)(40,0){$+$}
  \node[Nmarks=f](4)(60,0){$+$}
  \drawedge(1,2){$\lparent{h}\mid\Blue{\varepsilon}$}
  \drawedge(3,4){$\rparent{h}\mid\Blue{-}$}
  \node[Nw=27,Nh=15,Nmr=5](h)(30,0){$\T_h$}
\end{gpicture}
\begin{gpicture}[name=Tv,ignore]
  \gasset{Nw=5,Nh=5,loopdiam=4,curvedepth=0,ilength=4,flength=4}
  \unitlength=.75mm
  \node[Nmarks=i](1)(0,0){$+$}
  \node(2)(20,0){$+$}
  \node[Nmarks=f](3)(40,0){$+$}
  \drawedge(1,2){$\lparent{h}\mid\Blue{\varepsilon}$}
  \drawloop(2){$a\in\Sigma\mid\Blue{\varepsilon}$}
  \drawedge(2,3){$\rparent{h}\mid\Blue{v}$}
\end{gpicture}
\begin{gpicture}[name=Tduplicate,ignore]
  \gasset{Nw=5,Nh=5,loopdiam=4,curvedepth=0,ilength=4,flength=4}
  \unitlength=.75mm
  \node[Nmarks=i](1)(0,0){$+$}
  \node(2)(20,0){$+$}
  \node(3)(40,0){$-$}
  \node(4)(60,0){$+$}
  \node[Nmarks=f](5)(80,0){$+$}
  \drawedge(1,2){$\lparent{h}\mid\Blue{\varepsilon}$}
  \drawloop(2){$a\mid a$}
  \drawedge(2,3){$\rparent{h}\mid \#$}
  \drawloop(3){$a\mid\Blue{\varepsilon}$}
  \drawedge(3,4){$\lparent{h}\mid\Blue{\varepsilon}$}
  \drawloop(4){$a\mid a$}
  \drawedge(4,5){$\rparent{h}\mid\Blue{\varepsilon}$}
\end{gpicture}
\begin{gpicture}[name=Treverse,ignore]
  \gasset{Nw=5,Nh=5,loopdiam=4,curvedepth=0,ilength=4,flength=4}
  \unitlength=.75mm
  \node[Nmarks=i](1)(0,0){$+$}
  \node(2)(20,0){$+$}
  \node(3)(40,0){$-$}
  \node(4)(60,0){$+$}
  \node[Nmarks=f](5)(80,0){$+$}
  \drawedge(1,2){$\lparent{h}\mid\Blue{\varepsilon}$}
  \drawloop(2){$a\mid\Blue{\varepsilon}$}
  \drawedge(2,3){$\rparent{h}\mid\Blue{\varepsilon}$}
  \drawloop(3){$a\mid a$}
  \drawedge(3,4){$\lparent{h}\mid\Blue{\varepsilon}$}
  \drawloop(4){$a\mid\Blue{\varepsilon}$}
  \drawedge(4,5){$\rparent{h}\mid\Blue{\varepsilon}$}
\end{gpicture}
\begin{gpicture}[name=Tplus,ignore]
  \gasset{Nw=5,Nh=5,loopdiam=4,curvedepth=0,ilength=4,flength=4}
  \unitlength=.75mm
  \node[Nmarks=i](1)(0,0){$+$}
  \node[fillcolor=red!20!white](2)(20,0){$+$}
  \node[fillcolor=blue!20!white](3)(80,0){$+$}
  \node[Nmarks=f](4)(100,0){$+$}
  \drawedge(1,2){$\lparent{h}\mid\Blue{\varepsilon}$}
  \drawedge(3,4){$\rparent{h}\mid\Blue{\varepsilon}$}
  \node[Nw=27,Nh=15,Nmr=5](f)(50,10){$\T_f$}
  \node[Nw=27,Nh=15,Nmr=5](g)(50,-10){$\T_g$}
  \node(*)(40,10){$+$}
  \drawedge[curvedepth=4](2,*){$\lparent{f}\mid\Blue{\varepsilon}$}
  \node(*)(40,-10){$+$}
  \drawedge[curvedepth=-4,ELside=r](2,*){$\lparent{g}\mid\Blue{\varepsilon}$}
  \node(*)(60,10){$+$}
  \drawedge[curvedepth=4](*,3){$\rparent{f}\mid\Blue{-}$}
  \node(*)(60,-10){$+$}
  \drawedge[curvedepth=-4,ELside=r](*,3){$\rparent{g}\mid\Blue{-}$}
\end{gpicture}
\begin{gpicture}[name=Tcauchy,ignore]
  \gasset{Nw=5,Nh=5,loopdiam=4,curvedepth=0,ilength=4,flength=4}
  \unitlength=.75mm
  \node[Nmarks=i](1)(0,0){$+$}
  \node(2)(20,0){$+$}
  \drawedge(1,2){$\lparent{h}\mid\Blue{\varepsilon}$}
  \node(3)(40,0){$+$}
  \drawedge(2,3){$\lparent{f}\mid\Blue{\varepsilon}$}
  \node(4)(60,0){$+$}
  \node(5)(80,0){$+$}
  \drawedge(4,5){$\rparent{f}\mid\Blue{-}$}
  \node(6)(100,0){$+$}
  \drawedge(5,6){$\lparent{g}\mid\Blue{\varepsilon}$}
  \node(7)(120,0){$+$}
  \node(8)(140,0){$+$}
  \drawedge(7,8){$\rparent{g}\mid\Blue{-}$}
  \node[Nmarks=f](9)(160,0){$+$}
  \drawedge(8,9){$\rparent{h}\mid\Blue{\varepsilon}$}
  \node[Nw=27,Nh=15,Nmr=5](f)(50,0){$\T_f$}
  \node[Nw=27,Nh=15,Nmr=5](g)(110,0){$\T_g$}
\end{gpicture}
\begin{gpicture}[name=Tcauchyrev,ignore]
  \gasset{Nw=5,Nh=5,loopdiam=3,curvedepth=0,ilength=4,flength=4}
  \unitlength=.65mm
  \node[Nmarks=i](1)(0,0){$+$}
  \node(2)(20,0){$+$}
  \drawloop(2){\scriptsize$\alpha\mid\Blue{\varepsilon}$}
  \drawedge(1,2){\scriptsize$\lparent{h}\mid\Blue{\varepsilon}$}
  \node(3)(40,0){$+$}
  \drawedge(2,3){\scriptsize$\lparent{g}\mid\Blue{\varepsilon}$}
  \node(4)(60,0){$+$}
  \node(5)(80,0){$+$}
  \drawedge(4,5){\scriptsize$\rparent{g}\mid\Blue{\varepsilon}$}
  \node(6)(100,0){$-$}
  \drawloop(6){\scriptsize$\beta\mid\Blue{\varepsilon}$}
  \drawedge(5,6){\scriptsize$\rparent{h}\mid\Blue{-}$}
  \node(7)(120,0){$+$}
  \drawedge(6,7){\scriptsize$\lparent{h}\mid\Blue{-}$}
  \node(8)(140,0){$+$}
  \drawedge(7,8){\scriptsize$\lparent{f}\mid\Blue{\varepsilon}$}
  \node(9)(160,0){$+$}
  \node(10)(180,0){$+$}
  \drawloop(10){\scriptsize$\gamma\mid\Blue{\varepsilon}$}
  \drawedge(9,10){\scriptsize$\rparent{f}\mid\Blue{-}$}
  \node[Nmarks=f](11)(200,0){$+$}
  \drawedge(10,11){\scriptsize$\rparent{h}\mid\Blue{\varepsilon}$}
  \node[Nw=29,Nh=19,Nmr=5](f)(50,2){$\T_g$}
  \node[Nw=29,Nh=19,Nmr=5](g)(150,2){$\T_f$}
\end{gpicture}
\begin{gpicture}[name=Thadamard,ignore]
  \gasset{Nw=5,Nh=5,loopdiam=3,curvedepth=0,ilength=4,flength=4}
  \unitlength=.65mm
  \node[Nmarks=i](1)(0,0){$+$}
  \node(2)(20,0){$+$}
  \drawedge(1,2){\scriptsize$\lparent{h}\mid\Blue{\varepsilon}$}
  \drawloop[linecolor=red](2){\scriptsize\textcolor{red}{$x\mid\varepsilon$}}
  \node(3)(40,0){$+$}
  \drawloop[linecolor=red](3){\scriptsize\textcolor{red}{$x\mid\varepsilon$}}
  \drawedge(2,3){\scriptsize$\lparent{f}\mid\Blue{\varepsilon}$}
  \node[Nw=29,Nh=24,Nmr=5,NLangle=-90,NLdist=0](f)(50,4){$\T_f^{+g}$}
  \node(4)(60,0){$+$}
  \drawloop[linecolor=red](4){\scriptsize\textcolor{red}{$x\mid\varepsilon$}}
  \node(5)(80,0){$+$}
  \drawloop[linecolor=red](5){\scriptsize\textcolor{red}{$x\mid\varepsilon$}}
  \drawedge(4,5){\scriptsize$\rparent{f}\mid\Blue{-}$}
  \node(6)(100,0){$-$}
  \drawedge(5,6){\scriptsize$\rparent{h}\mid\Blue{\varepsilon}$}
  \drawloop(6){\scriptsize$y\mid\Blue{\varepsilon}$}
  \node(7)(120,0){$+$}
  \drawedge(6,7){\scriptsize$\lparent{h}\mid\Blue{\varepsilon}$}
  \drawloop[linecolor=red](7){\scriptsize\textcolor{red}{$z\mid\varepsilon$}}
  \node(8)(140,0){}
  \drawloop[linecolor=red](8){\scriptsize\textcolor{red}{$z\mid\varepsilon$}}
  \drawedge(7,8){\scriptsize$\lparent{g}\mid\Blue{\varepsilon}$}
  \node[Nw=29,Nh=24,Nmr=5,NLangle=-90,NLdist=0](g)(150,4){$\T_g^{+f}$}
  \node(9)(160,0){}
  \drawloop[linecolor=red](9){\scriptsize\textcolor{red}{$z\mid\varepsilon$}}
  \node(10)(180,0){$+$} 
  \drawloop[linecolor=red](10){\scriptsize\textcolor{red}{$z\mid\varepsilon$}}
  \drawedge(9,10){\scriptsize$\rparent{g}\mid\Blue{-}$}
  \node[Nmarks=f](11)(200,0){$+$}
  \drawedge(10,11){\scriptsize$\rparent{h}\mid\Blue{\varepsilon}$}
\end{gpicture}
\begin{gpicture}[name=Tstar,ignore]
  \gasset{Nw=5,Nh=5,loopdiam=4,curvedepth=0,ilength=4,flength=4}
  \unitlength=.75mm
  \node[Nmarks=i](1)(0,5){$+$}
  \node[fillcolor=red!20!white](2)(30,5){$+$}
  \node[Nmarks=f](3)(60,5){$+$}
  \drawedge[ELside=r](1,2){$\lparent{h}\mid\Blue{\varepsilon}$}
  \drawedge[ELside=r](2,3){$\rparent{h}\mid\Blue{\varepsilon}$}
  \node[Nw=27,Nh=15,Nmr=5](f)(30,20){$\T_f$}
  \node[fillcolor=red!20!white](4)(0,20){$+$}
  \node(*)(20,20){$+$}
  \drawedge(4,*){$\lparent{f}\mid\Blue{\varepsilon}$}
  \node[fillcolor=red!20!white](5)(60,20){$+$}
  \node(*)(40,20){$+$}
  \drawedge(*,5){$\rparent{f}\mid\Blue{-}$}
\end{gpicture}
\begin{gpicture}[name=Trstar,ignore]
  \gasset{Nw=5,Nh=5,loopdiam=4,curvedepth=0,ilength=4,flength=4}
  \unitlength=.75mm
  \node[Nmarks=i](1)(0,0){$+$}
  \node(2)(20,0){$+$}
  \node(3)(50,0){$-$}
  \node(4)(80,0){$+$}
  \node[Nmarks=f](5)(100,0){$+$}
  \drawedge(1,2){$\lparent{h}\mid\Blue{\varepsilon}$}
  \drawloop(2){$\alpha\mid\Blue{\varepsilon}$}
  \drawedge(2,3){$\rparent{h}\mid\Blue{\varepsilon}$}
  \drawedge(3,4){$\lparent{h}\mid\Blue{\varepsilon}$}
  \drawloop(4){$\alpha\mid\Blue{\varepsilon}$}
  \drawedge(4,5){$\rparent{h}\mid\Blue{\varepsilon}$}
  \node(6)(20,-20){$-$}
  \drawloop[loopangle=180](6){$\beta\mid\Blue{\varepsilon}$}
  \node(7)(80,-20){$-$}
  \drawloop[loopangle=0](7){$\beta\mid\Blue{\varepsilon}$}

  \node[Nw=27,Nh=15,Nmr=5](f)(50,-20){$\T_f$}
  \node(*)(40,-20){$+$}
  \drawedge(6,*){$\lparent{f}\mid\Blue{\varepsilon}$}
  \node(*)(60,-20){$+$}
  \drawedge(*,7){$\rparent{f}\mid\Blue{\varepsilon}$}

  \drawbpedge[ELpos=50,ELside=r,ELdist=1](3,-100,20,6,80,20){$\rparent{f}\mid\Blue{\varepsilon}$}
  \drawbpedge[ELpos=50,ELside=r,ELdist=1](7,100,20,3,-80,20){$\lparent{f}\mid\Blue{\varepsilon}$}
\end{gpicture}
\begin{gpicture}[name=1Pconst,ignore]
  \gasset{Nw=5,Nh=5,loopdiam=4,curvedepth=0,ilength=4,flength=4}
  \unitlength=.75mm
  \node[Nmarks=i](1)(0,0){}
  \node(2)(20,0){}
  \node[Nmarks=f](3)(40,0){}
  \drawedge(1,2){$\varepsilon \mid \lparent{h}$}
  \drawloop(2){$a \mid a$}
  \drawedge(2,3){$\varepsilon \mid \rparent{h}$}
\end{gpicture}
\begin{gpicture}[name=1PLv,ignore]
  \gasset{Nw=5,Nh=5,loopdiam=4,curvedepth=0,ilength=4,flength=4}
  \unitlength=.7mm
  \node[Nmarks=i](1)(0,0){$q^{h}_0$}
  \node(2)(20,0){$q^e_0$}
  \drawedge(1,2){\scriptsize$\varepsilon\mid\lparent{h}$}
  \node(3)(40,0){}
  \node(8)(40,-10){}
  \node(9)(40,10){}
  \drawedge(2,3){\scriptsize$b\mid b$}
  \drawedge[curvedepth=-4,ELpos=60,ELside=r,ELdist=0](2,8){\scriptsize$a\mid a$}
  \drawedge[curvedepth=4,ELpos=60,ELdist=0](2,9){\scriptsize$a\mid a$}
  \node[Nmarks=r](4)(70,0){}
  \node[Nmarks=r](5)(70,10){}
  \node[Nmarks=r](6)(70,-10){}
  \node[Nmarks=f](7)(90,0){$q^{h}_F$}
  \drawedge(4,7){\scriptsize$\varepsilon \mid\rparent{h}$}
  \drawedge[curvedepth=4,ELpos=40,ELdist=0](5,7){\scriptsize$\varepsilon \mid\rparent{h}$}
  \drawedge[curvedepth=-4,ELpos=40,ELside=r,ELdist=0](6,7){\scriptsize$\varepsilon \mid\rparent{h}$}
  \node[Nw=58,Nh=32,Nmr=5,NLangle=0,NLdist=10](f)(45,0){$\aut'_e$}
\end{gpicture}
\begin{gpicture}[name=1Pplus,ignore]
  \gasset{Nw=5,Nh=5,loopdiam=4,curvedepth=0,ilength=4,flength=4}
  \unitlength=.75mm
  \node[Nmarks=i](0)(0,0){$q^{h}_0$}
  \node[fillcolor=red!20!white](*)(20,0){}
  \drawedge[ELpos=50](0,*){$\varepsilon\mid\lparent{h}$}
  
  \node[Nmarks=f](3)(100,0){$q^h_F$}
  \node[fillcolor=blue!20!white](*)(80,0){}
  \drawedge[ELpos=50](*,3){$\varepsilon\mid\rparent{h}$}
  
  \node[fillcolor=red!20!white](1)(20,15){$q^f_0$}
  \node[fillcolor=blue!20!white](5)(80,15){$q^f_F$}
  \node[Nw=32,Nh=22,Nmr=5](f)(50,15){$\TParse_f$}
  \node(*)(40,21){}
  \drawedge[curvedepth=2,ELdist=0](1,*){}
  \node(*)(40,9){}
  \drawedge[ELside=r,curvedepth=-2,ELdist=0](1,*){}
  \node(*)(60,21){}
  \drawedge[curvedepth=2](*,5){}
  \node(*)(60,9){}
  \drawedge[ELside=r,curvedepth=-2](*,5){}
  
  \node[fillcolor=red!20!white](2)(20,-15){$q^g_0$}
  \node[fillcolor=blue!20!white](6)(80,-15){$q^g_F$}
  \node[Nw=32,Nh=22,Nmr=5](g)(50,-15){$\TParse_g$}
  \node(*)(40,-9){}
  \drawedge[curvedepth=2,ELdist=0](2,*){}
  \node(*)(40,-21){}
  \drawedge[ELside=r,curvedepth=-2,ELdist=0](2,*){}
  \node(*)(60,-9){}
  \drawedge[curvedepth=2](*,6){}
  \node(*)(60,-21){}
  \drawedge[ELside=r,curvedepth=-2](*,6){}
\end{gpicture}
\begin{gpicture}[name=1Pdot,ignore]
  \gasset{Nw=5,Nh=5,loopdiam=4,curvedepth=0,ilength=4,flength=4}
  \unitlength=.70mm
  \put(10,0){
  \node[Nw=30,Nh=24,Nmr=5](f)(30,0){$\TParse_f$}
  \node(2)(0,0){$q^f_0$}
  \node(*)(20,7){}
  \drawedge[curvedepth=2](2,*){}
  \node(*)(20,-7){}
  \drawedge[ELside=r,curvedepth=-2](2,*){}
  \node[fillcolor=red!20!white](3)(60,0){$q^f_F$}
  \node(*)(40,7){}
  \drawedge[curvedepth=2](*,3){}
  \node(*)(40,-7){}
  \drawedge[ELside=r,curvedepth=-2](*,3){}
  \node[Nmarks=i](1)(-20,0){$q^h_0$}
  \drawedge(1,2){$\varepsilon\mid\lparent{h}$}
  }

  \node[Nw=30,Nh=24,Nmr=5](g)(110,0){$\TParse_g$}
  \node[fillcolor=red!20!white](4)(80,0){$q^g_0$}
  \node(*)(100,7){}
  \drawedge[curvedepth=2](4,*){}
  \node(*)(100,-7){}
  \drawedge[ELside=r,curvedepth=-2](4,*){}
  \node(5)(140,0){$q^g_F$}
  \node(*)(120,7){}
  \drawedge[curvedepth=2](*,5){}
  \node(*)(120,-7){}
  \drawedge[ELside=r,curvedepth=-2](*,5){}
  \node[Nmarks=f](6)(160,0){$q^h_F$}
  \drawedge(5,6){$\varepsilon \mid\rparent{h}$}
\end{gpicture}
\begin{gpicture}[name=1Pstar,ignore]
  \gasset{Nw=5,Nh=5,loopdiam=4,curvedepth=0,ilength=4,flength=4}
  \unitlength=.75mm
  \node[Nmarks=i](1)(0,0){$q^h_0$}
  \node[fillcolor=red!20!white](2)(30,0){}
  \node[Nmarks=f](3)(60,0){$q^h_F$}
  \drawedge[ELside=r](1,2){$\varepsilon \mid\lparent{h}$}
  \drawedge[ELside=r](2,3){$\varepsilon \mid\rparent{h}$}

  \node[Nw=30,Nh=24,Nmr=5](f)(30,20){$\TParse_f$}
  \node[fillcolor=red!20!white](5)(0,20){$q^f_0$}
  \node(*)(20,27){}
  \drawedge[curvedepth=3](5,*){}
  \node(*)(20,13){}
  \drawedge[ELside=r,curvedepth=-3](5,*){}
  \node[fillcolor=red!20!white](6)(60,20){$q^f_F$}
  \node(*)(40,27){}
  \drawedge[curvedepth=3](*,6){}
  \node(*)(40,13){}
  \drawedge[ELside=r,curvedepth=-3](*,6){}
\end{gpicture}
\begin{gpicture}[name=Hadamard1FT1-alt,ignore]
  \gasset{Nw=5,Nh=5,loopdiam=4,curvedepth=0,ilength=4,flength=4}
  \unitlength=.75mm
  \node[Nmarks=i](1)(20,0){}
  \node(2)(40,0){}
  \node(3)(80,0){}
  \drawedge(1,2){$\varepsilon \mid\lparent{h}$}
  \node[Nmarks=f](4)(100,0){}
  \drawedge(3,4){$\varepsilon \mid\rparent{h}$}
  \node[Nw=47,Nh=15,Nmr=5](f)(60,0){$\TParse_f \otimes \TParse_g$}
\end{gpicture}
\begin{gpicture}[name=Tk-star,ignore]
  \gasset{Nw=4,Nh=4,loopdiam=3,curvedepth=0,ilength=4,flength=4}
  \unitlength=.75mm
  \node[Nmarks=i](1)(30,90){$+$}
  \node(2)(30,75){$+$}
  \drawloop[loopangle=180,ELpos=70,ELdist=0](2){\scriptsize$\alpha \mid \varepsilon$}
  \drawedge[ELside=r,ELpos=30](1,2){\scriptsize$\lparent{h}\mid\Blue{\varepsilon}$}
  
  \node(501)(30,60){$-$}
  \node(502)(30,45){$+$}
  \drawedge[ELside=r](2,501){\scriptsize$\lparent{f}^{1} \mid\Blue{\varepsilon}$}
  \drawedge[ELside=r](501,502){\scriptsize$\lparent{h} \mid\Blue{\varepsilon}$}
  
  \node[Nw=32,Nh=16,Nmr=4,NLangle=0,NLdist=0](f)(65,74){\scriptsize$\T^1_f$}
  \node(101)(55,70){$+$}
  \drawloop[loopangle=90](101){\scriptsize $x_1 \mid \varepsilon$}
  \node(102)(75,70){$+$}
  \drawloop[loopangle=90](102){\scriptsize $x_1 \mid \varepsilon$}
  \drawqbpedge[ELpos=60,ELdist=0](502,65,101,0){\scriptsize$\lparent{f}^{1} \mid\Blue{\varepsilon}$}

  \node(5)(92,62.5){$+$}
  \drawedge[curvedepth=2,ELpos=60,ELdist=0](102,5){\scriptsize$\rparent{f}^1\mid-$}
  \node(41)(107,62.5){$-$}
  \drawedge(5,41){\scriptsize$\rparent{h}\mid\Blue{\varepsilon}$}
  
  \node(6)(75,55){$-$}
  \drawedge[curvedepth=2,ELpos=40,ELdist=0](5,6){\scriptsize$\#_{e} \mid\Blue{\varepsilon}$}

  \node(7)(55,55){$-$}
  \drawedge[ELside=r](6,7){\scriptsize$\rparent{f}^1 \mid\Blue{\varepsilon}$}
  \drawloop[loopangle=150,ELpos=70,ELdist=1](7){\scriptsize$y_1 \mid \varepsilon$}

  \node[Nw=32,Nh=16,Nmr=4,NLangle=0,NLdist=0](f)(65,39){\scriptsize$\T^2_f$}
  \node(201)(55,35){$+$}
  \drawloop[loopangle=90](201){\scriptsize $x_2 \mid \varepsilon$}
  \node(202)(75,35){$+$}
  \drawloop[loopangle=90](202){\scriptsize $x_2 \mid \varepsilon$}
  \drawqbpedge[ELside=r,ELpos=25,ELdist=0](7,-140,201,180){\scriptsize$\lparent{f}^{2} \mid\Blue{\varepsilon}$}

  \node(8)(92,27.5){$+$}
  \drawedge[curvedepth=2,ELpos=60,ELdist=0](202,8){\scriptsize$\rparent{f}^2\mid-$}
  \node(42)(107,27.5){$-$}
  \drawedge(8,42){\scriptsize$\rparent{h}\mid\Blue{\varepsilon}$}
  
  \node(9)(75,20){$-$}
  \drawedge[curvedepth=2,ELpos=40,ELdist=0](8,9){\scriptsize$\#_{e} \mid\Blue{\varepsilon}$}

  \node(10)(55,20){$-$}
  \drawedge[ELside=r](9,10){\scriptsize$\rparent{f}^2 \mid\Blue{\varepsilon}$}
  \drawloop[loopangle=150,ELpos=70,ELdist=1](10){\scriptsize$y_2 \mid \varepsilon$}

  \node(11)(55,0){$-$}
  \drawloop[loopangle=150,ELpos=70,ELdist=1](11){\scriptsize$y_{k-1} \mid \varepsilon$}
  
  \node[Nw=32,Nh=16,Nmr=4,NLangle=0,NLdist=0](f)(65,-16){\scriptsize$\T^k_f$}
  \node(301)(55,-20){$+$}
  \drawloop[loopangle=90](301){\scriptsize $x_k \mid \varepsilon$}
  \node(302)(75,-20){$+$}
  \drawloop[loopangle=90](302){\scriptsize $x_k \mid \varepsilon$}
  \drawqbpedge[ELside=r,ELpos=25,ELdist=0](11,-140,301,180){\scriptsize$\lparent{f}^{k} \mid\Blue{\varepsilon}$}
  
  \node(12)(92,-27.5){$+$}
  \drawedge[curvedepth=2,ELpos=60,ELdist=0](302,12){\scriptsize$\rparent{f}^k\mid-$}
  \node(43)(107,-27.5){$-$}
  \drawedge(12,43){\scriptsize$\rparent{h}\mid\Blue{\varepsilon}$}
  
  \node(13)(75,-35){$-$}
  \drawedge[curvedepth=2,ELpos=40,ELdist=0](12,13){\scriptsize$\#_{e} \mid\Blue{\varepsilon}$}
  
  \node(14)(55,-35){$-$}
  \drawedge[ELside=r](13,14){\scriptsize$\rparent{f}^k \mid\Blue{\varepsilon}$}
  \drawloop[loopangle=135,ELpos=70,ELdist=0](14){\scriptsize$y_k \mid \varepsilon$}
  
  \node(503)(30,20){$-$}
  \drawqbpedge(14,0,503,-90){\scriptsize$\lparent{f}^{1} \mid\Blue{\varepsilon}$}
  \drawedge(503,502){\scriptsize$\beta \mid\Blue{\varepsilon}$}

  \node(20)(125,75){$+$}
  \node[Nmarks=f](21)(125,90){$+$}
  \drawedge[ELside=r,ELdist=0.5](20,21){\scriptsize$\rparent{h} \mid\Blue{\varepsilon}$}

  \drawedge[ELdist=0](41,20){\scriptsize$\rparent{f}^{1}\mid\Blue{\varepsilon}$}
  \drawedge[ELpos=30,ELdist=0](42,20){\scriptsize$\rparent{f}^{2} \mid\Blue{\varepsilon}$}
  \drawqbpedge(43,0,20,-90){\scriptsize$\rparent{f}^{k} \mid\Blue{\varepsilon}$}

  \node(401)(75,90){$-$}
  \drawqbpedge(2,45,401,180){\scriptsize$\rparent{h} \mid\Blue{\varepsilon}$}
  \drawqbpedge(401,0,20,135){\scriptsize$\alpha\mid\Blue{\varepsilon}$}
\end{gpicture}
\begin{gpicture}[name=Tk-r-star,ignore]
  \gasset{Nw=5,Nh=5,loopdiam=4,curvedepth=0,ilength=4,flength=4}
  \unitlength=.75mm
  \node[Nmarks=i](1)(-10,-5){$+$}
  \node(2)(-10,25){$+$}
  \drawloop[loopangle=180](2){\scriptsize $\alpha \mid\Blue{\varepsilon}$}
  \drawedge(1,2){\scriptsize$\lparent{h}\mid\Blue{\varepsilon}$}
  \node(3)(-10,55){$-$}
  \drawedge(2,3){\scriptsize$\rparent{h}\mid\Blue{\varepsilon}$}

  \node(31)(30,75){$-$}
  \drawqbpedge[ELpos=75](3,60,31,0){\scriptsize $\rparent{f}^{k} \mid\Blue{\varepsilon}$}
  \drawloop[loopangle=90](31){\scriptsize $y_{k} \mid\Blue{\varepsilon}$}

  \node(32)(30,25){$-$}
  \drawqbpedge[ELpos=15](3,-30,32,60){\scriptsize $\rparent{f}^{k-1} \mid\Blue{\varepsilon}$}
  \drawloop[loopangle=-90](32){\scriptsize $y_{k-1} \mid\Blue{\varepsilon}$}
  \node(33)(30,-45){$-$}
  \drawqbpedge[ELpos=15](3,105,33,0){\scriptsize $\rparent{f}^{1} \mid\Blue{\varepsilon}$}
  \drawloop[loopangle=90](33){\scriptsize $y_{1} \mid\Blue{\varepsilon}$}
  \node(34)(30,-5){$-$}
  \drawloop[loopangle=-90](34){\scriptsize $y_{k-2} \mid\Blue{\varepsilon}$}

  \node[Nw=37,Nh=25,Nmr=5,NLangle=0,NLdist=0](f)(70,80){\scriptsize $\T^k_f$}
  \node(101)(55,75){$+$}
  \drawloop[loopangle=75](101){\scriptsize $x_k \mid\Blue{\varepsilon}$}
  \node(102)(85,75){$+$}
  \drawloop[loopangle=105](102){\scriptsize $x_k \mid\Blue{\varepsilon}$}
  \drawedge(31,101){\scriptsize $\lparent{f}^{k} \mid\Blue{\varepsilon}$}

  \node(5)(110,75){$-$}
  \drawedge(102,5){\scriptsize $\rparent{f}^k\mid-$}
  \drawloop[loopangle=45](5){\scriptsize$z_{k} \mid\Blue{\varepsilon}$}
  
  \node(6)(110,55){$-$}
  \drawedge(5,6){\scriptsize $\lparent{f}^{k} \mid\Blue{\varepsilon}$}

  \node(7)(80,55){$+$}
  \drawedge(6,7){\scriptsize $\#_e \mid\Blue{\varepsilon}$}
  \drawqbpedge[ELside=r](7,0,32,100){\scriptsize $\lparent{f}^{k} \mid\Blue{\varepsilon}$}
  \node(71)(130,55){$+$}
  \drawedge(6,71){\scriptsize $\lparent{h} \mid\Blue{\varepsilon}$}

  \node[Nw=37,Nh=25,Nmr=5,NLangle=0,NLdist=0](f)(70,30){\scriptsize $\T^{k-1}_f$}
  \node(201)(55,25){$+$}
  \drawloop[loopangle=70](201){\scriptsize $x_{k-1} \mid\Blue{\varepsilon}$}
  \node(202)(85,25){$+$}
  \drawloop[loopangle=105](202){\scriptsize $x_{k-1} \mid\Blue{\varepsilon}$}
  \drawedge(32,201){\scriptsize $\lparent{f}^{k-1} \mid\Blue{\varepsilon}$}

  \node(8)(110,25){$+$}
  \drawedge(202,8){\scriptsize $\rparent{f}^{k-1}\mid-$}
  \drawloop[loopangle=45](8){\scriptsize$z_{k-1} \mid\Blue{\varepsilon}$}
  
  \node(9)(110,5){$-$}
  \drawedge(8,9){\scriptsize $\lparent{f}^{k-1} \mid\Blue{\varepsilon}$}

  \node(10)(80,5){$+$}
  \drawedge(9,10){\scriptsize $\#_e \mid\Blue{\varepsilon}$}
  \drawqbpedge[ELside=r,ELpos=25](10,-5,34,110){\scriptsize $\lparent{f}^{k-1} \mid\Blue{\varepsilon}$}
  \node(72)(130,5){$+$}
  \drawedge(9,72){\scriptsize $\lparent{h} \mid\Blue{\varepsilon}$}

  \node[Nw=37,Nh=25,Nmr=5,NLangle=0,NLdist=0](f)(70,-40){\scriptsize $\T^1_f$}
  \node(301)(55,-45){$+$}
  \drawloop[loopangle=75](301){\scriptsize $x_1 \mid\Blue{\varepsilon}$}
  \node(302)(85,-45){$+$}
  \drawloop[loopangle=105](302){\scriptsize $x_1 \mid\Blue{\varepsilon}$}
  \drawedge(33,301){\scriptsize $\lparent{f}^{1} \mid\Blue{\varepsilon}$}
  
  \node(12)(110,-45){$+$}
  \drawedge(302,12){\scriptsize $\rparent{f}^{1}\mid-$}
  \drawloop[loopangle=45](12){\scriptsize$z_{1} \mid\Blue{\varepsilon}$}
  
  \node(13)(110,-65){$-$}
  \drawedge(12,13){\scriptsize $\lparent{f}^1 \mid\Blue{\varepsilon}$}
  
  \node(14)(80,-65){$+$}
  \drawedge(13,14){\scriptsize $\#_e \mid\Blue{\varepsilon}$}
  \drawqbpedge[ELpos=25](14,10,31,70){\scriptsize $\lparent{f}^{1} \mid\Blue{\varepsilon}$}
  
  \node(73)(130,-65){$+$}
  \drawedge(13,73){\scriptsize $\lparent{h} \mid\Blue{\varepsilon}$}
  
  \node(20)(145,-25){$-$}
  \node(21)(160,10){$+$}
  \drawloop[loopangle=135](21){\scriptsize $\alpha \mid\Blue{\varepsilon}$}
  \drawedge[ELside=r](20,21){\scriptsize $\lparent{h}\mid\Blue{\varepsilon}$}

  \drawqbpedge[ELside=r,ELpos=60](71,0,20,90){\scriptsize $\lparent{f}^{k}\mid\Blue{\varepsilon}$}
  \drawedge[ELside=r](72,20){\scriptsize $\lparent{f}^{k-1} \mid\Blue{\varepsilon}$}
  \drawqbpedge[ELside=r](73,-140,20,90){\scriptsize $\lparent{f}^{1} \mid\Blue{\varepsilon}$}

  \node(22)(160,40){$-$}
  \node(23)(160,70){$+$}
  \node[Nmarks=f](24)(160,100){$+$}
  \drawedge[ELside=r](21,22){\scriptsize $\rparent{h} \mid\Blue{\varepsilon}$}
  \drawedge[ELside=r](22,23){\scriptsize $\kaparent{} \mid\Blue{\varepsilon}$}
  \drawedge[ELside=r](23,24){\scriptsize $\rparent{h} \mid\Blue{\varepsilon}$}

  \drawqbpedge(3,60,23,-30){\scriptsize $a \in \Sigma \cup \{\#_e\} \mid\Blue{\varepsilon}$}
\end{gpicture}
\begin{gpicture}[name=k-star,ignore]
  \gasset{Nw=5,Nh=5,loopdiam=4,curvedepth=0,ilength=4,flength=4}
  \unitlength=.75mm
  \node[Nmarks=i](1)(20,0){}
  \node(21)(40,15){}
  \node(22)(40,-15){}
  \node[Nmarks=r](5)(80,21){}
  \node[Nmarks=r](6)(80,09){}
  \node[Nmarks=r](7)(80,-21){}
  \node[Nmarks=r](9)(80,-09){}
  \drawqbpedge(1,90,21,0){\scriptsize$\varepsilon \mid\lparent{h}$}
  \drawqbpedge[ELside=r](1,-90,22,0){\scriptsize$\varepsilon \mid\lparent{h}$}
  \node[Nmarks=f](4)(100,0){}
  \drawqbpedge[ELside=l,ELpos=30,ELdist=0](5,0,4,-90){\scriptsize $\varepsilon \mid\rparent{h}$}
  \drawqbpedge[ELside=l,ELpos=35,ELdist=0](6,0,4,-60){\scriptsize $\varepsilon \mid\rparent{h}$}
  \drawqbpedge[ELside=r,ELpos=30,ELdist=0](7,0,4,90){\scriptsize $\varepsilon \mid\rparent{h}$}
  \drawqbpedge[ELside=r,ELpos=35,ELdist=0](9,0,4,60){\scriptsize $\varepsilon \mid\rparent{h}$}
  \node[Nw=47,Nh=25,Nmr=5](f)(60,15){$\TParse'_h$}
  \node[Nw=47,Nh=25,Nmr=5](f)(60,-15){$\TParse''_h$}
\end{gpicture}
%%%%%%%%%%%%%%%%%%%%%%%%%%%%%%%%%%%%%%%%%%%%%%%%%%%%%%%%
%% END FIGURES
%%%%%%%%%%%%%%%%%%%%%%%%%%%%%%%%%%%%%%%%%%%%%%%%%%%%%%%
One of the most celebrated results in theoretical computer science is the robust
characterization of languages using machines, expressions and logic.  For regular
languages, these three dimensions are given by finite state automata, regular expressions
as well as monadic second-order logic, while for aperiodic languages, the respective three
pillars are counter-free automata, star-free expressions and first-order logic.  The
B\"{u}chi-Elgot-Trakhtenbrot theorem was generalized by Engelfreit and Hoogeboom
\cite{EH01}, where regular transformations were defined using two-way transducers (2DFTs)
as well as by the MSO transductions of Courcelle~\cite{Cou94}.  The analogue of Kleene's
theorem for transformations was proposed by \cite{AlurFreilichRaghothaman14} 
 and \cite{DGK-lics18}, while \cite{DGK-lics21} proved the analogue of
Schützenberger's theorem \cite{Schutzenberger1975d} for transformations.   
In another related work, \cite{BR-DLT18,BR-Journal20} proposes a  translation from unambiguous two-way transducers to regular function expressions extending the Brzozowski and McCluskey algorithm. 
All these papers
propose declarative languages which are expressive enough to capture all regular
(respectively, aperiodic) transformations.

Our starting point in this paper is the combinator expressions presented in the
declarative language \RTE of \cite{DGK-lics18,DGK22}.  Like classical regular expressions, these
expressions provide a robust foundation for specifying transducer patterns in a
declarative manner, and can be widely used in practical applications.  An important
question left open in \cite{AlurFreilichRaghothaman14}, \cite{DGK-lics21}, \cite{DGK-lics18} 
 is the complexity of the procedure that builds the transducer from the
combinator expressions.  
Providing efficient constructions of finite state transducers equivalent to expressions
is a fundamental problem, and is often the first step of algorithmic applications, such 
as evaluation.
In this paper, we focus on this problem.  First, we recall the combinators from \cite{AlurFreilichRaghothaman14}, \cite{DGK-lics21} 
and \cite{DGK-lics18} in order of
increasing difficulty.  It is known that the combinators of 
\cite{AlurFreilichRaghothaman14} and \cite{DGK-lics18} are equivalent (they both characterize regular transformations), even though the notations differ slightly. Our notations are closer to \cite{DGK-lics18}.

We begin presenting the combinators.  
\begin{itemize}[nosep]
	\item[($i$)] The base combinator is  $\SimpleFun{e}{v}$
which maps any $w \in L(e)$ to the (constant) value $v$.  
\item[($ii$)] The sum combinator $f+g$  where $f$ or
$g$ is applied to the input $w$ depending on whether $w \in \dom{f}$ or $w \in \dom{g}$,
\item[($iii$)] The
Cauchy product $f\cdot g$  splits the input into two parts and outputs
the concatenation of the results obtained by applying $f$ on the first part and $g$ on the
second part, and 
\item[($iv$)] The star combinator $f^{\star}$ which splits the input into multiple parts
and outputs the concatenation of evaluating $f$ on the respective parts from left to right. 
\smallskip 

All these operators can be ambiguous and imply a relational semantics.
We denote the fragment of \RTE restricted to combinators ($i{-}iv$) as $\RTE[\Rat]$; this
fragment captures all rational transformations (those computed by a non-deterministic one way transducer, 1NFT).

\item[($v$)] The reverse concatenation combinator  $f\cdot_r g$ which works like $f\cdot g$ except that the
output is now the concatenation of the result of applying $g$ on the second part of the split of the input followed by $f$ on the first part,
\item[($vi$)] The reverse star combinator  $\rstar{f}$,
also like $f^{\star}$, splits the input into multiple parts but outputs the
concatenation of evaluating $f$ on each of the parts from right to left. 

\smallskip 
 The fragment of
\RTE with combinators ($i{-}vi$) is denoted $\RTE[\Rat,\reverse]$.

Finally, we have the most involved combinators, namely 
\item[($vii$)] The Hadamard product $f
\odot g$, which outputs the concatenation of applying $f$ on the input followed by
applying $g$ on the input, as long as the input is in the domain of both $f$ and $g$. With $\odot$ also, we have the fragment 
$\RTE[\Rat,\reverse,\odot]$. 
\item[($viii$)]The chained $k$-star $\kstar{k}{e}{f}$  
which factorizes an input $w$ into $u_1 u_2
\cdots u_n$, each $u_i$ belonging to the language of $e$,
and applies $f$ on all contiguous $k$
blocks $u_{i+1} \cdots u_{i+k}$, $0\leq i \leq n-k$ and finally concatenates the result.
\item[ ($ix$)] The reverse  chained $k$-star 
$\krstar{k}{e}{f}$ also factorizes an input $w$ into $u_1 u_2
\cdots u_n$, each $u_i$ belonging to the language of $e$,
and applies $f$ on all contiguous $k$ blocks from the right to the
left, and the result is concatenated. 
\smallskip 

 \RTE is the full class consisting of all
combinators ($i{-}ix$), and its unambiguous fragment is equivalent to regular transformations (those computed by a
deterministic two-way transducer, 2DFT).  Note that we consider chained $k$-star of
\cite{DGK-lics21} here, even though chained 2-star suffice for expressing all regular
transformations, since the idea of our construction is general.
\end{itemize}

\medskip
\noindent\textbf{Our Contributions}.  Given an \RTE $h$, we give an efficient
procedure to directly construct a reversible two-way transducer that computes $h$.  Even
though \cite{AlurFreilichRaghothaman14} and \cite{DGK-lics18} construct SST/2DFT from
combinator expressions, 
\begin{enumerate}[nosep]
	\item  they do not perform a complexity analysis, 
	\item  the
constructed machines in these papers rely on intermediate compositions, which incur an exponential blowup at each
step, making them unsuitable in practice for applications, 
\item translating the
SST/2DFT from these papers into 2RFT results in a further exponential blow up. The emphasis on 2RFT is due to the fact that, unlike
SST/2DFT, these machines incur only a polynomial complexity  for composition, making them the preferred machine model for handling 
modular specifications. 
\end{enumerate}

\medskip
We list our main contributions.
\begin{enumerate}[nosep]
  \item {\bf{A clean semantics}}.  As our first contribution, we propose a \emph{globally}
  unambiguous semantics ($\gu$-semantics for short) $\usem{h}$ for all $h\in\RTE$.  The
  previous papers \cite{AlurFreilichRaghothaman14}, \cite{DGK-lics18} proposed a different unambiguous semantics for the product combinators $\cdot, \star, \odot$, that we refer to here as 
  \emph{locally} unambiguous semantics ($\lu$-semantics for short)  to distinguish from our
  $\gu$-semantics (see Section~\ref{sec:comparison} for a comparison). We now illustrate why the $\gu$-semantics can be a preferred choice 
  rather than the $\lu$-semantics.  
  \begin{itemize}[leftmargin=2pt]
    \item[ $\bullet$] Consider the Cauchy product $f \cdot g \cdot h$ with
    $\dom{f}=\dom{g}=\{a,aa\}$, $\dom{h}=\{b,ab\}$, and $f(a)=c$, $f(aa)=cc$, $g(a)=d$,
    $g(aa)=dd$, $h(b)=e$, $h(ab)=ee$.  Consider $w=aaab$.  Under the $\lu$-semantics, $w$
    admits a unique factorization $(a\cdot a)\cdot ab$ for $(f \cdot g) \cdot h$ with $aa
    \in \dom{f.g}$, $ab \in \dom{h}$.  Also, $w$ admits a unique factorization
    $aa\cdot(a\cdot b)$ for $f \cdot (g \cdot h)$ with $aa \in \dom{f}$, $ab \in
    \dom{g\cdot h}$.  Note that $a\cdot(aab)$ does not qualify as a factorization for
    $f \cdot (g \cdot h)$ since $aab$ has more than one factorization for $g \cdot h$.
    However, $((f \cdot g) \cdot h)(w)=cdee \neq ccde=(f \cdot(g \cdot h))(w)$.      For the $\gu$-semantics, we define the unambiguous domain $\udom{h}$ of an expression $h$
    as the set of words which can be parsed unambiguously with respect to the
    \emph{global} expression $h$.  For the Cauchy product, $\udom{f \cdot g}$ is the set
    of words $w$ having a unique factorization $w=u \cdot v$ with $u \in \dom{f}$, $v \in
    \dom{g}$, and, in addition, $u \in \udom{f}$, $v \in \udom{g}$.
    For the example above, we have $aaab \notin \udom{(f\cdot g)\cdot
    h}=\udom{f\cdot(g\cdot h)}$.  Thus, the Cauchy product is associative under the $\gu$-semantics. Associativity is natural for the Cauchy product, and not having this is confusing for 
    a user working on specifications in the $\lu$-semantics.

    \item[$\bullet$]    
    The $\lu$-semantics of the Cauchy product used in 
    previous papers \cite{AlurFreilichRaghothaman14,DGK-lics18}, allows to get symmetric
    differences of domains, hence also complements.  Consider two regular expressions
    $e_{1},e_{2}$ over alphabet $\Sigma$ and a marker $\$\notin\Sigma$. For $i=1,2$ let 
    $h_{i}=\SimpleFun{(\varepsilon+\$e_{i})}{\varepsilon}$ with domain 
    $\varepsilon+\$L(e_{i})$. The domain of $h_{1}\cdot h_{2}$ is
    $\varepsilon+\$(L(e_{1})\Delta L(e_{2}))+\$L(e_{1})\$L(e_{2})$, 
    where $\Delta$ denotes symmetric difference. 
    If we intersect with $\$\Sigma^{\star}$, we get the symmetric difference 
    $\$(L(e_{1})\Delta L(e_{2}))$. If $L(e_{2})=\Sigma^*$,  we obtain the 
    complement $\$(\Sigma^{\star}\setminus L(e_{1}))$.  This explains that an exponential
    blow-up is unavoidable when dealing with the $\lu$-semantics.  Note that for a
    standalone expression $h_1 \cdot h_2$, the $\lu$, $\gu$ semantics agree, however, things
    are different when one deals with a nested expression containing $h_1\cdot h_2$.
    To illustrate this,  consider the expression $h=(f_1 \cdot f_2)+f_3$. 
    On an input $w$, $h(w)$ is $f_3(w)$ if $w \notin \dom{f_1 \cdot f_2}$.  
    Next,  to check if $w \in \dom{f_1\cdot f_2}$, one has to
    verify if $w$ has an unambiguous split as $u_1u_2$ with $u_1 \in \dom{f_1}$, $u_2 \in
    \dom{f_2}$.  This requires us to complement the set of all words having more than one
    split $w=u_1u_2$.  Thus, evaluating $h$ on $w$ requires two nested complements. 
    In general, evaluating an expression in $\lu$-semantics may require arbitrary nested
    complementation.
    This is required due to the ``local unambiguity check'' at each local nesting level in
    the $\lu$-semantics, accentuating the exponential blow up problem.  In contrast, under
    the $\gu$-semantics, the unambiguity requirement is at a global level.
  \end{itemize} 
  
  To summarize, the $\lu$-semantics of \cite{AlurFreilichRaghothaman14}, \cite{DGK-lics18} may be
  difficult to comprehend for a user specifying with \RTE given that 
  $\cdot$ is non-associative and the same kind of unexpected behaviours arises with 
  iterations. It allows more
  inputs to be in the domain, but it may not be obvious to check if a given input is in
  the domain or to predict which output will be produced.  Our $\gu$-semantics on the other
  hand, is more intuitive, and hence easier to use.  Another important point to note
  is that the $\gu$-semantics does not restrict the expressiveness of \RTEs.  Although
  \cite{AlurFreilichRaghothaman14}, \cite{DGK-lics18} proposed the $\lu$-semantics and
  showed the equivalence between \RTEs and SST/2DFT, it can be seen from their equivalence
  proofs that the \RTE $h$ constructed there from a SST/2DFT $\tra$ satisfies
  $\udom{h}=\dom{\tra}$ and $\usem{h}=\sem{\tra}$.
  \medskip
  
  \item \textbf{Efficient Construction of 2RFT}.  The second contribution is the efficient
  construction of 2RFTs from \RTE specifications.  Given $h \in \RTE$, and a word $w$, we
  first ``parse'' $w$ according to $h$ using a 1NFT $\TParse_h$ called the parser (see Page~\pageref{remark:parser-1v2} for a discussion why 
  $\TParse_h$ is a one way machine).     The parsing 
  relation of $h$ \wrt a word $w$, $\Parse_h(w)$,  can be seen as a traversal of possible parse trees of $w$ \wrt $h$. 
  Examples are given in Sections~\ref{sec:rat-functions} and~\ref{sec:had}.
  Each possible parsing in $\Parse_h(w)$ introduces pairs of parentheses $\parent{i}{~}$ to
  bracket the factor of $w$ matching subexpressions $h_i$ of $h$.
  To illustrate the need for such a parsing, consider the expression 
  $h=f_1\cdot (f_2 \odot f_3)$. To evaluate $h$ on some input $w$, one must guess the position in $w$ where the scope of $\dom{f_1}$ ends, and $\dom{f_2 \odot f_3}$ begins.  Note that 
  we must apply $f_2, f_3$ on the same suffix of $w$, necessitating a two way behaviour. After applying $f_2$ on a suffix $v$ of $w=u\cdot v$,  one must come back to the beginning of $v$ to apply $f_3$. It is unclear how one can do this without inserting some markers, especially if the decomposition is ambiguous.   
 
  If $w$ does not have an unambiguous parsing \wrt $h$, then $\TParse_h$ will
  non-deterministically produce the parsings of $w$.
  For each $\alpha \in
  \sem{\TParse_h}(w)$, the projection of $\alpha$ to $\Sigma$ is $w$.  Next, we construct
  an evaluator which is a two-way reversible transducer (2RFT) $\T_h$ which takes words in
  $\TParse_h(w)$ as input, and produces words in $\rsem{h}(w)$, where $\rsem{h}$ denotes
  the relational semantics of $h$.  That is, $\rsem{h}=\sem{\T_h}\circ\Parse_h$.
 
\begin{figure}[!h]
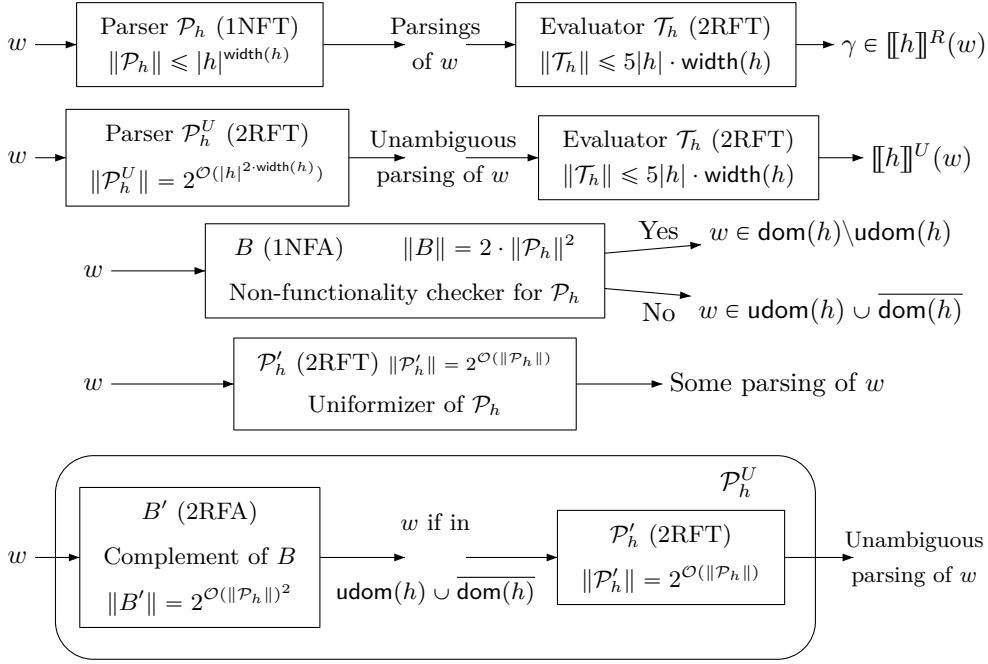

  \centering
  \gusepicture{Overview}
  \caption{The topmost figure shows the parser 1NFT $\TParse_h$ which, on an input $w$,
  produces a parsing in $P_h(w)$.  This is taken as input by the 2RFT $\T_h$, and producing
  a possible output $\gamma$ in $h(w)$.  This denotes the relational semantics of
  $h$, where $h(w)$ is not unique.
  The second figure from the top shows the 2RFT $\TParse_h^U$ which works only on words
  $w$ having a unique parsing, and produces this unique parsing $P_h(w)$.  This is then
  taken as input by the 2RFT $\T_h$, producing the output $h(w)$.  Here, $w \in \udom{h}$,
  and $\T_h$ produces the output $h(w)$.  
  The third figure describes the automaton $B$ used for checking the functionality of the
  1NFT $\TParse_h$: $B$ accepts all words $w$ which are not in $\udom{h}$.  Note that the
  complement of $B$, a reversible automaton $B'$ is used in $\TParse_h^U$.
  The fourth figure shows the uniformization of $\TParse_h$, given by the 2RFT
  $\TParse_h'$.  This machine outputs \emph{some} parsing of $w$.
  The last figure shows $\TParse_h^U$.  This first uses the reversible automaton $B'$ to
  filter words not in $\udom{h}$.  All words accepted by $B'$ either lie in $\udom{h}$, or
  outside $\dom{h}$.  For words $w \in \udom{h}$, there is a unique parsing, and
  $\TParse'_h$ produces this unique parsing of $w$.
  }
  \label{fig:overview}
\end{figure}
  
  \noindent{\textbf{2RFT for the globally unambiguous semantics}}.  %
  Note that the 1NFT $\TParse_h$ does not check whether $w\in\udom{h}$.  To obtain the $\gu$-semantics
  $\usem{h}$, we have to restrict to words in $\udom{h}$.  This is achieved by proving
  that $\udom{h}$ coincides with words $w\in\dom{\Parse_h}$ such that $|\Parse_h(w)|=1$.
  The unambiguity of the domain is checked by constructing an automaton that accepts the
  set of words having at most one parsing \wrt $h$.  We construct a reversible
  automaton $B'$ of size $2^{\mathcal{O}(\norm{\TParse_h}^2)}$ to do this, where $\norm{\TParse_h}$ denotes the size of $\TParse_h$.
  Next, we \emph{uniformize} $\TParse_h$ to obtain a 2RFT $\TParse'_h$ with the same domain as
  $\TParse_{h}$ and such that $\sem{\TParse'_{h}}\subseteq\rsem{\TParse_{h}}$: when
  running on $u\in\dom{\TParse_{h}}$, the 2RFT $\TParse'_{h}$ produces \emph{some output} $v$ such that
  $(u,v) \in \rsem{\TParse_h}$. The size of $\TParse'_{h}$ is
  $2^{\mathcal{O}(\norm{\TParse_h})}$.
  Then, we construct a machine $\TParse^{U}_{h}$ that first runs the automaton $B'$
  without producing anything and then runs $\TParse'_h$ if $B'$ has accepted.  This
  transducer is reversible and computes the parsing relation on words belonging to
  $\udom{h}$.  Its size is $2^{\mathcal{O}(\norm{\TParse_h}^2)}$.
  Finally, the composition of $\T_h$ and $\TParse^{U}_h$ gives a 2RFT
  $\T_{h}^{U}$ which realizes $\usem{h}$. 
\end{enumerate} 
Figure \ref{fig:overview} shows all the components used in our construction, and their
interconnection: the parser $\TParse_h$ (a 1NFT), the uniformizer $\TParse'_h$ of the
parser $\TParse_h$ (a 2RFT), the functionality checker of the parser (NFA $B$ and RFA
$B'$) and the final transducer $\T_h$ (a 2RFT).  

\smallskip
  
We now discuss the sizes of the 2RFT obtained for various \RTE fragments.    	
\begin{itemize}[nosep]
  \item $\RTE[\Rat]$.  In this case, the parser $\TParse_h$ and the evaluator $\T_h$ have sizes $\norm{\TParse_h}, \norm{\T_h} \leq |h|$.  
  Thus, the composed machine obtained from $\TParse^{U}_h$ and $\T_h$ has size $2^{\mathcal{O}(|h|^2)}$.  
  Notice that a one-way deterministic automaton accepting the domain of $h$ would already 
  be of exponential size.
  Indeed we can use a standard construction producing a rational transducer from
  an expression $h\in\RTE[\Rat]$, but it would realize the relational semantics of $h$ and not its unambiguous semantics.
  
  \item $\RTE[\Rat,\reverse]$.  We have the same complexity here as for $\RTE[\Rat]$.  
  Even if we add on to this fragment, the useful functions
         $\duplicate$ and $\revfn$ which respectively duplicates and reverses the input, the complexity is  
  still the same.
  
  \item $\RTE[\Rat,\reverse,\odot]$.  Unlike the $\Rat,\reverse$ combinators, $\odot$
  requires to read the input twice.  It is noteworthy that our parser $\TParse_h$ is still
  a 1NFT. However, in this case, its size is $\norm{\TParse_h} \leq |h|^{\w{h}}$ while
  $\norm{\T_h} \leq 5|h|$.  The \emph{width} of an \RTE $h$ is intuitively the maximal
  number of times a position in $w$ needs to be read to produce the output.
  Even though our parser is still a 1NFT, its size is affected by the width.
Notice that the domain of a Hadamard product $h=f\odot g$ is the intersection of the
domains of its arguments.  Moreover, the parser $\TParse_{h}$ may be used to recognize 
  the domain of $h$. This gives an exponential lower bound on the size of any possible 
  parser for expressions in $\RTE[\Rat,\odot]$ (see Proposition~\ref{prop:had-lb}).

  For expressions $h\in\RTE[\Rat,\reverse,\odot]$, the size of the final 2RFT $\T_{h}^{U}$
  is $2^{\mathcal{O}(|h|^{2\cdot\w{h}})}$.  Note that the fragments $\RTE[\Rat,\reverse]$,
  along with $\duplicate,\revfn$ have $\w{h}=1$.

  \item For full \RTE, the parser $\TParse_h$ is still
  a 1NFT and the bounds are the same as $\RTE[\Rat,\reverse,\odot]$, except now, we have $\norm{\T_h} \leq 5|h|\w{h}$. 
\end{itemize}

\medskip\noindent{\bf{Related Work}}.  %
A paper which has looked at the evaluation of transducer expressions is \cite{drex}.
Here, the authors investigate the complexity of evaluation of a
DReX program on a given input word.  A DReX program is a combinator expression
\cite{AlurFreilichRaghothaman14} and works with the $\lu$-semantics.  A
major difference between \cite{drex} and our paper is that \cite{drex} does not construct
a machine equivalent to a DReX program, citing complexity considerations and the
difficulty in coming up with an automaton for the $\lu$-semantics.  Instead, 
\cite{drex} directly solve the evaluation problem  
using dynamic programming.

To the best of our knowledge, our paper is the first one to efficiently construct a 2RFT
from an \RTE.
This 2RFT may be used to solve algorithmic problems on transformations
specified by transducer expressions.
One such problem is indeed the evaluation of any number of input words $w$; we can simply
run our constructed 2RFT on $w$ in time linear in $|w|$.  Note that \cite{drex} also evaluates with
the same linear bound, under what they call the ``consistent'' semantics, a restriction of the
$\lu$-semantics.  
The consistent semantics is also more restrictive than our
$\gu$-semantics.  
To mention an instance, the ${\tt{combine}}(f,g)$
combinator in \cite{drex} is analogous to the Hadamard product $f\odot g$, with the added restriction that 
that $\dom {f}=\dom{g}$.  Our $\gu$-semantics for $f \odot g$ only requires that the input is in
$\dom{f}\cap\dom{g}$.

\medskip\noindent{\bf{Structure of the paper}}. %
Section~\ref{sec:defs} introduces our models of automata and transducers while
Section~\ref{sec:transducer-expressions} defines the Regular Transducers Expressions, as
well as the relational semantics and the unambiguous semantics considered throughout the
paper.  It also states our results.  The following sections are devoted to the
constructions of transducers and the proofs of our main results, in an incremental
fashion: Section~\ref{sec:rat-functions} treats the case of Rational relations,
Section~\ref{sec:rat-rev-functions} handles some simple extensions, and
Sections~\ref{sec:had} and~\ref{sec:kstar} treats the Hadamard product and $k$-star
operators respectively.  Finally, Section~\ref{sec:unamb} shows how to compute a 
reversible transducer for the unambiguous semantics.

\section{Automata and Transducers}\label{sec:defs}

\myparagraph{Automata}
Let $\alp$ be an \emph{alphabet}, i.e., a finite set of letters. 
A word $u$ over $\alp$ is a possibly empty sequence of letters.
The set of words is denoted $\alp^*$, with $\epsilon$ denoting the empty word. Given an alphabet $\alp$, we denote by $\vdashv{\alp}$ the set $\alp\uplus\{\vdash,\dashv\}$, where $\vdash$ and $\dashv$ are two fresh symbols called the left and right endmarkers.
A \emph{two-way finite state automaton} (2NFA) is a tuple $\aut = (\alp, Q, q_I, F,
\Delta)$, where $\alp$ is a finite alphabet,
$Q$ is a finite set of states partitioned into the set of forward states $Q^{+}$ and the set of backward states $Q^{-}$,
$q_I \in Q^{+}$ is the initial state,
$F \subseteq Q^{+}$ is the set of final states,
$\Delta \subseteq Q \times \vdashv{\alp} \times Q$ is the state transition relation.
By convention, $q_I$ and $q_F\in F$ are the only forward states verifying $(q_I,\vdash,q) \in \Delta$
and $(q,\dashv,q_F) \in \Delta$ for some $q \in Q$.
However, for any backward state $p^{-} \in Q^{-}$, $\Delta$ might contain transitions $(p^{-},\vdash,q)$
and $(q,\dashv,p^{-})$, for some $q \in Q$.

Before defining the semantics of our two-way automata, 
let us remark that we choose one of several equivalently expressive semantics of two-way. The particularity of the one we chose, which is the one in~\cite{DFJL17}, is that $(1)$ the reading head is put between positions rather than on, and $(2)$ the set of states is divided into $+$ states and $-$ states. The advantage of this semantics is that the sign of a state defines what position the head reads both before and after this state in a valid run.
A $+$ state (resp.\ $-$ state) reads the position to its right (resp.\ to its left) and the previous position read was on its left (resp.\ on its right).
Intuitively, in a transition $(p,a,q)$ both states move the reading head half a position, either to the right for $+$ states or to the left for $-$ states. Hence if $p$ and $q$ are of different signs, the reading head does not move, but the position read will be different.

We now formally define the semantics.
A configuration $u.p.u'$ of $\aut$ is composed of two words $u,u'$ such that $uu'\in{\vdash}\Sigma^*{\dashv}$ and a state $p \in Q$.
The configuration $u.p.u'$ admits a set of successor configurations, defined as follows.
If $p \in Q^+$, the input head currently reads the first letter of the suffix $u' = a'v'$.
The successor of $u.p.u'$ after a transition $(p,a',q)\in \Delta$ is either $ua'.q.v'$ if $q\in Q^+$, or 
$u.q.u'$ if $q\in Q^-$.
Conversely, if $p \in Q^-$, the input head currently reads the last letter of the prefix $u = v a$.
The successor of $u.p.u'$ after $(p,a,q)\in \Delta$ is $u.q.u'$ if $q\in Q^+$, or $v.q.au'$ if $q\in Q^-$.
A \emph{run} of $\aut$ on a word $u\in{\vdash}\Sigma^*{\dashv}$ is a sequence of successive configurations
$\rho = u_0.q_0.u_0',\ldots, u_m.q_m.u_m'$ such that for every $0 \leq i \leq m$, $u_iu_{i}' = u$.
The run $\rho$ is called \emph{initial} if it starts in configuration $q_I.u$,
\emph{final} if it ends in configuration $u.q$ with $q\in F$, \emph{accepting} if it is
both initial and final.
The language $\lang_{\aut}$ \emph{recognized} by $\aut$ is the set of words $u\in\alp^*$
such that ${\vdash} u {\dashv}$ admits an accepting run.
The automaton $\aut$ is called 
\begin{itemize}[nosep]
\item
a \emph{one-way finite state automaton} (1NFA) if the set $Q^-=\emptyset$,
\item
\emph{deterministic} (2DFA) if for all $(p,a) \in Q \times \vdashv{\alp}$, there is at most one $q \in Q$ verifying $(p,a,q) \in \Delta$,
\item \emph{co-deterministic} if for all $(q,a) \in Q \times \vdashv{\alp}$, there is at most one $p \in Q$ verifying $(p,a,q) \in \Delta$ and $F=\{q_F\}$.
\item \emph{reversible} (2RFA) if it is both deterministic and co-deterministic.
\end{itemize}

\myparagraph{Example}
Let us consider the language $\mathcal{L}_{a} \subseteq \{a,b\}^*$ composed of the  words
that contains at least one $a$ symbol. This language is recognized by the deterministic one-way automaton $\aut_1$ and
represented in Figure \ref{A1}, and by the reversible two-way automaton $\aut_2$, represented in Figure~\ref{A2}.
Note that $\aut_1$ is not co-deterministic in state $1$ reading an $a$. In fact, this language is not recognizable by a one-way reversible automaton because reading an $a$ from state $1$ cannot lead to state $0$, and adding a new state simply moves the problem forward.
The reversible two-way transducer solves this problem by using the left endmarker.

\begin{figure}
\begin{subfigure}{.5\textwidth}
\begin{tikzpicture}[scale=0.7]
\node[circle,draw,inner sep=3,initial] (init) at (-2,0) {$q_I$} ;
\node[circle,draw,inner sep=3] (n0) at (0,0) {$0$} ;
\node[circle,draw,inner sep=3] (n1) at (2,0) {$1$} ;
\node[circle,draw,inner sep=3,accepting] (final) at (4,0) {$q_F$} ;
\draw[>=stealth]  (init) -> node[midway,above] {$\vdash$} (n0);
\draw[>=stealth]  (n1) -> node[midway,above] {$\dashv$}  (final);
\draw (n0) edge  node[midway,above] {$a$} (n1);
\draw (n0) edge[loop above] node[midway, above]{$b$} (n0);
\draw (n1) edge[loop above] node[midway, above]{$a,b$} (n1);
\end{tikzpicture}
 \caption{A 1DFA $\aut_1$}
\label{A1}
 \end{subfigure}
\begin{subfigure}{.5\textwidth}
\begin{tikzpicture}[scale=0.7]
\node[circle,draw,inner sep=3,initial] (init) at (-2,0) {$q_I$} ;
\node[circle,draw,inner sep=3] (n0) at (0,0) {$+$} ;
\node[circle,draw,inner sep=3] (n1) at (2,0) {$-$} ;
\node[circle,draw,inner sep=3] (n2) at (4,0) {$+$};
\node[circle,draw,inner sep=3,accepting] (final) at (6,0) {$q_F$} ;
\draw[>=stealth]  (init) -> node[midway,above] {$\vdash$} (n0);
\draw[>=stealth]  (n2) -> node[midway,above] {$\dashv$}  (final);
\draw (n0) edge  node[midway,above] {$a$} (n1);
\draw (n1) edge  node[midway,above] {$\vdash$} (n2);
\draw (n0) edge[loop above] node[midway, above]{$b$} (n0);
\draw (n1) edge[loop above] node[midway, above]{$b$} (n1);
\draw (n2) edge[loop above] node[midway, above]{$a,b$} (n2);
\end{tikzpicture}
 \caption{A 2RFT $\aut_2$}
\label{A2}
 \end{subfigure}
\caption{Two automata recognizing the same language $\alp^*a\alp^*$.}
\label{fig::2aut}
\end{figure}

\myparagraph{Transducers}
A \emph{two-way finite state transducer} (2NFT) is a tuple $\tra = (\alp,\alpo, Q, q_I,
F, \Delta, \mu)$, where $\alpo$ is a finite alphabet;
$\aut_{\tra} = (\alp, Q, q_I, F, \Delta)$ is a 2NFA, called the \emph{underlying automaton} of $\tra$;
and $\mu : \Delta \rightarrow \alpo^*$ is the output function.
A run of $\tra$ is a run of its underlying automaton, and the language $\lang_{\tra}$ recognized by $\tra$
is the language $\lang_{\aut_{\tra}} \subseteq \alp^*$ recognized by its underlying automaton.
Given a run $\rho$ of $\tra$, we set $\mu(\rho) \in \alpo^*$ as the concatenation of the images by $\mu$ of the transitions of $\tra$
occurring along $\rho$.
The transduction $\trans_{\tra} \subseteq \alp^* \times \alpo^*$ defined by $\tra$ is the set of pairs $(u,v)$
such that $u \in \lang_{\tra}$ and $\mu(\rho) = v$ for an accepting run $\rho$ of
$\aut_{\tra}$ on ${\vdash} u {\dashv}$.
Two transducers are called \emph{equivalent} if they define the same transduction.
A transducer $\tra$ is respectively called one-way (1NFT), deterministic (2DFT),
co-deterministic or reversible (2RFT), if its underlying automaton has the corresponding
property.

Note that while a generic transducer defines a relation over words, a deterministic, co-deterministic or a reversible one defines a (partial) function from the input words to the output words since any input word has at most one accepting run, and hence at most one image. Extracting a maximal function from a relation is called a \emph{uniformization} of a relation. 
Formally, given a relation on words $R\subseteq \Sigma^*\times \Gamma^*$, 
a \emph{uniformization} of $R$ is a function $f$ such that:
\begin{itemize}[nosep]
  \item $\dom{f}=\dom{R}=\{u\in \Sigma^*\mid \exists v\in\Gamma^{*}, (u,v)\in R\}$
  \item $\forall u\in \dom{f}, (u,f(u))\in R$.
\end{itemize}
Intuitively, a uniformization chooses, for each left component of $R$, a unique right component.
If $R$ is already a function, then it is its only possible uniformization.

Reversible transducers can be composed easily, and the composition can be done with a single machine having a polynomial number of states.
Hence, when dealing with two-way machines, it is always beneficial to handle reversible machines.
To this end, we specialize results from~\cite{DFJL17} and~\cite{KW97} respectively.
\begin{lemma}\label{lem-1w to Rev}
  Let $\T$ be a 1NFT with $n$ states.
  Then we can construct a reversible 2RFT $\T'$ such that $\sem{\T'}$ is a
  uniformization of $\sem{\T}$ and $\T'$ has at most $144n^22^{2n}$ states.
  \label{lem:unifsize}
\end{lemma}

\begin{proof}
  Let $\T$ be a 1NFT and $n$ its number of states.
  We write $\T'$ as the composition $D\circ C$ where $C$ is a co-deterministic one-way transducer and $D$ is a deterministic one.
  The co-deterministic transducer $C$ is a classical powerset construction that computes and adds to the input the set of co-reachable states of $\T$.
  Its number of states is at most $2^n$.
  The set of states of the deterministic transducer $D$ is the same as $\T$, and at each
  step, if it is in a state $q$, it uses the information given by $C$ to select a successor
  of $q$ that is also co-reachable.
  It can be made deterministic by using an arbitrary global order on the set of states of $\T$. Its number of states is then $n$.
  
  We conclude on the size of $\T'$ using two theorems from~\cite{DFJL17}, stating that $C$ can be made into a reversible two-way $C'$ with $4m^2$ states with $m$ the number of states of $C$ (Theorem 2) and that $D$ can be made into a reversible $D'$ with $36n^2$ states (Theorem 3).
  Finally, $\T'$ is defined as the composition $D'\circ C'$, whose number of states is at most $36n^2\cdot 4(2^n)^2=144n^22^{2n}$.
\end{proof}

We will also need a more specific result for computing the complement of an automaton. We rely on Proposition~4 of~\cite{KW97}.
\begin{lemma}\label{lem-Aut to Rev}
  Let $A$ be a 1NFT with $n$ states.
  Then we can compute a 2RFT $B$ such that $L(B)$ is the complement of $L(A)$, and $B$ has at most $2^{n+1}+6$ states.
\end{lemma}
\begin{proof}
  The proof is straightforward. First, we transform $A$ into a deterministic automaton $C$ by doing a classical powerset construction. 
  Then $C$ has $2^n$ states. By inverting the accepting states, we obtain $C'$ which is deterministic and recognizes the complement of $L(C)=L(A)$.
  We conclude by using Proposition~4 of~\cite{KW97}, which states that from a
  deterministic automaton $C'$, we can construct a 2RFT $B$ by adding $3$
  states to $C'$ and doubling its number of states.
  The resulting automaton $B$ is then reversible and its number of states is $2(2^n+3)=2^{n+1}+6$.
\end{proof}

\section{Transducer expressions and their semantics}\label{sec:transducer-expressions}
In this section, we formally define \RTE, and then propose 
the most natural relational semantics. Then we 
define the unambiguous domain of a relation, 
and propose 
our global unambiguous semantics (called unambiguous semantics from here on) 
as a restriction 
of the relational semantics to the unambiguous domain. As already mentioned 
in the introduction, this semantics refines the   unambiguous semantics which has been proposed in earlier papers. Finally, we state the main results of the paper, and an overview of our results. 

\myparagraph{Regular transducer expressions (\RTEs)} Let $\Sigma$ be the input alphabet and $\Gamma$ be the output alphabet.
For the combinator expressions, we use the following syntax:
$$
h ::= \SimpleFun{e}{v} \mid h+h \mid h\cdot h\mid h\cdot_r h \mid h^{\star} \mid \rstar{h} \mid h\odot h 
\mid \kstar{k}{e}{h} \mid \krstar{k}{e}{h}
$$
where $e$ is a regular expression over $\Sigma$, $v\in\Gamma^{\star}$ and $k\geq 1$.

The semantics of the basic expression $\SimpleFun{e}{v}$ is the partial function with
(constant) value $v$ and domain $L(e)$, the regular language denoted by $e$.  For
instance, the semantics of $\SimpleFun{\emptyset}{v}$ is the partial function with empty
domain and the semantics of $\SimpleFun{\Sigma^{\star}}{v}$ is the \emph{total} constant
function with value $v$.  We use $\bot$ and $v$ as macros to respectively
denote $\SimpleFun{\emptyset}{\epsilon}$ and $\SimpleFun{\Sigma^*}{v}$.

Since our goal is to construct ``small'' transducers from \RTEs, we have to define
formally the \emph{size} of expressions.  We use the classical syntax for regular
expressions over $\Sigma$:
$$
e ::= \emptyset \mid \varepsilon \mid a \mid e+e \mid e\cdot e \mid e^{\star}
$$ 
where $a\in\Sigma$. 
We also define inductively the number of literals occurring in a regular expression $e$, 
denoted by $\mathsf{nl}(e)$: $\mathsf{nl}(\varepsilon)=\mathsf{nl}(\emptyset)=0$, 
$\mathsf{nl}(a)=1$ for $a\in\Sigma$, $\mathsf{nl}(e_1+e_2)=\mathsf{nl}(e_1\cdot 
e_2)=\mathsf{nl}(e_1)+\mathsf{nl}(e_2)$ and $\mathsf{nl}(e_1^{\star})=\mathsf{nl}(e_1)$.
Notice that we have $\mathsf{nl}(e)\leq|e|$ for all regular expressions $e$, where $|e|$ denotes the standard size of expressions. Actually, if 
$e$ is not a single letter $a\in\Sigma$, we even have $1+\mathsf{nl}(e)\leq|e|$.

Now, we define the size $|h|$ of a regular transducer expression $h$.  
For the base case, we define $|\SimpleFun{e}{v}| = 1+(1+\mathsf{nl}(e))+\max(1,|v|)$.  
Note that when $v=\varepsilon$ it still contributes 1 to the size of
$\SimpleFun{e}{\varepsilon}$ since it appears as a symbol.
Also, we have chosen that the regular expression $e$ contributes to $1+\mathsf{nl}(e)$ in
this size.  This is because the number of states of the Glushkov automaton associated with
$e$ (which will be used in our construction) is $1+\mathsf{nl}(e)$.  As discussed above,
unless $e$ is a single letter from $\Sigma$, we have $1+\mathsf{nl}(e)\leq|e|$ (and
otherwise $\mathsf{nl}(e)=1=|e|$).
For the inductive cases, we let $|f+g| =|f\cdot g| =|f\cdot_{r} g| =|f\odot g| = 
1+|f|+|g|$, $|f^{\star}| = |\rstar{f}| = 1 + |f|$, and 
$|\kstar{k}{e}{f}| = |\krstar{k}{e}{f}|  = 1 + \mathsf{nl}(e) + |f| + k + 1$.

\subsection{Relational semantics}\label{sec:rsem}

In general, the semantics of a regular transducer expression $h$ is a relation
$\rsem{h}\subseteq\Sigma^{\star}\times\Gamma^{\star}$.  This is due to the fact that input
words may be parsed in several ways according to a given expression.  For instance, when
applying a Cauchy product $h=f\cdot g$ to an input word $w\in\Sigma^{\star}$, we split
$w=uv$ and we output the concatenation of $f$ applied to $u$ and $g$ applied to $v$.
There might be several decompositions $w=uv$ with $u\in\dom{f}$ and $v\in\dom{g}$,  in
which case, the parsing is ambiguous and $h$ applied to $w$ may result in several outputs.

We define inductively for an \RTE $h$, the domain $\dom{h}\subseteq\Sigma^{\star}$ and the
relational semantics $\rsem{h}\subseteq\Sigma^{\star}\times\Gamma^{\star}$.  As usual, for
$u\in\Sigma^{\star}$, we let $\rsem{h}(u)=\{v\in\Gamma^{\star}\mid(u,v)\in\rsem{h}\}$.

We also define simultaneously the \emph{unambiguous domain} $\udom{h}\subseteq\dom{h}$
which is the set of words $w\in\dom{h}$ such that parsing $w$ according to $h$ is
unambiguous. This is used in the next subsection to define a functional semantics.
\begin{itemize}[nosep]
  \item $h=\SimpleFun{e}{v}$: As already discussed, we set 
  $\rsem{\SimpleFun{e}{v}}=\{(u,v)\mid u\in L(e)\}$ 
  and $\udom{h}=\dom{h}=L(e)$.
  
  \item $h=f+g$: We have $\dom{f+g}=\dom{f}\cup\dom{g}$, 
  $\rsem{f+g}=\rsem{f}\cup\rsem{g}$  
  and $\udom{f+g}=(\udom{f}\setminus\dom{g})\cup(\udom{g}\setminus\dom{f})$.
    
  \item $h=f\cdot g$ (Cauchy product): We have $\dom{f\cdot g}=\dom{f}\cdot\dom{g}$, 
  for $w\in\Sigma^{*}$ we let $\rsem{f\cdot g}(w)=\bigcup_{w=uv}\rsem{f}(u)\cdot\rsem{g}(v)$ 
  and a word $w$ is in the unambiguous domain of $f\cdot g$ if there is a 
  \emph{unique} factorization $w=uv$ with $u\in\dom{f}$ and $v\in\dom{g}$ and moreover 
  this factorization satisfies $u\in\udom{f}$ and $v\in\udom{g}$:
  $\udom{f\cdot g}= (\udom{f}\cdot\udom{g}) \setminus
  \{uvw\mid v\neq\varepsilon \text{ and } u,uv\in\dom{f} \text{ and } vw,w\in\dom{g}\}$.
    
  \item $h=f\cdot_{r}g$ (reverse Cauchy product): We have $\dom{f\cdot_{r}g}=\dom{f\cdot g}$, 
  $\udom{f\cdot_{r}g}=\udom{f\cdot g}$, and for $w\in\Sigma^{*}$ we let
  $\rsem{f\cdot_{r}g}(w)=\bigcup_{w=uv}\rsem{g}(v)\cdot\rsem{f}(u)$.
    
  \item $h=f^{\star}$ (Kleene star): We have $\dom{h}=\dom{f}^{\star}$, and for
  $w\in\Sigma^{*}$ we let $\rsem{h}(w)=\bigcup_{w=u_{1}\cdots
  u_{n}}\rsem{f}(u_{1})\cdots\rsem{f}(u_{n})$. Notice that if $\rsem{f}$ is 
  \emph{proper}, i.e., if $\varepsilon\notin\dom{f}$, then we may restrict $n$ to at most 
  $|w|$ in the union above. Otherwise, the union will have infinitely many nonempty terms 
  and $\rsem{h}(w)$ may be an infinite language.
  Finally, $\udom{h}$ is the set of words $w\in\Sigma^{\star}$ which have a unique
  factorization $w=u_{1}\cdots u_{n}$ with $n\geq0$ and $u_{i}\in\dom{f}$ for all $1\leq
  i\leq n$, and moreover, for this factorization, we have $u_{i}\in\udom{f}$ for all 
  $1\leq i\leq n$.  Notice that if $\varepsilon\in\dom{f}$, then $\udom{h}=\emptyset$.
    
  \item $h=\rstar{f}$ (reverse Kleene star): We have $\dom{\rstar{f}}=\dom{f^{\star}}$, 
  $\udom{\rstar{f}}=\udom{f^{\star}}$, and for $w\in\Sigma^{*}$ we let
  $\rsem{\rstar{f}}(w)=\bigcup_{w=u_{1}\cdots u_{n}}\rsem{f}(u_{n})\cdots\rsem{f}(u_{1})$.

  \item $h=f\odot g$ (Hadamard product): We have $\dom{f\odot g}=\dom{f}\cap\dom{g}$, 
  for $w\in\Sigma^{*}$ we let $\rsem{f\odot g}(w)=\rsem{f}(w)\cdot\rsem{g}(w)$ 
  and $\udom{f\odot g}=\udom{f}\cap\udom{g}$.

  \item $h=\kstar{k}{e}{f}$ ($k$-star): The domain of $h$ is the set of words
  $w\in\Sigma^{\star}$ which have a factorization $w=u_{1}\cdots u_{n}$ satisfying
  ($\dagger$) $n\geq0$, $u_{i}\in L(e)$ for $1\leq i\leq n$, and $u_{i+1}\cdots
  u_{i+k}\in\dom{f}$ for $0\leq i\leq n-k$.  For $w\in\Sigma^{*}$ we let
  $$
  \rsem{\kstar{k}{e}{f}}(w)=\bigcup_{\substack{w=u_{1}\cdots u_{n} \\ u_{1},\ldots,u_{n}\in L(e)}}
  \prod_{i=0}^{n-k}\rsem{f}(u_{i+1}\cdots u_{i+k}) \,.
  $$
  Notice that a factorization $w=u_{1}\cdots u_{n}$ with $n<k$ and $u_{i}\in L(e)$ for all
  $1\leq i\leq n$, automatically satisfies ($\dagger$)\label{dagger}.  Hence,
  $\bigcup_{n<k}L(e)^{n}\subseteq\dom{h}$.  Moreover, when $n<k$, the empty product in the
  definition above evaluates to $\{\varepsilon\}$ which is the unit for concatenation of
  languages.
  The \emph{unambiguous} domain of $h$ is the set of words $w\in\dom{h}$ which have
  a \emph{unique} factorization $w=u_{1}\cdots u_{n}$ satisfying ($\dagger$) and moreover,
  for this factorization, we have $u_{i+1}\cdots u_{i+k}\in\udom{f}$ for all $0\leq i\leq
  n-k$.
  
  \item $h=\krstar{k}{e}{f}$ (reverse $k$-star): We have $\dom{\krstar{k}{e}{f}}=\dom{\kstar{k}{e}{f}}$, 
  $\udom{\krstar{k}{e}{f}}=\udom{\kstar{k}{e}{f}}$, and for $w\in\Sigma^{*}$ we let
  $$
  \rsem{\krstar{k}{e}{f}}(w)=\bigcup_{\substack{w=u_{1}\cdots u_{n} \\ u_{1,\ldots,u_{n}\in L(e)}}}
  \rsem{f}(u_{n-k+1}\cdots u_{n}) \cdots \rsem{f}(u_{1}\cdots u_{k}) \,.
  $$
\end{itemize}

We show that the inductive definitions of $\dom{h}$ and $\udom{h}$ indeed give the
\emph{domain} of the relation $\rsem{h}$ and ensure functionality.  The proof is an
easy structural induction.

\begin{lemma}\label{lem:dom-rsem}
  Let $h$ be an \RTE and $w\in\Sigma^{\star}$. Then, 
  \begin{enumerate}[nosep]
    \item  $w\in\dom{h}$ if and only if $\rsem{h}(w)\neq\emptyset$.
  
    \item  If $w\in\udom{h}$, then $\rsem{h}(w)$ is a singleton.
  \end{enumerate}
\end{lemma}

\subsection{Functional semantics}\label{sec:fsem}
Our goal is now to define \emph{functions} with regular transducer expressions $h$.  This
can be achieved by a restriction of the relational semantics $\rsem{h}$ to a suitable subset of the domain $\dom{h}$.

The first natural idea is to restrict to the set of input words $w$ on
which $h$ is functional.  Formally, let $\fdom{h}=\{w\in\dom{h}\mid|\rsem{h}(w)|=1\}$ be
the functional domain of $h$.  A functional semantics is obtained by restricting the
relational semantics $\rsem{h}$ to the functional domain $\fdom{h}$.  The unacceptable
problem with this approach is that $\fdom{h}$ need not be regular. For instance, consider 
$h=f+g$ where $f=(\SimpleFun{a}{c}+\SimpleFun{b}{\varepsilon})^{\star}$ and
$g=(\SimpleFun{a}{\varepsilon}+\SimpleFun{b}{c})^{\star}$. We have 
$\dom{h}=\dom{f}=\dom{g}=\{a,b\}^{\star}$. Both $f$ and $g$ are functional: for 
$w\in\{a,b\}^{\star}$ we have $\rsem{f}(w)=\{c^{|w|_{a}}\}$ and 
$\rsem{g}(w)=\{c^{|w|_{b}}\}$. We deduce that $\fdom{h}$ is the set of words with same 
number of $a$'s and $b$'s, which is not a regular set.

We adopt the next natural idea, which is to restrict the relational semantics $\rsem{h}$
to its \emph{unambiguous} domain $\udom{h}$ which ensures functionality by
Lemma~\ref{lem:dom-rsem}.  It is not hard to check by structural induction that both
$\dom{h}$ and $\udom{h}$ are regular languages over $\Sigma$.
We define the \emph{unambiguous} semantics as the restriction of the relational semantics
$\rsem{h}$ to the unambiguous domain $\udom{h}$.  Formally, this is a partial function
$\usem{h}\colon\Sigma^{\star}\to\Gamma^{\star}$ defined for $w\in\udom{h}$ by the equation
$\rsem{h}(w)=\{\usem{h}(w)\}$.

\medskip\noindent

\subsection{Comparing the unambiguous semantics of \cite{AlurFreilichRaghothaman14},
\cite{DGK-lics18} with ours}
\label{sec:comparison}
\begin{enumerate}
	\item \textcolor{red}{$L/v$ of \cite{AlurFreilichRaghothaman14}, $\Ifthenelse{L}{v}{\bot}$ of \cite{DGK-lics18}, our $\SimpleFun{e}{v}$}.

	The base function $L/v$ in \cite{AlurFreilichRaghothaman14} corresponds to our simple 
expression $\SimpleFun{e}{v}$ when $L(e)=L$.
This can be written as a if-then-else $\Ifthenelse{L}{v}{\bot}$ of \cite{DGK-lics18}.  The
semantics of if-then-else $\Ifthenelse{K}{f}{g}$ checks if $w$ is in the regular language
$K$ or not, and appropriately produces $f(w)$ or $g(w)$.  
\item \textcolor{red}{$f \triangleright g$ of \cite{AlurFreilichRaghothaman14}, $\Ifthenelse{\dom{f}}{f}{g}$ of \cite{DGK-lics18}, our $f+g$}.

The \emph{conditional choice}
combinator $f \triangleright g$ of \cite{AlurFreilichRaghothaman14} maps an input $w$ to
$f(w)$ if it is in $dom(f)$, and otherwise it maps it to $g(w)$.  This can again be
written as the $\Ifthenelse{\dom{f}}{f}{g}$ of \cite{DGK-lics18}.  Our analogue of this is
$f+g$.  Note that when $\dom{f}\cap\dom{g}=\emptyset$, all these combinators coincide.

\item \textcolor{red}{$f \oplus g$ of
\cite{AlurFreilichRaghothaman14}, $f \boxdot g$ of \cite{DGK-lics18}, our $h=f\cdot
g$}.
 
The \emph{split-sum} combinator $f \oplus g$ of
\cite{AlurFreilichRaghothaman14} is the Cauchy product $f \boxdot g$ of \cite{DGK-lics18}.
The semantics of $f \boxdot g$, when applied on $w \in \Sigma^*$ produces $f(u)\cdot g(v)$
if there is a unique factorization $w=u\cdot v$ with $u\in\dom{f}$ and $v\in\dom{g}$. Our analogue is $h=f\cdot
g$.  As
mentioned in the introduction,  $f \oplus g$ and  $f \boxdot g$  are  different from our Cauchy product $h=f\cdot
g$ which works on $\udom{h}$.  While our definition preserves associativity $f\cdot(g\cdot
h)=(f\cdot g)\cdot h$, the notions $\oplus, \boxdot$ 
from \cite{AlurFreilichRaghothaman14},
\cite{DGK-lics18} do not.  
\item \textcolor{red}{$\Sigma f$ of
\cite{AlurFreilichRaghothaman14}, $f^{\boxplus}$ of \cite{DGK-lics18}, our $f^{\star}$}.

The \emph{iterated sum} $\Sigma f$ of
\cite{AlurFreilichRaghothaman14} is the Kleene-plus $f^{\boxplus}$ of \cite{DGK-lics18}
which, when applied to $w \in \Sigma^*$ produces $f(u_1) \cdots f(u_n)$ if $w=u_1 \cdots
u_n$ is an unambiguous factorization of $w$, with each $u_i \in \dom{f}$.  This has the
same problems as the Cauchy product compared to our $f^{\star}$.  
\item \textcolor{red}{$f+g$ of \cite{AlurFreilichRaghothaman14}, $f \odot g$ of
\cite{DGK-lics18}, our $f \odot g$}.

The \emph{sum} $f+g$ of
two functions in \cite{AlurFreilichRaghothaman14} is the Hadamard product $f \odot g$ of
\cite{DGK-lics18}, which when applied to $w$ produces $f(w)\cdot g(w)$, provided $w \in
\dom{f}\cap\dom{g}$.  This agrees with our notion of Hadamard product. 
\item  \textcolor{red}{$\Sigma(f, L)$ of \cite{AlurFreilichRaghothaman14}, $\twoplus{L}{f}$ of \cite{DGK-lics18}, our $\kstar{ 2}{L}{f}$}.

Finally, the
\emph{chained sum} $\Sigma(f, L)$ of \cite{AlurFreilichRaghothaman14} is the two-chained
Kleene-plus $\twoplus{L}{f}$ of \cite{DGK-lics18}, which, when applied to $w$ having a
unique factorization $w=u_1\cdot u_2 \cdots u_n$ with $n\geq1$ and $u_i\in L$ for all
$1\leq i\leq n$ produces $\twoplus{L}{f}(w)=f(u_1u_2)\cdot f(u_2u_3) \cdots
f(u_{n-1}u_n)$.  We consider $\kstar{ k}{L}{f}$ instead of $\kstar{2}{L}{f}$ in this
paper, and we additionally check if the blocks $u_{i+1}\cdots u_{i+k} \in \udom{f}$ for
all $0 \leq i \leq n-k$.  
\end{enumerate}
To summarize, our notion of unambiguity is a global one, compared to the
notion in \cite{AlurFreilichRaghothaman14}, \cite{DGK-lics18}, which checks it only at a
local level, thereby leading to the undesirable properties as pointed out already for the
Cauchy, Kleene-star operators.

\subsection{Main results}
The goal of the paper is to construct efficiently, a two-way reversible transducer
equivalent to a given \RTE under the unambiguous semantics.  Consider an \RTE $h$ and some
word $w \in \dom{h}$.  We first parse $w$ by adding some marker symbols $(_f, )_f$ inside
$w$, signifying the scope of the subexpressions $f$ in $h$.  If $w \notin \udom{h}$, then
there can be many ways of parsing $w$.  We build a non-deterministic one-way transducer
$\TParse_h$ which produces all possible parsings of $w$.  Figure \ref{fig:overview} helps
to get an overview of the construction.  
We check if $w$ has at most one parsing
using a 2RFA $B'$; and if so, apply the uniformized parser transducer $\TParse_h'$ on $w$
to obtain the unique parsing of $w$.  $\TParse_h^U$ represents the sequential composition
of $B'$ and $\TParse_h'$.  Finally, the parsing of $w$, $\TParse_h^U(w)$ is taken as input
by a 2RFT, the evaluator transducer $\T_h$, and produces the output of $w$ according to
$h$, making use of the markers.  

\myparagraph{Why not a 2-way machine for the parser?}\label{remark:parser-1v2}
  A natural question to ask is whether we can have a two-way transducer for $\TParse_h$ or even directly construct a two-way machine that evaluates $h$.
  We discuss some difficulties in this direction. Consider for instance $h= f \cdot
  g$. We could have a non-deterministic two-way transducer (2NFT) which guesses the
  point where the scope of $\dom{f}$ ends in the input $w$ and where $\dom{g}$ begins; if
  $w \notin \udom{h}$, it is unclear if in the backward sweep, the machine can go back to
  this correct point so as to apply $f,g$.  On another note, if we design a 2NFT which
  first inserts a marker where the scope of $\dom{f}$ ends and $\dom{g}$ begins, and a
  second 2NFT which processes this, we will require the composition of these two machines. It is
  unclear how we can go about composition of two 2NFTs.  Irrespective of these
  difficulties, a 1NFT is easier to use anytime than a 2NFT, if one can construct one,
  justifying our choice.

In Section~\ref{sec:rat-functions}, we define the parsing relation and construct the
corresponding parser and evaluator for $\RTE[\Rat]$, and then extend to
$\RTE[\Rat,\reverse]$ in Section~\ref{sec:rat-rev-functions}.  For these fragments, both
the parser and the evaluator have size linear in the given expression.  In
Section~\ref{sec:had}, we extend the constructions to handle Hadamard product.  There, we
show that the size of the parser for $h$ is at most exponential in a new parameter, called 
the \emph{width} of $h$.  Section~\ref{sec:kstar} concludes by showing how to handle the
$k$-star operators.

We define the \emph{width} of an RTE $h$, denoted $\w{h}$, intuitively as the maximum number of times a position in $w$ needs to be read to output $h(w)$:
$\w{\SimpleFun{e}{v}}=1$, 
$\w{f + g} = \w{f \cdot g} = \w{f \cdot_r g} = \max(\w{f},\w{g})$, 
$\w{f^{\star}} = \w{\rstar{f}} = \w{f}$, $\w{f \odot g} = \w{f} + \w{g}$
and $\w{\kstar{k}{e}{f}} = \w{\krstar{k}{e}{f}} = 2 + k \times \w{f}$. 
Note that, for $k$-star and reverse $k$-star, we define the width as
$2 + k \times \w{f}$ instead of $1 + k \times \w{f}$ simply to get a uniform expression
for the complexity bounds.

We are ready to state our main theorems, which are proven for each fragment in the corresponding sections.
\begin{theorem}\label{thm:main-relational}
  Let $h$ be a regular transducer expression.  We can define a parsing relation
  $\Parse_{h}$, and construct an evaluator $\T_{h}$ and a parser $\TParse_{h}$ such that
  \begin{enumerate}[nosep]
    \item We have $\dom{\Parse_{h}}=\dom{h}$ and $\udom{h}=\fdom{\Parse_{h}}$.
    
    Moreover, for each $w\in\dom{h}$ and $\alpha\in\Parse_{h}(w)$, the projection of
    $\alpha$ on $\Sigma$ is $\pi_{\Sigma}(\alpha)=w$.
    
    \item The evaluator $\T_h$ is a 2RFT and, when composed with the parsing relation
    $\Parse_{h}$, it computes the relational semantics of $h$:
    $\rsem{h}=\sem{\T_h}\circ\Parse_h$.
    
    Moreover, the number of states of the evaluator is 
    $\norm{\T_{h}}\leq 5|h|\w{h}$.
    
    If $h$ does not use $k$-star or reverse $k$-star, then $\norm{\T_h}\leq 5|h|$.

    \item The parser $\TParse_h$ is a 1NFT which computes the
    parsing relation $\rsem{\TParse_{h}}=\Parse_h$.  
    
    The number of states of the parser is 
    $\norm{\TParse_h}\leq|h|^{\w{h}}$.
    
    If $h$ does not use Hadamard product or $k$-star or reverse $k$-star then
    $\norm{\TParse_h}\leq|h|$.
  \end{enumerate}
\end{theorem}

\begin{theorem}\label{thm:main-functional}
  Let $h$ be a regular transducer expression.  We can construct a 2RFT $\T^{U}_{h}$ which
  computes the unambiguous semantics of $h$: $\usem{h}=\sem{\T^{U}_h}$.  
  The number of states of $\T^{U}_{h}$ is
  $\norm{\T^{U}_{h}}\leq 2^{\mathcal{O}(|h|^{2\cdot \w{h}})}$. 
  Moreover, if $h$ does not use Hadamard product, $k$-star or reverse $k$-star, the
  number of states of $\T^{U}_{h}$ is $2^{\mathcal{O}(|h|^{2})}$.
\end{theorem}

Theorem~\ref{thm:main-relational} is proved in the following sections.  The statements for
rational functions are proved in Lemma~\ref{lem:parser-rational}.
Section~\ref{sec:rat-rev-functions} shows that the result still hold when the rational
functions are enriched with reverse products.  Lemmas~\ref{lem:correctness parsing
Hadamard},~\ref{lem:correctness translator Hadamard} and~\ref{lem:correctness parser
Hadamard} respectively show that items 1,2 and 3 still hold when the Hadamard product is
present, and finally Lemmas~\ref{lem:correctness parsing kstar},~\ref{lem:evaluator kstar}
and~\ref{lem:correctness parser k-star} extend this to the full \RTEs.
Theorem~\ref{thm:main-functional} is proved in Section~\ref{sec:unamb}.

%%%%%%%%%%%%%%%%%%%%%%%%%%%%%%%%%%%%%%%%%%%%%%%%%
\section{Rational functions}\label{sec:rat-functions}

In this section, we deal with \emph{rational} transducer expressions which consist of the
fragment of \emph{regular} transducer expressions defined by the syntax:
$$
h ::= \SimpleFun{e}{v} \mid h+h \mid h\cdot h \mid h^{\star}
$$
where $e$ is a regular expression over the input alphabet $\Sigma$ and $v\in\Gamma^{\star}$.
The semantics $\rsem{h}$ and $\usem{h}$ are inherited from \RTEs 
(Section~\ref{sec:transducer-expressions}).

Our goal is to parse an input word $w\in\Sigma^{\star}$ according to a given rational
transducer expression $h$ and to insert markers resulting in parsed words
$\alpha\in\Parse_h(w)$.  From these marked words, it will be easier to compute the value
defined by the expression, especially in the forthcoming sections where we also deal with
Hadamard product, reverse Cauchy product, reverse Kleene star, and (reverse) $k$-star.

We construct a 1-way non-deterministic transducer (1NFT) $\TParse_h$ which computes the
parsing relation $\Parse_h$.  We also construct a 2-way reversible transducer (2RFT)
$\T_h$, called the \emph{evaluator}, such that 
$\rsem{h}=\sem{\T_{h}}\circ\Parse_{h}$.

In this section, for a \emph{rational} transducer expression $h$, both $\TParse_h$ and
$\T_h$ will have a number of states linear in the size of $h$.
We denote by $\norm{\mathcal{T}}$ the number of states of a transducer $\mathcal{T}$, 
sometimes also called the size of $\mathcal{T}$.
The evaluator $\T_h$ will actually be 1-way, i.e., it is a 1RFT.

\paragraph*{Parsing and Evaluation}

The parser of an expression $h$ will not try to check that a given input word $w$ can be
unambiguously parsed according to $h$.  Instead, it will compute the set $\Parse_h(w)$ of
all possible ways to parse $w$ \wrt $h$.  Hence, we define a parsing relation $\Parse_h$.

We start with an example on a classical regular expression $e=a^{\star}\cdot b + a\cdot
b^{\star}$.  For each occurrence of a subexpression $e_i$, we introduce a pair of
parentheses $\parent{i}{~}$ which is used to bracket the factor of the input word matching
$e_i$.  The above expression $e$ has 9 occurrences of subexpressions: $e_1=a$,
$e_2=e_1^{\star}$, $e_3=b$, $e_4=e_2\cdot e_3$, $e_5=a$, $e_6=b$, $e_7=e_6^{\star}$,
$e_8=e_5\cdot e_7$ and $e_9=e_4+e_8=e$.  Hence, we use 9 pairs of parentheses
$\parent{i}{~}$ for $1\leq i\leq 9$.  The input word $aab$ can be unambiguously parsed
according to $e$, whereas the input word $ab$ admits two parsings:
\begin{align*}
  \Parse_e(aab) & =\{ \parent{9}{\parent{4}{\parent{2}{\parent{1}{a}\parent{1}{a}} \parent{3}{b}}} \} \\
  \Parse_e(ab) & =\{ \parent{9}{\parent{4}{\parent{2}{\parent{1}{a}} \parent{3}{b}}} ~,
  \parent{9}{\parent{8}{\parent{5}{a} \parent{7}{\parent{6}{b}}}} \}
\end{align*}

Observe that the parsing of a word \wrt an expression $e$ can be viewed as the traversal
of the parse tree of $w$ \wrt $e$.  For instance, in Figure~\ref{fig:parse-trees} we
depict the unique parse tree of the word $aab$ \wrt $e$ given above.  
We also give the two parse trees of the word $ab$ \wrt $e$.  

\begin{figure}[!h]
  \begin{subfigure}[b]{0.36\textwidth}
    \centering
    \scalebox{1}{  
      \begin{tikzpicture}[level distance=1cm,
        level 1/.style={sibling distance=3cm},
        level 2/.style={sibling distance=1.5cm}]
        \node {$e_9$}
          child {node {$e_4$}
            child {node {$e_2$}
              child {node {$e_1$} child {node {$a$}}}
              child {node {$e_1$} child {node {$a$}}}
              }
            child {node {$e_3$} child {node {$b$}}}       
            };
      \end{tikzpicture}
    }
  \caption{The unique parse tree of $aab$}
  \label{fig:parse-tree-aab}
\end{subfigure}%
  \begin{subfigure}[b]{0.3\textwidth}
  \centering
  \scalebox{1}{  
    \begin{tikzpicture}[level distance=1cm,
      level 1/.style={sibling distance=3cm},
      level 2/.style={sibling distance=1.5cm}]
      \node {$e_9$}
        child {node {$e_4$}
          child {node {$e_2$}
            child {node {$e_1$}
              child {node {$a$}}
            }
          }
          child {node {$e_3$}
            child {node {$b$}}
            }       
          };
    \end{tikzpicture}
      }
  \caption{A parse tree for $ab$} 
    \label{fig:parse-tree-1-ab}
\end{subfigure}%
  \begin{subfigure}[b]{0.33\textwidth}
  \centering
  \scalebox{1}{  
    \begin{tikzpicture}[level distance=1cm,
      level 1/.style={sibling distance=3cm},
      level 2/.style={sibling distance=1.5cm}]
      \node {$e_9$}
        child {node {$e_8$}
          child {node {$e_5$}
            child {node {$a$}}
            }
          child {node {$e_7$}
            child {node {$e_6$}
              child {node {$b$}}
            }       
          }
        };
    \end{tikzpicture}
    
        }
    \caption{Another parse tree for $ab$}
      \label{fig:parse-tree-2-ab}
  \end{subfigure}%
\caption{Parse trees for $aab$ and $ab$ \wrt $e=a^{\star}\cdot b + a\cdot b^{\star}$
}
\label{fig:parse-trees}
\end{figure}

Below, we define inductively the parsing relation $\Parse_h$ for a rational transducer
expression $h$.  As in the examples above, when $\alpha\in\Parse_h(w)$ with
$w\in\Sigma^{*}$, then the projection of $\alpha$ on $\Sigma^{\star}$ is $w$.  We
simultaneously define the 1NFT parser $\TParse_h$ which implements $\Parse_{h}$ and the
2RFT evaluator $\T_h$.
The \emph{Parser} $\TParse_h$ is a 1NFT satisfying the following invariants:
\begin{enumerate}[nosep]
  \item $\TParse_h$ has a unique initial state $q^{h}_0$, which has no incoming
  transitions,

  \item $\TParse_h$ has a unique final state $q^{h}_F$, which has no outgoing transitions,
  
  \item a transition of $\TParse_h$ either reads a visible letter $a\neq\varepsilon$ and
  outputs $a$,
  or it reads $\varepsilon$ and outputs $\lparent{f}$ or $\rparent{f}$, 
  where $f$ is some subexpression of $h$.
\end{enumerate}

The \emph{Evaluator} $\T_h$ is a reversible transducer that computes $h$ when composed
with $\Parse_h(w)$.   
It will be of the form given in
Figure~\ref{fig:evaluator-generic}, where $\parent{h}{~}$ is the pair of parentheses
associated with (this occurrence of the transducer expression) $h$. It always starts on the left of its input, and ends on the right. This is trivial in this section as $\T_h$ is one-way, but is a useful invariant in the following sections.
Note that an output
denoted $-$ on a transition stands for any word $v\in\Gamma^{\star}$.
In all our figures, we will often use $\TParse_f$ and $\T_f$ for $f$ a subexpression of $h$. Everytime, we separate the initial and final states of these machines from the main part which will be represented by a rectangle. This choice allows us to highlight the particular role of these states which are the unique entry and exit points. Nevertheless, they are still states of $\TParse_f$ and $\T_f$ when counting the number of states.

\begin{figure}[!h]
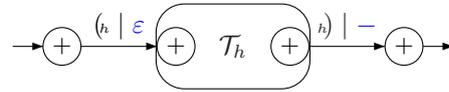

  \centering
  \gusepicture{generic 2RFT}
  \caption{Format of the evaluator transducer}
  \label{fig:evaluator-generic}
\end{figure}

\begin{itemize}
  \item  Let $h= \SimpleFun{e}{v}$ be a basic expression.
  The parsing relation is defined by $\Parse_h(w)=\{\parent{h}{w}\mid w\in L(e)\}$.
  Here, $\Parse_h$ is actually functional and we have 
  $\dom{\Parse_{h}}=L(e)=\fdom{\Parse_{h}}$. Notice that 
  $L(e)=\dom{h}=\udom{h}$.

  Let $\aut_e$ be the non-deterministic Glushkov
  automaton 
  that recognizes $L(e)$.  Recall that $\aut_e$ has a unique initial state
  $q_0^{e}$ with no incoming transitions.  Also, the number of states of $\aut_e$ is
  $1+\mathsf{nl}(e)$ where $\mathsf{nl}(e)$ is the number of literals in the regular
  expression $e$ which denotes the language $L(e)$, \cite{glush}.

  Then, let $\aut'_e$ be the transducer with $\aut_e$ as the underlying input automaton
  and such that each transition simply copies the input letter to the output.  The 1NFT
  $\aut'_e$ realizes the \emph{identity} function restricted to the domain $L(e)$.  The
  parser $\TParse_h$ is then $\aut'_e$ enriched with an initial and a final states $q_0^h$ and $q_F^h$, and transitions $q_0^h\xrightarrow{\epsilon\mid \lparent{h}}q_0^e$ and $q\xrightarrow{\epsilon\mid\rparent{h}}q^h_F$ for $q\in F_{\aut'_e}$, as given on Figure~\ref{fig:parser-L-v}.  The number
  of states in $\TParse_h$ is $\norm{\TParse_h}=\mathsf{nl}(e)+3\leq|h|$.
  
  Notice that $\TParse_h$ is possibly non-deterministic due to the Glushkov automaton
  $\aut_e$.  But it is \emph{functional} and realizes the parsing relation:
  $\rsem{\TParse_{h}}=\Parse_{h}$.

  \begin{figure}[tbh]
    \centering
      \gusepicture{1PLv}
    \caption{Parser for $h= \SimpleFun{e}{v}$.  
    The doubly circled states on the
    right are the accepting states of $\aut_e$.  Note that, if the initial state of
    $\aut_e$ is also an accepting state, then there is a transition
    $q_0^{e}\xrightarrow{\epsilon\mid\rparent{h}}q_F^{h}$ in $\TParse_h$.}
    \label{fig:parser-L-v}
  \end{figure}
  
    The evaluator $\T_h$ for $h= \SimpleFun{e}{v}$, as given in Figure~\ref{fig:evaluator-L-v}, simply reads $\parent{h}{\Sigma^*}$ and outputs $v$ while reading $\rparent{h}$, satisfying the format described in Figure~\ref{fig:evaluator-generic}.

  Remark that in this case the evaluator $\T_h$ is independent of $e$, since the filtering
  on the domain $L(e)$ is done by the parser $\TParse_h$ and not the evaluator $\T_h$.
  We have $\norm{\T_h}=3\leq|h|$. Also, for all $w\in L(e)$ we have 
  $\Parse_h(w)=\{\parent{h}{w}\}$ and $\sem{\T_h}(\parent{h}{w})=v$. 
  Therefore, $\sem{h}=\sem{\T_h}\circ\Parse_h$.

  \begin{figure}[tbh]
    \centering
      \gusepicture{Tv}
    \caption{Evaluator for $h= \SimpleFun{e}{v}$.}
    \label{fig:evaluator-L-v}
  \end{figure}

  \item Let $h=f+g$ and let $\parent{h}{~}$ be the associated pair of parentheses.  Then,
  for all $w\in\Sigma^{*}$, $\Parse_h(w)=\{\parent{h}{\alpha}\mid
  \alpha\in\Parse_f(w)\cup\Parse_g(w)\}$.  We have 
  $\dom{\Parse_{h}}=\dom{\Parse_{f}}\cup\dom{\Parse_{g}}=\dom{f}\cup\dom{g}=\dom{h}$.
  Notice that here the parsing relation is not functional when
  $\dom{f}\cap\dom{g}\neq\emptyset$.

  The parser $\TParse_h$ is as depicted in Figure~\ref{fig:parser-plus} where the three
  pink states are merged and similarly the three blue states are merged so that the number
  of states of $\TParse_h$ is 
  $\norm{\TParse_h}=\norm{\TParse_f}+\norm{\TParse_g}\leq|f|+|g|<|h|$.
  Clearly, $\rsem{\TParse_h}=\Parse_h$.
  
  \begin{figure}[!h]
    \centering
      \gusepicture{1Pplus}
    \caption{Parser for $h= f + g$.}
    \label{fig:parser-plus}
  \end{figure}

  The evaluator $\T_h$ for $h= f + g$ is as given in Figure~\ref{fig:evaluator-plus}.
  The initial states of $\T_f$ and $\T_g$ have been merged in the pink state, similarly
  the final states of $\T_f$ and $\T_g$ have been merged in the blue state.
  $\T_h$ first reads $\lparent{h}$ and goes to a common initial state of $\T_f$ and $\T_g$, then goes to the corresponding evaluator depending on whether it reads $\lparent{f}$ or $\lparent{g}$. When reading $\rparent{f}$ or $\rparent{g}$ it goes to a common state instead of their final one, where it go to the unique final state by $\rparent{h}$ and producing $\epsilon$.  
  We have $\norm{\T_h}=\norm{\T_f}+\norm{\T_g}\leq|f|+|g|<|h|$.
  
  \begin{figure}[!h]
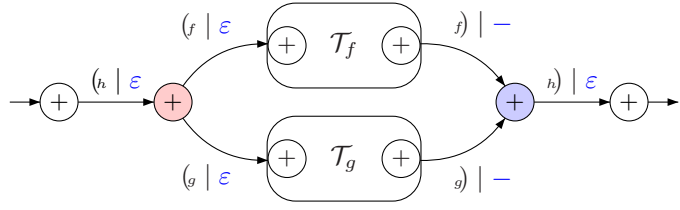

    \centering
      \gusepicture{Tplus}
    \caption{Evaluator for $h= f + g$.}
    \label{fig:evaluator-plus}
  \end{figure}
  
  \item Let $h=f\cdot g$ and let $\parent{h}{~}$ be the associated pair of parentheses.
  Then, for all $w\in\Sigma^{*}$, $\Parse_h(w)=\{\parent{h}{\alpha\,\beta}\mid w=uv,
  \alpha\in\Parse_f(u), \beta\in\Parse_g(v)\}$.  We have 
  $\dom{\Parse_{h}}=\dom{\Parse_{f}}\cdot\dom{\Parse_{g}}=\dom{f}\cdot\dom{g}=\dom{h}$.
  Notice again that if the product of languages $\dom{f}\cdot\dom{g}$ is ambiguous, then $\Parse_h$ is not functional.

  The parser $\TParse_h$ is given in Figure~\ref{fig:parser-cauchy} where the two pink
  states are merged so that the number of states of $\TParse_h$ is
  $\norm{\TParse_h}=1+\norm{\TParse_f}+\norm{\TParse_g}\leq 1+|f|+|g|=|h|$.
  We have, $\rsem{\TParse_h}=\Parse_h$.
  \begin{figure}[!h]
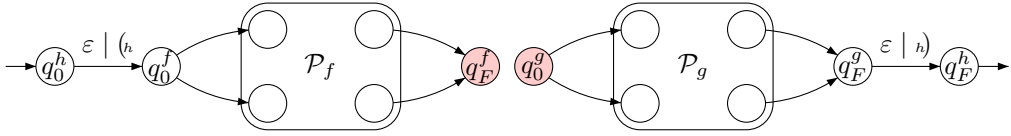

    \centering
      \gusepicture{1Pdot}
    \caption{Parser for $h= f \cdot g$.}
    \label{fig:parser-cauchy}
  \end{figure}

  The evaluator $\T_h$ for $h= f \cdot g$ is as given in Figure~\ref{fig:evaluator-cauchy}.
  The final state of $\T_f$ is merged with the initial state of $\T_g$.
  We have $\norm{\T_h}=1+\norm{\T_f}+\norm{\T_g}\leq 1+|f|+|g|=|h|$.

  \begin{figure}[!h]
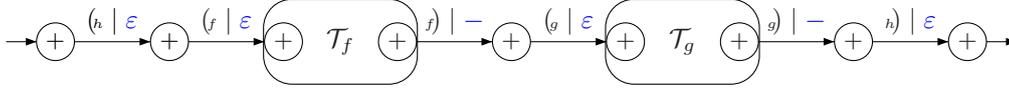

    \centering
      \gusepicture{Tcauchy}
    \caption{Evaluator for $h= f\cdot g$.}
    \label{fig:evaluator-cauchy}
  \end{figure}

  \item Let $h=f^{\star}$ and let $\parent{h}{~}$ be the associated pair of parentheses.
  Then, for all $w\in\Sigma^{*}$, 
  $\Parse_h(w)=\{\parent{h}{\alpha_1\cdots\alpha_n}\mid w=u_1\cdots u_n
  \text{ and } \alpha_i\in\Parse_f(u_i) \text{ for all } 1\leq i\leq n\}$.
  We have $\dom{\Parse_{h}}=\dom{\Parse_{f}}^{\star}=\dom{f}^{\star}=\dom{h}$.
  As above, if the Kleene star of the language $\dom{f}$ is ambiguous, then $\Parse_h$ is not functional.
  
  The parser $\TParse_h$ is as given in Figure~\ref{fig:parser-star-2} where
  the three pink states are merged so that the number of states of
  $\TParse_h$ is $\norm{\TParse_h}=1+\norm{\TParse_f}\leq 1+|f|=|h|$.  
  It is easy to see that $\rsem{\TParse_h}=\Parse_h$.

  \begin{figure}[tbh]
    \centering
      \gusepicture{1Pstar}
    \caption{Parser for $h=f^{\star}$.}
    \label{fig:parser-star-2}
  \end{figure}
  
  The evaluator $\T_h$ for $h=f^{\star}$ is as given in Figure~\ref{fig:evaluator-star}
  where the three pink states are merged so that the number of states of $\T_h$ is
  $\norm{\T_h}=1+\norm{\T_f}\leq 1+|f|=|h|$.

  \begin{figure}[tbh]
    \centering
      \gusepicture{Tstar}
    \caption{Evaluator for $h=f^{\star}$.}
    \label{fig:evaluator-star}
  \end{figure}

\end{itemize}

\begin{lemma}\label{lem:parser-rational}
	Theorem~\ref{thm:main-relational} holds for any rational transducer expression $h$.
	Moreover, in this case we have  $\norm{\T_h}\leq|h|$. 
\end{lemma}

\begin{proof}
  We have already argued during the construction that $\dom{\Parse_{h}}=\dom{h}$,
  $\rsem{\TParse_h}=\Parse_h$, $\norm{\TParse_h}\leq|h|$ and $\norm{\T_h}\leq|h|$.
  It is also  easy to see that the projection on $\Sigma$ of a parsed word
  $\alpha\in\Parse_{h}(w)$ is $w$ itself.
  The remaining claims are proved by structural induction. 
  
  For the base case $h=\SimpleFun{e}{v}$ we have already seen that $\Parse_h$ is
  functional with domain $L(e)=\udom{h}$ and that
  $\rsem{h}(w)=\{v\}=\sem{\T_h}(\Parse_h(w))$ for all $w\in L(e)$.
  
  For the induction, we prove in details the case of Cauchy product. The other cases of 
  sum and Kleene star can be proved similarly. Let $h=f\cdot g$. We first show that 
  $\rsem{h}=\sem{\T_h}\circ\Parse_h$. Let $w\in\dom{h}$.
  
  Let $w'\in\rsem{h}(w)$.  There is a factorization $w=uv$ and $u'\in\rsem{f}(u)$,
  $v'\in\rsem{g}(v)$ with $w'=u'v'$.  By induction, $\rsem{f}=\sem{\T_{f}}\circ\Parse_{f}$ so
  we find $\alpha\in\Parse_{f}(u)$ with $u'=\sem{\T_{f}}(\alpha)$.  Similarly, we find
  $\beta\in\Parse_{g}(v)$ with $v'=\sem{\T_{g}}(\beta)$.  Let
  $\gamma=\parent{h}{\alpha\,\beta}\in\Parse_{h}(w)$.  We have
  $\sem{\T_{h}}(\gamma)=\sem{\T_{f}}(\alpha)\sem{\T_{g}}(\beta)=u'v'=w'$.
  Hence, $w'\in(\sem{\T_{h}}\circ\Parse_{h})(w)$.
  
  Conversely, let $w'\in(\sem{\T_{h}}\circ\Parse_{h})(w)$ and consider $\gamma\in\Parse_{h}(w)$
  such that $w'=\sem{\T_{h}}(\gamma)$.  There is a factorization $w=uv$ and
  $\alpha\in\Parse_{f}(u)$, $\beta\in\Parse_{g}(v)$ with
  $\gamma=\parent{h}{\alpha\,\beta}$. From the definition of the evaluator $\T_{h}$, using 
  the form of parsed words $\alpha=\parent{f}{\alpha'}$ and $\beta=\parent{g}{\beta'}$, 
  we deduce that $\sem{\T_{h}}(\gamma)=\sem{\T_{f}}(\alpha)\sem{\T_{g}}(\beta)
  \in(\sem{\T_{f}}\circ\Parse_{f})(u) \cdot (\sem{\T_{g}}\circ\Parse_{g})(v)
  =\rsem{f}(u)\cdot\rsem{g}(v)\subseteq\rsem{h}(w)$.
  
  It remains to prove that $\udom{h}=\fdom{\Parse_{h}}$. Recall that 
  $$
  \udom{h}=(\udom{f}\cdot\udom{g}) \setminus
  \{uvw\mid v\neq\varepsilon \text{ and } u,uv\in\dom{f} \text{ and } vw,w\in\dom{g}\} \,.
  $$
  Let $w\in\udom{h}$. Then, there is a unique factorization $w=uv$ with $u\in\dom{f}$ and 
  $v\in\dom{g}$, i.e., such that $\Parse_{f}(u)\neq\emptyset$ and 
  $\Parse_{g}(v)\neq\emptyset$. We deduce that 
  $\Parse_{h}(w)=\{\parent{h}{\alpha\,\beta}\mid\alpha\in\Parse_{f}(u), \beta\in\Parse_{g}(v)\}$.
  But we also have $u\in\udom{f}=\fdom{\Parse_{f}}$ and $v\in\udom{g}=\fdom{\Parse_{g}}$. 
  We deduce that $\Parse_{h}(w)$ is a singleton.
  
  Conversely, let $w\in\fdom{\Parse_{h}}$.  Assume that $w=uv=u'v'$ with $u,u'\in\dom{f}$
  and $v,v'\in\dom{g}$.  Let $\gamma=\parent{h}{\parent{f}{\alpha}\parent{g}{\beta}}$ with
  $\parent{f}{\alpha}\in\Parse_{f}(u)$ and $\parent{g}{\beta}\in\Parse_{g}(v)$.
  Similarly, let $\gamma'=\parent{h}{\parent{f}{\alpha'}\parent{g}{\beta'}}$ with
  $\parent{f}{\alpha'}\in\Parse_{f}(u')$ and $\parent{g}{\beta'}\in\Parse_{g}(v')$.  We
  have $\gamma,\gamma'\in\Parse_{h}(w)$, hence $\gamma=\gamma'$.  Therefore also
  $\alpha=\alpha'$.  The projection on $\Sigma$ of $\alpha$ (resp.\ $\alpha'$) is $u$
  (resp.\ $u'$) and we deduce that $u=u'$ and $v=v'$.  Therefore, $w$ has a unique
  factorization $w=uv$ with $u\in\dom{f}$ and $v\in\dom{g}$.  It follows that
  $\Parse_{h}(w)=\{\parent{h}{\alpha\,\beta}\mid\alpha\in\Parse_{f}(u),
  \beta\in\Parse_{g}(v)\}$.  Since $\Parse_{h}(w)$ is a singleton, we deduce that both
  $\Parse_{f}(u)$ and $\Parse_{g}(v)$ are singletons.  By induction, we get $u\in\udom{f}$
  and $v\in\udom{g}$. Finally, we have proved $w\in\udom{h}$.
\end{proof}

%%%%%%%%%%%%%%%%%%%%%%%%%%%%%%%%%%%%%%%%%%%%%%%%%

Consider rational transducer expressions where basic expressions are simply of the form
$\SimpleFun{a}{v}$ with $a\in\Sigma$ and $v\in\Gamma^{*}$ (instead of the more general
$\SimpleFun{e}{v}$).
We can directly construct from such an expresion $h$, a transducer $\aut_{h}$ which
realizes the relational semantics of $h$.  This folklore and rather simple
construction gives a 1NFT of size linear in $|h|$.
Notice also that the unambiguous domain of the expression $h$ coincide with the
unambiguous domain of the 1NFT $\aut_{h}$.  Therefore, we can even restrict to $\udom{h}$
and implement the unambiguous semantics by techniques similar to the constructions in
Section~\ref{sec:unamb}.

But, since the output semiring $(2^{\Gamma^{\star}},\cup,\cdot,\emptyset,\{\varepsilon\})$
is not commutative, we cannot extend this approach and construct a 1NFT for the relational
semantics of regular transducer expressions using Hadamard product, reverse Cauchy
product, reverse Kleene star, $k$-star or reverse $k$-star. Therefore, we decided to use 
in Section~\ref{sec:rat-functions} an approach which can be extended to the two-way 
operators as explained in the remaining sections.

\section{Rational Functions with Reverse Products}\label{sec:rat-rev-functions}

Before turning to the Hadamard product and the (reverse) $k$-star operators, we first
focus on a set of operators which, despite expressing non-rational functions, still enjoy
the linear complexity of the previous section.  Intuitively, it is the set of operators
that only require to process the input once.  It consists of the reverse Cauchy
product
and the reverse Kleene star. We see at the end of this section
that we may also 
add duplicate and reverse functions without changing the linear complexity.
For each of these expressions $h$, we define inductively the parsing relation $\Parse_h$,
the parser $\TParse_h$ and the evaluator $\T_h$.  We will show that 
$\norm{\TParse_h}\leq|h|$ and that $\norm{\T_h}\leq 5|h|$.

\paragraph*{Reverse Cauchy product}

For $h = f\cdot_r g$, the parsing $\Parse_h$ is exactly the same as that for the
expression $h=f\cdot g$.  As a consequence, the parser $\TParse_h$ for $h=f\cdot_r g$ is
as given in Figure~\ref{fig:parser-cauchy}.  Recall that the number of states of
$\TParse_h$ is $\norm{\TParse_h} = \norm{\TParse_f} + \norm{\TParse_g} + 1 \leq |f| + |g| + 1 \leq |h|$.

The evaluator $\T_h$ for $h = f\cdot_r g$ is as given in
Figure~\ref{fig:evaluator-cauchy-rev}. 
Notice that  $\T_{h}$ is a 2RFT and its number of states is
$\norm{\T_h} = \norm{\T_f} + \norm{\T_g} + 3\leq 5|f|+5|g|+3 \leq 5|h|$.

\begin{figure}[!h]
  \centering
    \gusepicture{Tcauchyrev}
  \caption{Evaluator for $h = f\cdot_r g$, with 
  $\alpha$ any letter different from $\lparent{h},\lparent{g}$, 
  $\beta$ any letter different from $\parent{h}{,}$, and
  $\gamma$ any letter different from $\rparent{f},\rparent{h}$.}
  \label{fig:evaluator-cauchy-rev}
\end{figure}

\paragraph*{Reverse Kleene star}

For $h=\rstar{f}$, the parsing $\Parse_h$ is exactly the same as that for the expression
$h=f^{\star}$.  Therefore, the parser $\TParse_h$ for $h= \rstar{f}$ is as given in
Figure~\ref{fig:parser-star-2}.  Recall that the number of states of $\TParse_h$, $\norm{\TParse_h} = \norm{\TParse_f} + 1 \leq |f| + 1 = |h|$.

The evaluator $\T_h$ is depicted in Figure~\ref{fig:evaluator-rev-star}.  
Note that $\T_{h}$ is a 2RFT and its number of states is $\norm{\T_h}=\norm{\T_f}+5\leq 5|f|+5=5|h|$.

\begin{figure}[!h]
  \centering
    \gusepicture{Trstar}
  \caption{Evaluator for $h=\rstar{f}$, with
  $\alpha$ any letter different from $\parent{h}{,}$, and $\beta$ any letter different from $\parent{f}{,}$.}
  \label{fig:evaluator-rev-star}
\end{figure}

\paragraph*{Duplicate and reverse}

The function $\revfn$ simply reverses its input: for $w=a_{1}\cdots a_{n}$ with
$a_{i}\in\Sigma$ we have $\sem{\revfn}(w)=a_{n}\cdots a_{1}$.  We can express it with a
reverse Kleene star: $\revfn=\rstar{(\mathsf{copy})}$ where
$\mathsf{copy}=\sum_{a\in\Sigma}\SimpleFun{a}{a}$ simply copies an input letter.  Using 
constructions above, we obtain a parser and an evaluator of size $\mathcal{O}(|\Sigma|)$.
We construct below even simpler parser and evaluator for $\revfn$.

The duplicate function $\duplicate_{\#}$ is parametrized by a separator symbol $\#$, its 
semantics is given for $w\in\Sigma^{*}$ by $\sem{\duplicate_{\#}}(w)=w\#w$. The function 
$\duplicate_{\#}$ has to read its input twice and cannot be expressed with the combinator 
considered so far. It may be expressed with the Hadamard product as follows:
$\duplicate_{\#}=(\mathsf{copy}\cdot(\SimpleFun{\varepsilon}{\#}))\odot\mathsf{copy}$.
But in general, when we allow Hadamard products, the size of the one-way parser is no 
more linear in the size of the given expression. Hence, we give below direct 
constructions with better complexity.

For $h=\duplicate_{\#}$ or $h=\revfn$, let $\parent{h}{~}$ be the associated pair of
parentheses.  Since they are basic \emph{total} functions, the parsing is defined as
$\Parse_h(w) = \{\lparent{h}w\rparent{h}\mid w\in\Sigma^{\star}\}$, and the parser $\TParse_h$
is as given in Figure~\ref{fig:parser-constant}.  We have $\norm{\TParse_h} = 3$.  We define
the size $|\duplicate_{\#}|=3=|\revfn|$ in order to get $\norm{\TParse_{h}}\leq|h|$ for these
basic functions as well.
The evaluator $\T_h$ is given in Figure~\ref{fig:evaluator-duplicate} when 
$h=\duplicate_{\#}$ and in Figure~\ref{fig:evaluator-reverse} when $h=\revfn$.
In both cases, we have $\norm{\T_{h}}=5\leq 3|h|$.

\begin{figure}[!h]
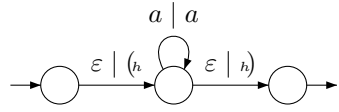

  \centering
    \gusepicture{1Pconst}
  \caption{Parser for duplicate and reverse functions, with $a\in\Sigma$.}
  \label{fig:parser-constant}
\end{figure}

\begin{figure}[!h]
  \centering
    \gusepicture{Tduplicate}
  \caption{Evaluator for $h= \duplicate_{\#}$, with $a\in\Sigma$.}
  \label{fig:evaluator-duplicate}
\end{figure}

\begin{figure}[!h]
  \centering
    \gusepicture{Treverse}
  \caption{Evaluator for $h= \revfn$}
  \label{fig:evaluator-reverse}
\end{figure}

To conclude this section, let us remark that, having rational functions,
duplicate, reverse as well as the composition operator gives the expressive power of the
full \RTEs.

\section{Hadamard product}
\label{sec:had}
In Section~\ref{sec:rat-functions} and Section~\ref{sec:rat-rev-functions}, we consider
rational functions with sum, Cauchy product and Kleene star as well as rational-reverse
functions.  In this section, we extend this fragment with the Hadamard product.

\begin{figure}[!h]
  \centering
  \scalebox{0.8}{  
    \begin{tikzpicture}[state/.style={draw, thick, circle, inner sep=2pt}]
      \node (e9) at (5,10)  {$e_9$};
  
      \node (e8) at (6.5,8.5)  {$e_8$};
      \node (e7) at (7,7)  {$e_7$};
      \node (e5) at (5.5,5.5)  {$e_5$};
      \node (e6) at (7,5.5)  {$e_6$};
  
      \node (e4) at (3.5,8.5)  {$e_4$};
      \node (e2) at (3,7)  {$e_2$};
      \node (e1) at (3,5.5)  {$e_1$};
      \node (e3) at (4.5,5.5)  {$e_3$};
  
      \node (e10) at (4,4)  {$a$};
      \node (e11) at (6,4)  {$b$};
  
      \begin{scope}[->, >=stealth, thick]
        \draw [color=black] (e9) to (e8); 
        \draw [color=black] (e9) to (e4); 
        
        \draw [color=blue] (e4) to (e2); 
        \draw [color=blue] (e4) to (e3); 
        \draw [color=blue] (e2) to (e1); 
        \draw [color=blue] (e3) to (e11); 
        \draw [color=blue] (e1) to (e10); 
        
        \draw [color=red] (e8) to (e7); 
        \draw [color=red] (e8) to (e5); 
        \draw [color=red] (e7) to (e6); 
        \draw [color=red] (e6) to (e11); 
        \draw [color=red] (e5) to (e10); 
      \end{scope}
    \end{tikzpicture}
  }
  \caption{Parse dag for $ab$ \wrt the Hadamard product $f\odot g$, where $\dom{f} = a^*b$ and $\dom{g} = ab^*$.}
  \label{fig:parse-tree-hadamard}
\end{figure}

Just as for the rational functions, we first motivate the parsing of a word \wrt a
Hadamard product expression $e$ as the traversal of the parse \emph{dag} of $w$ \wrt $e$.
As an example, consider $e= (a^{\star}\cdot b) \odot (a\cdot b^{\star})$.  For each
occurrence of a subexpression $e_i$, we introduce a pair of parentheses $\parent{i}{~}$
which is used to bracket the factor of the input word matching $e_i$.  The above
expression $e$ has 9 occurrences of subexpressions: $e_1=a$, $e_2=e_1^{\star}$, $e_3=b$,
$e_4=e_2\cdot e_3$, $e_5=a$, $e_6=b$, $e_7=e_6^{\star}$, $e_8=e_5\cdot e_7$ and
$e_9=\textcolor{red}{e_4}\odot\textcolor{blue}{e_8}=e$.  Hence, we use 9 pairs of
parentheses $\parent{i}{~}$ for $1\leq i\leq 9$.
The input word $ab$ can be unambiguously parsed according to $e$ as follows
$$
\Parse_e(ab) =\{ \lparent{9} \textcolor{blue}{\lparent{4} \lparent{2} \lparent{1}} 
\textcolor{red}{\lparent{8} \lparent{5}} a  \textcolor{blue}{\rparent{1} \rparent{2} 
\lparent{3}} \textcolor{red}{\rparent{5} \lparent{7}\lparent{6}} b 
\textcolor{blue}{\rparent{3} \rparent{4}} \textcolor{red}{\rparent{6}\rparent{7} 
\rparent{8}}\rparent{9} \} \,.
$$

In Figure~\ref{fig:parse-tree-hadamard} we depict the unique parse dag of the word $ab$
\wrt $e$ given above.  Consider a traversal of this parse dag with two heads.  The first
head carries out the traversal according to the parse tree of $e_{4}$ given in blue, while
the second head carries out the traversal according to the parse tree of $e_{8}$ given in
red.  We also fix an ordering on the movement of the two heads as follows.  Between the
points of the traversal which reach a leaf, the first head carries out its traversal, and
then the second head carries out its traversal.  This means that from the root, the first
head moves first, and continues its traversal until it reaches a leaf of the parse tree,
after which the second head starts its traversal and continues until it reaches the same
leaf.  Then, the first head continues its traversal until the next leaf is reached, and
then the second head continues its traversal until seeing the same leaf, and so on.  Then,
the parsing of $w$ \wrt $e$ can be viewed as such a traversal of the parse tree of $w$
\wrt $e$.

\paragraph*{Parsing relation for $h=f\odot g$.}
For $w\in\Sigma^{*}$ we let $\Parse_{h}(w)$ be the set of words $\parent{h}{\alpha}$ such
that
\begin{enumerate}[nosep]
  \item $\restrict{g}\alpha \in\Parse_f(w)$ and $\restrict{f}\alpha \in \Parse_g(w)$.
  Here, $\restrict{g}\alpha$ denotes the projection which erases parentheses which are
  indexed by subexpressions of $g$.  $\restrict{f}\alpha$ is defined similarly.

  \item $\alpha$ does not contain a parenthesis indexed by a subexpression of $g$
  immediately followed by a parenthesis indexed by a subexpression of $f$.
\end{enumerate}
Observe that there are several ways to satisfy the first condition above, even when we
fix the projections $\restrict{g}\alpha$ and $\restrict{f}\alpha$.  This is because
parentheses indexed with subexpressions of $f$ can be shuffled arbitrarily with
parentheses indexed with subexpressions of $g$.  This would break the equality
$\udom{h}=\fdom{\Parse_{h}}$.  Hence, we fix a specific order by giving priority to
parentheses \wrt first argument.
  
\begin{lemma}\label{lem:correctness parsing Hadamard}
  Let $h$ be an \RTE not using $k$-star or reverse $k$-star.  We have
  $\dom{\Parse_{h}}=\dom{h}$ and $\fdom{\Parse_{h}}=\udom{h}$.
\end{lemma}

\begin{proof}
  The proof is by structural induction.  The only new case is when $h=f\odot g$ is an
  Hadamard product.  Using the induction hypothesis and the definitions, it is sufficient
  to prove that $\dom{\Parse_{h}}=\dom{\Parse_{f}}\cap\dom{\Parse_{g}}$ and
  $\fdom{\Parse_{h}}=\fdom{\Parse_{f}}\cap\fdom{\Parse_{g}}$.
  
  Let $w\in\dom{\Parse_{f}}\cap\dom{\Parse_{g}}$.  We find $\alpha_{1}\in\Parse_{f}(w)$
  and $\alpha_{2}\in\Parse_{g}(w)$.  Let $\alpha$ be the (unique) word satisfying
  $\restrict{g}{\alpha}=\alpha_{1}$, $\restrict{f}{\alpha}=\alpha_{2}$ and condition 2 on
  the order of parentheses in $\Parse_{h}$.  We get
  $\parent{h}{\alpha}\in\Parse_{h}(w)\neq\emptyset$, hence we have $w\in\dom{\Parse_{h}}$.
  The converse inclusion $\dom{\Parse_{h}}\subseteq\dom{\Parse_{f}}\cap\dom{\Parse_{g}}$
  is even easier to show.
  
  Let $w\in\fdom{\Parse_{f}}\cap\fdom{\Parse_{g}}$.  It is easy to see that there is a
  unique $\alpha$ satisfying conditions 1 and 2 of the definition of $\Parse_{h}(w)$.
  Therefore, $w\in\fdom{\Parse_{h}}$.  Conversely, let $w\in\fdom{\Parse_{h}}$ and assume
  that $w\notin\fdom{\Parse_{f}}\cap\fdom{\Parse_{g}}$.  For instance, we have
  $w\in(\dom{\Parse_{f}}\setminus\fdom{\Parse_{f}})\cap\dom{\Parse_{g}}$.  Then, we find
  $\alpha_{1},\alpha'_{1}\in\Parse_{f}(w)$ with $\alpha_{1}\neq\alpha'_{1}$ and
  $\alpha_{2}\in\Parse_{g}(w)$.  There is a unique $\alpha$ (resp.\ $\alpha'$) with
  $\restrict{g}{\alpha}=\alpha_{1}$ (resp.\ $\restrict{g}{\alpha'}=\alpha'_{1}$),
  $\restrict{f}{\alpha}=\alpha_{2}$ and condition 2 on the order of parentheses in
  $\Parse_{h}$.  We get $\alpha\neq\alpha'$ and
  $\parent{h}{\alpha},\parent{h}{\alpha'}\in\Parse_{h}(w)$, which contradicts 
  $w\in\fdom{\Parse_{h}}$.
\end{proof}

\paragraph*{Evaluator for $h=f\odot g$}

Let $h=f\odot g$ and let $\parent{h}{~}$ be the associated pair of parentheses.
Then, recall that for all $w\in\Sigma^{*}$, $\Parse_h(w)\subseteq\{\parent{h}{\alpha} \mid
\restrict{g}\alpha\in\Parse_f(w) \text{ and } \restrict{f}\alpha\in\Parse_g(w)\}$.
From the 2RFT $\T_f$, we construct the 2RFT $\T_f^{+g}$ by adding self-loops to all
states labelled with all parentheses associated with subexpressions of $g$, the output of
these new transitions is $\varepsilon$.  When $\restrict{g}\alpha\in\Parse_f(w)$ then the
2RFT $\T_f^{+g}$ behaves on $\alpha$ as $\T_f$ would on $\restrict{g}\alpha$.  Similarly,
we construct $\T_g^{+f}$ from $\T_g$.  
Finally, the evaluator $\T_h$ is depicted in Figure~\ref{fig:evaluator-hadamard}.
Recall the generic form of an evaluator given in Figure~\ref{fig:evaluator-generic}
where we decided to draw the initial state and the final state outside the box to
underline the fact that there are no transitions going to the initial state and no
transitions starting from the final state.
The number of states of $\T_{h}$ is therefore $\norm{\T_{h}}=\norm{\T_{f}}+\norm{\T_{g}}+3$.

\begin{figure}[!h]
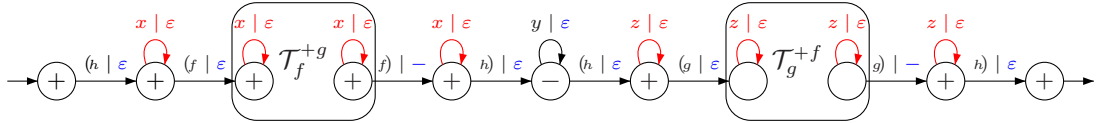

  \centering
  \gusepicture{Thadamard}
  \caption{Evaluator for $h = f \odot g$ where $x$ (resp.\ $z$) is any parenthesis from a 
  subexpression of $g$ (resp.\ $f$) and $y$ is any letter different from $\parent{h}{,}$.}
  \label{fig:evaluator-hadamard}
\end{figure}

\begin{lemma}\label{lem:correctness translator Hadamard}
  Let $h$ be an \RTE not using $k$-star or reverse $k$-star. The translator $\T_{h}$ 
  composed with the parsing relation $\Parse_{h}$ implements the relational semantics: 
  $\rsem{h}=\sem{\T_{h}}\circ\Parse_{h}$.
  Moreover, the number of states of the translator is $\norm{\T_h}\leq 5|h|$.
\end{lemma}

\begin{proof}
  The proof is by structural induction.  
  We have already seen that the statements hold
  when $h$ does not use a Hadamard product.
  The only new case is when $h=f\odot g$ is a Hadamard product. 
  We have $\norm{\T_h}=\norm{\T_f}+\norm{\T_g}+3\leq 5(|f|+|g|+1)=5|h|$.
  We turn to the proof of $\rsem{h}=\sem{\T_{h}}\circ\Parse_{h}$.
  Let $w\in\dom{h}=\dom{f}\cap\dom{g}$.
  
  Let $w'\in\rsem{h}(w)=\rsem{f}(w)\cdot\rsem{g}(w)$.  We write $w'=w_{1}w_{2}$ with
  $w_{1}\in\rsem{f}(w)$ and $w_{2}\in\rsem{g}(w)$.  Since
  $\rsem{f}=\sem{\T_{f}}\circ\Parse_{f}$, we find $\alpha_{1}\in\Parse_{f}(w)$ such that
  $\sem{\T_{f}}(\alpha_{1})=w_{1}$.  Similarly, we find $\alpha_{2}\in\Parse_{g}(w)$ such
  that $\sem{\T_{g}}(\alpha_{2})=w_{2}$. Let $\alpha$ be the unique word satisfying 
  $\restrict{g}{\alpha}=\alpha_{1}$, $\restrict{f}{\alpha}=\alpha_{2}$ and condition 2 on
  the order of parentheses in $\Parse_{h}$. We have $\parent{h}{\alpha}\in\Parse_{h}(w)$.
  It is easy to check that 
  $$
  \sem{\T_{h}}(\parent{h}{\alpha})
  =\sem{\T_{f}^{+g}}(\alpha)\cdot\sem{\T_{g}^{+f}}(\alpha) 
  =\sem{\T_{f}}(\alpha_{1})\cdot\sem{\T_{g}}(\alpha_{2}) 
  =w_{1}w_{2}=w \,.
  $$
  Conversely, let $\parent{h}{\alpha}\in\Parse_{h}(w)$. We have
  \begin{align*}
  \sem{\T_{h}}(\parent{h}{\alpha})
  &=\sem{\T_{f}^{+g}}(\alpha)\cdot\sem{\T_{g}^{+f}}(\alpha) 
  =\sem{\T_{f}}(\restrict{g}\alpha)\cdot\sem{\T_{g}}(\restrict{f}\alpha) 
  \\
  &\in(\sem{\T_{f}}\circ\Parse_{f})(w)\cdot(\sem{\T_{g}}\circ\Parse_{g})(w)
  =\rsem{f}(w)\cdot\rsem{g}(w)=\rsem{h}(w) \,. \qedhere
  \end{align*}
\end{proof}

\paragraph*{One-way parser for $h=f\odot g$}
  
For the Hadamard product $h = f \odot g$, we carry out the parsing of $f$ and $g$ in
parallel by shuffling the parentheses, giving priority to the left argument.

The 1-way parser $\TParse_h$ is depicted in Figure~\ref{fig:parser-hadamard} where
$\TParse_f \otimes \TParse_g$ is a product defined below of the parsers for expressions
$f$ and $g$.  The number of states of $\TParse_h$ is given by $\norm{\TParse_h} =
2+2\times\norm{\TParse_f}\times\norm{\TParse_g}$.

Let $\TParse_f = (Q_f, \Sigma, B, q^f_I, q^f_F, \Delta_f, \mu_f)$ and 
$\TParse_g = (Q_g, \Sigma, C, q^g_I, q^g_F, \Delta_g, \mu_g)$
be the 1-way parsers for expressions $f$ and $g$, respectively.

The set of states of $\TParse_f\otimes\TParse_g$ is  $Q_f \times Q_g\times\{0,1\}$. 
The input alphabet is $\Sigma$ and the output alphabet is $B \cup C$.
The initial state is $(q^f_I,q^g_I,0)$ and the accepting state is $(q^f_F,q^g_F,1)$. 
The transition function of $\TParse_f \otimes \TParse_g$ is defined as follows, 
with $(s,t)\in Q_f \times Q_g$ and $\nu\in\{0,1\}$:
\begin{itemize}[nosep]
  \item if 
  $s \xrightarrow{\varepsilon \mid x} s'$ in $\TParse_f$, then
  $(s,t,0)\xrightarrow{\varepsilon \mid x} (s',t,0)$ in $\TParse_f\otimes\TParse_g$,
  
  \item 
  if $t \xrightarrow{\varepsilon \mid x} t'$ in $\TParse_g$, then
  $(s,t,\nu)\xrightarrow{\varepsilon \mid x} (s,t',1)$ in $\TParse_f\otimes\TParse_g$,
  
  \item
  if $s \xrightarrow{a \mid a} s'$ in $\TParse_f$ and $t\xrightarrow{a \mid a} t'$ in
  $\TParse_g$, then $(s,t,\nu) \xrightarrow{a \mid a} (s',t',0)$ in $\TParse_f\otimes\TParse_g$.

\end{itemize}

\begin{figure}[!h]
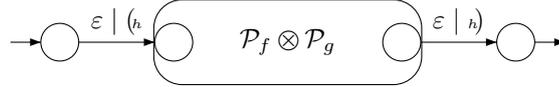

  \centering
  \gusepicture{Hadamard1FT1-alt}
  \caption{Parser for $h = f \odot g$}
  \label{fig:parser-hadamard}
\end{figure}

\begin{lemma}\label{lem:correctness parser Hadamard}
  Let $h$ be an \RTE not using $k$-star or reverse $k$-star.
  The parser $\TParse_{h}$ computes the parsing relation $\Parse_{h}$.
  Moreover, the number of states of the parser is $\norm{\TParse_h} \leq |h|^{\w{h}}$.
\end{lemma}

\begin{proof}
  We first prove that $\rsem{\TParse_{h}}=\Parse_{h}$ 
  by structural induction.  The only new case is when $h=f\odot g$ is a Hadamard product.
  Let $w\in\Sigma^{\star}$.
  
  We first show that $\Parse_{h}(w)\subseteq\rsem{\TParse_{h}}(w)$.
  Let $\parent{h}{\alpha}\in\Parse_{h}(w)$. 
  We have to show that $\alpha$ is accepted by $\TParse_f\otimes\TParse_g$.  
  Consider an accepting run $\rho_{f}$ of $\TParse_{f}$ (resp.\ $\rho_{g}$ of
  $\TParse_{g}$) reading the input word $w$ and producing the projection
  $\restrict{g}\alpha\in\Parse_f(w)$ (resp.\ $\restrict{f}\alpha\in\Parse_g(w)$).  
  We construct an accepting run $\rho$ of $\TParse_f\otimes\TParse_g$ reading $w$ and
  producing $\alpha$ by shuffling $\rho_{f}$ and $\rho_{g}$.  The transitions reading an
  input letter $a\in\Sigma$ are synchronized and between two such synchronized transitions
  we execute first the epsilon moves of $\TParse_{f}$ (the extra bit of the state being
  $0$) and then the epsilon moves of $\TParse_{g}$ (extra bit being $1$).  Notice that
  $\alpha$ starts with $\lparent{f}$ hence $\rho$ starts from the initial state
  $(q^f_I,q^g_I,0)$ and $\alpha$ ends with $\rparent{g}$ so $\rho$ ends in the final state
  $(q^f_F,q^g_F,1)$.
  
  Conversely, we show that $\rsem{\TParse_{h}}(w)\subseteq\Parse_{h}(w)$.
  Let $\parent{h}{\alpha}\in\rsem{\TParse_{h}}(w)$.  There is an accepting run $\rho$ of
  $\TParse_f\otimes\TParse_g$ reading $w$ and producing $\alpha$.
  Let $\rho_{f}$ be the projection on the first component of the run $\rho$ after removing
  transitions of the form $(s,t,\nu)\xrightarrow{\varepsilon\mid x} (s,t',1)$ coming
  from transitions of $\TParse_g$. It is easy to see that $\rho_{f}$ is an accepting run 
  of $\TParse_{f}$ reading $w$ and producing the projection $\restrict{g}\alpha$.
  Therefore, $\restrict{g}\alpha\in\Parse_f(w)$. Similarly, the projection on 
  the second component of $\rho$ after removing transitions of the form 
  $(s,t,0)\xrightarrow{\varepsilon \mid x} (s',t,0)$ is an accepting run $\rho_{g}$ of 
  $\TParse_{g}$ reading $w$ and producing the projection $\restrict{f}\alpha$. We get
  $\restrict{f}\alpha\in\Parse_g(w)$.  Finally, the definition of 
  $\TParse_f\otimes\TParse_g$ ensures that $\alpha$ does not contain a parenthesis
  indexed by a subexpression of $g$ (produced by a transition of the form 
  $(s,t,\nu)\xrightarrow{\varepsilon\mid x} (s,t',1)$) immediately followed by a
  parenthesis indexed by a subexpression of $f$ (produced by a transition of the form 
  $(s,t,0)\xrightarrow{\varepsilon \mid x} (s',t,0)$).
  
  \medskip
  We prove now that $\norm{\TParse_h} \leq |h|^{\w{h}}$.
  Again, the proof is by structural induction on the expression $h$.  We have already seen
  that, when $h$ does not use Hadamard products (in particular for the base cases), we
  have $\norm{\TParse_{h}}\leq|h|=|h|^{\w{h}}$.

  Consider the Hadamard product $h = f \odot g$.  We know that $\norm{\TParse_h} =
  2\times\norm{\TParse_f}\times\norm{\TParse_g} + 2$.  We also know that $|h| = |f| + |g|
  + 1$ and $\w{h} = \w{f} + \w{g}) \geq 2$.
  By induction hypothesis, we have $\norm{\TParse_f} \leq |f|^{\w{f}}$ and $\norm{\TParse_g} \leq |g|^{\w{g}}$.
  Then, we get
  \begin{align*}
    \norm{\TParse_h} = 2\times\norm{\TParse_f} \times \norm{\TParse_g} + 2
    &\leq 2\times|f|^{\w{f}} \times |g|^{\w{g}} + 2
    \\
    &\leq (|f| + |g|)^{\w{f}+\w{g}} \leq |h|^{\w{h}} \,.
  \end{align*}
  The other cases are easy.  When $h=f+g$ or $h=f\cdot g$ or $h=f\cdot_{r}g$, then we have
  \begin{align*}
    \norm{\TParse_h} \leq 1+\norm{\TParse_f}+\norm{\TParse_g}
    &\leq 1 + |f|^{\w{f}} + |g|^{\w{g}}
    \\
    &\leq 1 + |f|^{\w{h}} + |g|^{\w{h}}
    \leq |h|^{\w{h}} \,.
  \end{align*}
  When $h$ is $f^{\star}$ or $\rstar{f}$ then
  $\norm{\TParse_h} = 1+\norm{\TParse_f} \leq 1 + |f|^{\w{f}} = 1 + |f|^{\w{h}} \leq |h|^{\w{h}}$.
\end{proof}
Finally, the next proposition shows that this exponential blow-up in the width of the expression is unavoidable.

\begin{proposition}\label{prop:had-lb}
  For all $n>0$, there exists an RTE $C_n$ of size $O(n)$ such that $\w{C_n}=n$ and any
  parser of $C_n$ is of size $\Omega(2^n)$.
\end{proposition}

\begin{proof}
  The key argument is that any parser of an RTE has to at least recognize its domain.
  As the Hadamard product restrict the domain to the intersection of its subexpressions, we can construct an RTE $C_n$ which is the Hadamard product of 
  $n$ subexpressions of fixed size, and whose intersection is a single word of size $2^n$. 
  Consequently, any parser of $C_n$ is of size at least $2^n$.

  Let $\Sigma = \{ 1,\dots,n\}$, $\Sigma_{<i}=\{1,\ldots,i-1\}$ and 
  $\Sigma_{>i}=\{i+1,\ldots,n\}$ for all $i$.
  For $n>i\geq 1$, we define $L_i=(\Sigma_{<i}^*i \Sigma_{<i}^* \Sigma_{>i})^*$ and 
  $L_n=\Sigma_{<n}^*n \Sigma_{<n}^* n$ .
  Moreover, let us define
  recursively the sequence of words $(u_i)_{1\leq i\leq n}\in
  \Sigma^*$ as follows: $u_1 = 1$, $u_i = u_{i-1} i u_{i-1}$ for
  $2\leq i<n$ and $u_n = u_{n-1} n u_{n-1} n$.
  By construction, we have $|u_n|=2^n$.
  Also, each $L_i$ can be recognized by an automaton of size $2$ as described in
  Figure~\ref{Fig:AutomateAi}, and $\{u_n\}=\bigcap_{i=1}^n L_i$.
  
  \begin{figure}[t]
    \begin{center}
      \begin{tikzpicture}[initial text=]
        \node[circle,draw,initial=left,accepting] (p) at (0,0) {$0$};
        \node[circle,draw] (q) at (2,0) {$1$};
        
        \draw (p) edge[->,loop above] node {$\Sigma_{<i}$} (p);
        \draw (q) edge[->,loop above] node {$\Sigma_{<i}$} (q);
        \draw (p) edge[->,bend left] node[above] {$i$} (q);
        \draw (q) edge[->,bend left] node[below] {$\Sigma_{>i}$} (p);

        \node[circle,draw,initial=left] (a) at (5,0) {$0$};
        \node[circle,draw] (b) at (6.5,0) {$1$};
        \node[circle,draw,accepting] (c) at (8,0) {$2$};
        
        \draw (a) edge[->,loop above] node {$\Sigma_{<n}$} (a);
        \draw (b) edge[->,loop above] node {$\Sigma_{<n}$} (b);
        \draw (a) edge[->] node[above] {$n$} (b);
        \draw (b) edge[->] node[above] {$n$} (c);
      \end{tikzpicture}
      \vspace{-4mm}
      \caption{On the left, the automaton $A_i$, for $i<n$.  On the right, the automaton
      $A_n$.}\label{Fig:AutomateAi}
      \vspace{-7mm}
    \end{center}
  \end{figure}
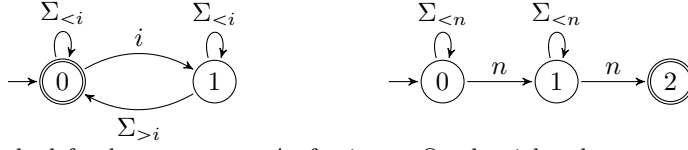
  
  Finally, let us define $C_n=\bigodot_{i=1}^n (\SimpleFun{L_i}{\epsilon})$ whose size is $O(n)$ and width is $n$.
  Its domain $\dom{C_n}$ is the singleton $\{u_n\}$ whose size is exponential in $n$, and thus any parser of $C_n$ has to be of size at least exponential in $n$.
  
  Note that our definition of languages $L_i$ depends on an alphabet of size $n$. Equivalently, we can use unary encodings of each integer $i$ to define languages of size linear instead of constant, but with an alphabet of constant size.
\end{proof}

%%%%%%%%%%%%%%%%%%%%%%%%%%%%%%%%%%%%%%%%%%%%%%%%%
\section{$k$-star operator and its reverse}\label{sec:kstar}

In this section, we extend the set of \RTEs discussed in Sections~\ref{sec:rat-functions},
\ref{sec:rat-rev-functions} and \ref{sec:had} with the $k$-star and the
reverse $k$-star operators.

\subsection{Parsing relation for $h=\kstar{k}{e}{f}$ and $h=\krstar{k}{e}{f}$}

We will first describe, with the help of an example, the parsing set $\Parse_h(w)$ that we
want to compute given a word $w$ in the domain of $h=\kstar{k}{e}{f}$ and
$h=\krstar{k}{e}{f}$.  The parsing relation $\Parse_h$ is exactly the same for both
$h=\kstar{k}{e}{f}$ and $h=\krstar{k}{e}{f}$.  Let $L=L(e)$.  To motivate the parsing
relation, we will also give a brief overview of the working of the evaluator $\T_h$ that
computes $h(w)$ from a parsing in $\Parse_h(w)$ of the word $w$.

Recall that for a word $w$ to be in the domain of $h$, $w$ should have a factorization
$w=u_{1}\cdots u_{n}$ satisfying ($\dagger$), i.e., $n\geq0$, $u_{i}\in L=L(e)$ for $1\leq
i\leq n$, and $u_{i+1}\cdots u_{i+k}\in\dom{f}$ for $0\leq i\leq n-k$.  Given a word $w$
and one such factorization $w = u_1u_2\cdots u_n$, we will refer to the factor $u_{i+1}
u_{i+2} \cdots u_{i+k}$ as the $i$th \emph{block} of $w$.  While evaluating the expression
$h = \kstar{k}{e}{f}$ on $w$, $f$ is first applied on the $0$th block of $w$, then the
$1$st block and so on, until the $(n-k)$th block of $w$.  For $h = \krstar{k}{e}{f}$, the
order of evaluation is reversed, i.e., $f$ is first applied on the $(n-k)$th block of $w$,
then the $(n-k-1)$th block and so on, until the $0$th block of $w$.

The parsing of $w$ \wrt $h$ should contain the information required for this evaluation.
In other words, we need to add the parentheses \wrt $f$ on the $0$th block, the $1$st
block and so on, until the $(n-k)$th block of $w$.  Since we want to construct a parser
that is one-way, we need to shuffle the parentheses that arise from the application of $f$
on different blocks.  Consequently, we need some way to distinguish the parentheses that
arise due to the application of $f$ on a block from the parentheses that arise due to the
application of $f$ on other blocks.  To this end, we will use parentheses indexed by
$\{1,2,\cdots,k\}$.

\begin{figure}[h]
  \centering
    \includegraphics[width=.7\linewidth,scale=.1]{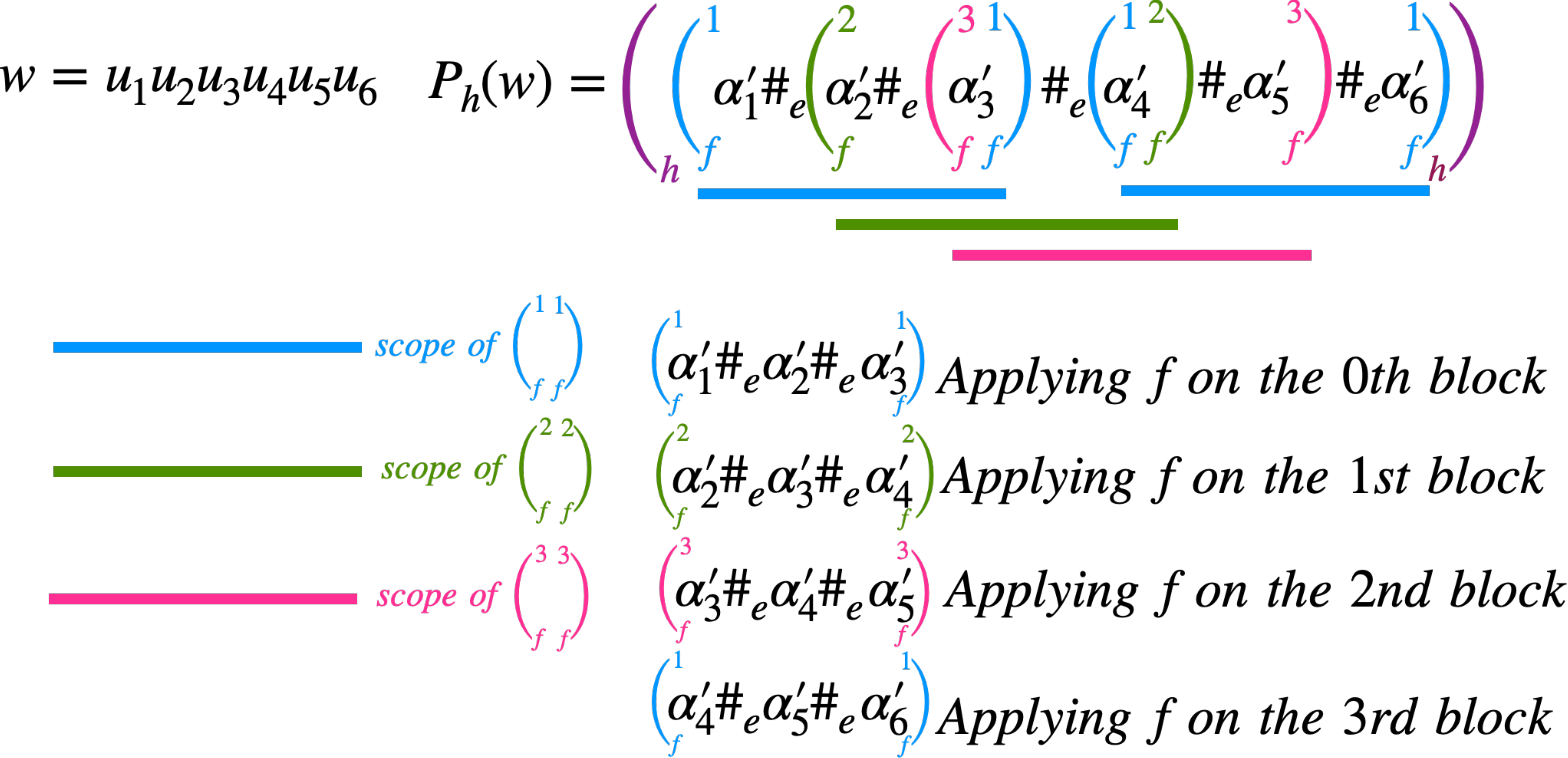}
        \captionof{figure}{$\Parse_h(w)$ for $h=\kstar{3}{e}{f}$ on a word $w$ in $L(e)^6$.}
    \label{fig:k-star-output}
\end{figure}

In Figure~\ref{fig:k-star-output}, we illustrate a parsing in $\Parse_h(w)$ when
$h=\kstar{3}{e}{f}$ on a word $w = u_1u_2u_3u_4u_5u_6$ with each $u_i\in L(e)$.  As depicted
in the figure, while processing the $i$th block $u_{i+1}u_{i+2}\cdots u_{i+k}$, $\T_h$
considers only the parentheses indexed by $i+1 \bmod k$, and ignores all other
parentheses.  After reading the $k$th $L$-factor of the $i$th block, $\T_h$ checks if
there are any more blocks to be read.  If not, then $\T_h$ is done with its computation.
Otherwise, $\T_h$ goes back, and repeats the same process on block $(i+1)$, but this time
considering only parentheses indexed by $i+2 \bmod k$.

To summarise, an $h$-parsing of $w$ \wrt a factorization $w=u_1u_2\cdots u_n$ satisfying 
($\dagger$) and $n\geq k$ is a word of the form 
$\parent{h}{\alpha_{1}\#_{e}\alpha_{2}\cdots\#_{e}\alpha_{n}}$ such that:
\begin{itemize}[nosep]
  \item It starts with an opening parenthesis $\lparent{h}$ and ends with a closing
  parenthesis $\rparent{h}$, and the projection of $\alpha_{i}$ on $\Sigma$ is $u_{i}$.

  \item Each block $u_{i+1} u_{i+2} \cdots u_{i+k}$ is decorated with a parenthesisation
  corresponding to $\Parse_f(u_{i+1} u_{i+2} \cdots u_{i+k})$, such that each of these
  parentheses is indexed by $i+1 \bmod k$.

  \item Between any two consecutive $L$-factors $u_i$ and $u_{i+1}$, there is a $\#_{e}$.
  In particular, immediately after each $\rparent{f}$ and immediately before each
  $\lparent{f}$, there is a $\#_{e}$.

  \item Between any two letters of the $i$th $L$-factor $u_i$ of $w$, the parentheses
  appear in non-decreasing order of their indices \wrt some order $\leq_i$ (described in
  detail in the formal parsing).
\end{itemize}
Note that there are several ways to satisfy just the first three conditions given above,
as the parentheses indexed $i$ could be shuffled arbitrarily with parentheses indexed by $j
\neq i$.  Since this would violate the requirement $\udom{h}=\fdom{\Parse_{h}}$ that we
crucially depend on, we have the final condition that fixes a specific order of
parenthesisation.

\begin{remark}
  Note that in a factorization of $u_1 \cdots u_n$ of $w$, it is possible that $u_i =
  \varepsilon$.  This could lead to problems that violate the requirement
  $\udom{h}=\fdom{\Parse_{h}}$.  To address this, we have additional conditions (rule 5)
  in the formal definition of the parsing.
\end{remark}

Consider a parsing $\parent{h}{\alpha_1\#_e\alpha_2\#_e\alpha_3\#_e
\alpha_4\#_e\alpha_5\#_e\alpha_6}$ for $h$ and $w$ from Figure~\ref{fig:k-star-output}.
Here, 
$\alpha_1$ has only parentheses indexed by $1$, $\alpha_2$ has only parentheses indexed by $1$ and $2$, $\alpha_3$ and $\alpha_4$ have parentheses indexed by $1$, $2$ and $3$, $\alpha_5$ has only parentheses indexed by $1$ and $3$, and $\alpha_6$ has only parentheses indexed by $1$.
In particular, $\alpha_3$ for instance can be of the form $\lparent{f}^{3}a_1\lparent{}^1\ \rparent{}^2\lparent{}^{2}\lparent{}^{3}a_2\lparent{}^2\ \rparent{}^{3} \cdots \rparent{}^{2}\lparent{}^{3}a_l\rparent{}^{3}\rparent{f}^{1}$, where $a_i \in \Sigma$.

\paragraph*{Formal definition of the parsing relation for $h=\kstar{k}{e}{f}$ and
$h=\krstar{k}{e}{f}$}

First let $B$ be the set of parentheses appearing in the parsing of $f$.
We define $B_i$, for $1\leq i \leq k$, to be the set $B$ indexed by $i$.
We write $\kaparent{i}$ for either $\klparent{i}$ or $\krparent{i}$.  Additionally, for
ease of notations $k\bmod k$ is set to $k$ instead of $0$ as commonly defined.
Let $L=L(e)$ and $w$ be an input word of $h$.
The parsing set $\Parse_h(w)$ is the set of words
$\parent{h}{\alpha_1\#_e\alpha_2\#_e\ldots\#_e\alpha_n}$ such that there is a
factorization $w=u_1\ldots u_n$ where for all $i\leq n$, $u_i\in L$ and
$\pi_\Sigma(\alpha_i)=u_i$ and either $n<k$ and $\alpha_i=u_i$, hence the parsing is
$\parent{h}{u_1\#_eu_2\#_e\ldots\#_eu_n}$, or $n\geq k$ and:
\begin{enumerate}[nosep]\label{def:parsing k-star}
  \item for all $0\leq i\leq n-k$, $\pi_{i+1\bmod k}(\alpha_{i+1}\ldots\alpha_{i+k})\in 
  \Parse_f(u_{i+1}\ldots u_{i+k})$ where $\pi_j$ is the function projecting away all 
  parentheses $\kaparent{\ell}$ for $\ell\neq j$ and erasing the exponent $j$,

  \item for all $1\leq i<j \leq k$, $\alpha_{i}$ does not contain any parenthesis $\kaparent{j}$,

  \item for all $n-k+1\leq j<i \leq n$, $\alpha_{i}$ does not contain any parenthesis 
  $\kaparent{j\bmod k}$, 

  \item $\alpha_i$ ends with $\krparentf{i+1\bmod k}{f}$ if $i\geq k$,
  
  \item $\alpha_i$ starts with $\klparentf{i\bmod k}{f}$ if $i\leq n-k+1$, and if 
  $n-k+1<i$, then either $\alpha_{i}$ starts with a letter of $\Sigma$ or
  $\alpha_i=\krparentf{i+1\bmod k}{f}$ or $\alpha_{i}=\epsilon$ (if also $i<k$),
  
  \item for all $\alpha_i$ and for all $j,j'$, if $\kaparent{j}\kaparent{j'}$ appears in
  $\alpha_i$ and $\kaparent{j'}\neq\krparentf{i+1\bmod k}{f}$, then $j\leq_i j'$, where
  $\leq_i$ is defined as $i+1 \bmod k\leq_i i+2 \bmod k \leq_i \cdots \leq_i i$.
\end{enumerate}

\begin{lemma}\label{lem:correctness parsing kstar}
  Let $h$ be an RTE.
  We have $\dom{\Parse_h}=\dom{h}$ and $\fdom{\Parse_h}=\udom{h}$.
\end{lemma}

\begin{proof}
  As with the previous cases, the proof is by structural induction.
  Using Lemma~\ref{lem:correctness parser Hadamard}, the only cases left are the $k$-star
  operator and its reverse.
  We only prove the result for $k$-star.  As the reverse $k$-star parses the input in the
  same way, the proof will hold for both operators.

  Let then $h=\kstar{k}{e}{f}$ and $w$ be an input word of $h$.
  If $w\in\dom{h}$, then there exists a factorization $w=u_1\cdots u_n$ satisfying that for all $i\leq n$, $u_i\in L(e)$ and 
  either $n<k$, in which case $\parent{h}{u_1\#_eu_2\#_e\cdots\#_eu_n}\in\Parse_h(w)$ and hence $w\in\dom{\Parse_h}$, 
  or $n\geq k$ and $u_{j+1}\cdots u_{j+k}$ belongs to $\dom{f}$ for all $0\leq j\leq n-k$. 
  We construct a parsing $\parent{h}{\alpha_1\#_e\cdots \#_e\alpha_n}$ for this factorization $w=u_1\cdots u_n$.
  Using the induction hypothesis, for all $0\leq j\leq n-k$, there exists 
  a word $\beta_j$ such that $(u_{j+1}\cdots u_{j+k},\beta_j)\in \Parse_f$. 
  Then by definition, $\pi_\Sigma(\beta_j)=u_{j+1}\cdots u_{j+k}$.  Let $\gamma_j$ be
  $\beta_j$ where all parentheses are indexed by $m=j+1\bmod k$.  
  There is a unique factorization $\gamma_j=\gamma_j^1\cdots \gamma_j^k$ such that
  $\pi_\Sigma(\gamma_j^\ell)=u_{j+\ell}$ for $1\leq\ell\leq k$, 
  and $\gamma_{j}^{\ell}$ starts with a letter from $\Sigma$ or $\gamma_{j}^{\ell}=\varepsilon$
  for $1<\ell<k$,
  and $\gamma_{j}^{k}$ starts with a letter from $\Sigma$ or $\gamma_{j}^{k}=\krparentf{m}{f}$.
  Notice that $\gamma_{j}^{1}$ starts with $\klparentf{m}{f}$ and $\gamma_{j}^{k}$ ends 
  with $\krparentf{m}{f}$.
  Then $\alpha_i$ is defined by merging all words $\gamma_j^\ell$ for $j+\ell=i$, 
  shuffling parentheses and synchronizing on letters from $\Sigma$.
  Notice that $\ell$ is comprised between $1$ and $k$, and thus we shuffle at most $k$
  such $\gamma_j^\ell$, with indices $j$ between $\max(0,i-k)$ and $\min(i-1,n-k)$.
  This shuffling can be uniquely defined as follows:
  \begin{enumerate}[nosep]
    \item[($i$)] if $i\leq n-k+1$ we take all parentheses of $\gamma_{i-1}^{1}$ up to the
    first letter of $\Sigma$ if any, in this case, $\alpha_i$ starts with
    $\klparentf{i\bmod k}{f}$,
  
    \item[($ii$)] if $k\leq i$, then $\gamma_{i-k}^{k}$ ends with $\krparentf{i+1\bmod 
    k}{f}$ which we put at the end of $\alpha_{i}$.
  \end{enumerate}    
  Then, we proceed from left to right, iterating the two steps below until exhaustion of all $\gamma_j^\ell$.
  \begin{enumerate}[nosep]
    \item[($iii$)] we take the next letter of $\Sigma$, if any, which occurs in each
    $\gamma_j^\ell$ as they all project onto $u_i$,
  
    \item[($iv$)] we take the following parentheses on each $\gamma_j^\ell$, with increasing indexes according to 
    the order $\leq_i$, until the next letter from $\Sigma$, if any.
  \end{enumerate}
  By construction, as $\alpha_i$ contains all $\gamma_j^\ell$ for $j+\ell=i$, we have 
  $\pi_{i+1\bmod k}(\alpha_{i+1}\ldots\alpha_{i+k}) 
  =\pi_{i+1\bmod k}(\gamma_{i}^1\ldots\gamma_i^k)=\beta_{i}
  \in\Parse_f(u_{i+1}\ldots u_{i+k})$, 
  hence condition of the parsing relation (1) is satisfied.
  Next, if $1\leq i<j\leq k$, then there is no $\gamma_{j-1}^\ell$ with $j-1+\ell=i$, 
  hence $\alpha_{i}$ satisfies condition (2) . 
  Similarly, we can check that it satisfies condition (3).
  By Step ($ii$) (resp.\ ($i$)) we see that $\alpha_{i}$ satisfies condition (4) (resp.\ 
  the first part of condition (5)).
  To prove the second part of condition (5), we first note that if $n-k+1<i\leq n$ and 
  $j+\ell=i$, then $\ell>1$. Then, either all 
  $\gamma_{j}^{\ell}$ with $j+\ell=i$ start with the same letter from $\Sigma$ or 
  $\alpha_i=\krparentf{i+1\bmod k}{f}$ or $\alpha_{i}=\epsilon$ (if $i<k$).
  Finally, step ($iv$) above ensures that we satisfy point (6) of the parsing relation. 
  As a result, the word $\parent{h}{\alpha_i\#_e\ldots \#_e\alpha_n}$ is a parsing of $w$,
  and thus $w\in\dom{\Parse_h}$.
  
  Conversely, let $h\in\dom{\Parse_h}$.
  Then there exists a parsing word $\parent{h}{\alpha_i\#_e\ldots \#_e\alpha_n}$ of $w$.
  As such, let $u_i=\pi_\Sigma(\alpha_i)$. By definition, $w=u_1\ldots u_n$ and for all $i$, $u_i\in L(e)$.
  If $n<k$, then $w\in\dom{h}$ by definition.
  Assume now $n\geq k$.
  Then for the words $\alpha_i$ we have in particular for all $0\leq i\leq n-k$, $\pi_{i+1\bmod k}(\alpha_{i+1}\ldots\alpha_{i+k})\in \Parse_f(u_{i+1}\ldots u_{i+k})$.
  Thus by induction hypothesis $u_{i+1}\ldots u_{i+k}$ belongs to the domain of $f$ and consequently $w\in\dom{h}$.
  
  We now turn to the equality $\fdom{\Parse_h}=\udom{h}$.  We prove that
  $\dom{\Parse_h}\setminus\fdom{\Parse_h}=\dom{h}\setminus\udom{h}$.  Together with the
  previous equality this gives the result.  Let $w$ be in $\dom{h}\setminus\udom{h}$.
  This means that either $(a)$ there exists two different factorizations $w=u_1\ldots u_n$
  satisfying condition ($\dagger$) on page \pageref{dagger}, i.e., such that either $n<k$
  or ($n\geq k$ and for all $0\leq i\leq n-k$, $u_{i+1}\ldots u_{i+k}$ belongs to
  $\dom{f}$), or $(b)$ there exists one such factorization with $k\leq n$ and at least one
  $i$ such that $u_{i+1}\ldots u_{i+k}$ belongs to $\dom{f}\setminus\udom{f}$.
  
  If $(a)$ holds, there exists two such factorizations, and the parsings corresponding to the two different factorization constructed by the procedure above
  will have the $\#_e$ symbols at different positions of the parsings, and hence we have two parsings of $w$ and $w\notin\fdom{\Parse_h}$.
  If $(b)$ holds, then there is one factorization with an integer $i$ such that $u_{i+1}\ldots u_{i+k}$ belongs to $\dom{f}\setminus\udom{f}$. By induction hypothesis there are two different parsings of $u_{i+1}\ldots u_{i+k}$ for $f$. Using the procedure above, we then get two different $\gamma_i$ and $\gamma_i'$, which means that we can construct two different $\alpha_{i+1}\ldots\alpha_{i+k}$ and 
  $\alpha_{i+1}'\ldots\alpha_{i+k}'$. In the end, we get two different parsings for $w$, and hence $w\notin\fdom{\Parse_h}$.
  
  Conversely, let $w$ be in $\dom{\Parse_h}\setminus\fdom{\Parse_h}$.
  Then there exist two different parsings $\parent{h}{\alpha_1\#_e\ldots \#_e\alpha_n}$ and $\parent{h}{\beta_1\#_e\ldots\#_e \beta_m}$ of $w$.
  If $\pi_\Sigma(\alpha_i)\neq \pi_\Sigma(\beta_i)$ for some $i$, then there exist two
  valid and different factorizations of $w$ satisfying ($\dagger$), and thus $w$ does not
  belong to $\udom{h}$.
  Otherwise, it means there is an integer $0\leq i\leq n-k$ such that 
  $\pi_{i+1\bmod k}(\alpha_{i+1}\ldots\alpha_{i+k})$ and
  $\pi_{i+1\bmod k}(\beta_{i+1}\ldots\beta_{i+k})$ are valid but different $f$-parsings of $u_{i+1}\ldots u_{i+k}$.
  It follows that $u_{i+1}\ldots u_{i+k}$ does not belong to $\udom{f}$ and thus $w$ does not belong to $\udom{h}$.
\end{proof}

\subsection{Evaluators}

Here, we propose the evaluators for the $k$-star and the reverse $k$-star operators, and give an upper bound on their size.

\begin{figure}[!h]
  \centering
    \gusepicture[scale=1.2]{Tk-star}
    \caption{Evaluator for $h = \kstar{k}{e}{f}$.  Here, $x_i\in\{\kaparent{j}\mid j\neq 
    i\}\cup\{\#_{e}\}$,
    $y_i \neq \rparent{f}^i, \lparent{f}^{i+1 \bmod k}$, $\alpha \in \Sigma \cup
    \{\#_e\}$, and $\beta \neq \lparent{h}$.}
    \label{fig:evaluator-k-star}
\end{figure}

\myparagraph{Evaluator for $k$-star}
We start with $h=\kstar{k}{e}{f}$ for which the evaluator $\T_h$ is depicted in
Figure~\ref{fig:evaluator-k-star}.  It is a 2RFT that takes as input a parsing in
$\Parse_h(w)$ for a word $w$, and computes the output $h(w)$.  More precisely, let $\T_f$
be the transducer that computes $f(w)$ given a parsing in $\Parse_f(w)$.  The 2RFT $\T_h$
has $k$ copies of $\T_f$, namely $\T^1_f, \T^2_f, \cdots, \T^k_f$.  The idea is that the
copy $\T^i_f$ should consider only parentheses indexed by $i$ and ignore all other
parentheses.  We construct $\T^i_f$ from $\T_f$ as follows: to all states of $\T_f$, we
add self-loops labelled with all parentheses indexed by $j$ where $j \neq i$, and
$\#_{e}$, and the output of these new transitions is $\varepsilon$.  Also, a transition of
$\T_{f}$ reading a parenthesis $\aparent{g}$ for $g$ a subexpression of $f$ is relabelled
with $\kaparentf{i}{g}$.  If $\T_f$ is a 2RFT, then so is $\T^i_f$ for $1\leq i\leq k$.

Let $w = u_1u_2\cdots u_n$ be a factorisation of $w$ with each $u_i\in L$.  If $n<k$ then
the parsing of $w$ \wrt $h$ is defined as $\parent{h}{u_1\#_eu_2\#_e\ldots\#_eu_n}$.  If
$k \leq n$, then the corresponding parsing is
$\parent{h}{\alpha_{1}\#_{e}\alpha_{2}\cdots\#_{e}\alpha_{n}}$ satisfying conditions (1-6)
on page~\pageref{def:parsing k-star}.  In particular, for $0\leq i\leq n-k$ and
$m=i+1\bmod k$, the block $\alpha_{i+1}\cdots\alpha_{i+k}$ starts with $\klparentf{m}{f}$
and ends with $\krparentf{m}{f}$.

The working of the transducer $\T_h$ is based on the above observations.  It starts by
reading the opening parenthesis $\lparent{h}$ that indicates that domain of $h$ is about
to be read.  If the next character read is an opening parenthesis $\lparent{f}^{1}$, then
it means that the parsing of $w$ contains $n\geq k$ $L$-factors.
Otherwise, it indicates that the parsing of $w$ contains $n<k$ $L$-factors.

In the case where the parsing $w$ contains less than $k$ $L$-factors, $\T_h$ remains in
this state, where it reads letters from $\Sigma\cup\{\#_{e}\}$, while producing nothing
until a closing parenthesis $\rparent{h}$ is read.  When it reads $\rparent{h}$ at the end
of $w$, it goes to the accepting state.

Otherwise, it reads an opening parenthesis $\lparent{f}^{1}$ that denotes the beginning of
the first block of $w$.  On reading $\lparent{f}^{1}$, $\T_h$ moves to the initial state
of $\T^1_f$.  
We know from the construction that $\T^1_f$ ignores all parsing symbols that are not indexed by $1$.  
The run of $\T_{f}^{1}$ goes on until we see a closing parenthesis $\rparent{f}^{1}$,
which means that $\T_h$ has finished processing the first block.

Then, $\T_h$ should go back and read the next block, if it exists.  In general, the block
$\alpha_{i+1}\cdots\alpha_{i+k}$ is processed by $\T_{f}^{i+1 \bmod k}$.  The decision
whether to go back or not (in other words, whether the block just read is the last one) is
taken depending on whether we see $\rparent{h}$ or $\#_{e} $ next.  If the closing
parenthesis $\rparent{h}$ is the next letter read, then $\T_h$ knows that the domain of
$h$ has been read completely, and therefore exits.  Otherwise, if a $\#_{e}$ is the next
letter read, then $\T_h$ knows that there are more $L$-factors to the right of the current
position, which implies that this was not the last block.  So, $\T_h$ will go back on the
parsed word until it reads $\lparent{f}^{i+1 \bmod k}$, which signals the beginning of the
next block.  Then, it repeats the above process on the next block by going to the initial
sate of $\T^{i+1 \bmod k}_f$.

In Figure~\ref{fig:k-star-runs}, we illustrate the run of $\T_h$ on an example, where
$k=3$ and $n=6$.

\begin{figure}[h]
  \centering
    \includegraphics[width=0.5\linewidth]{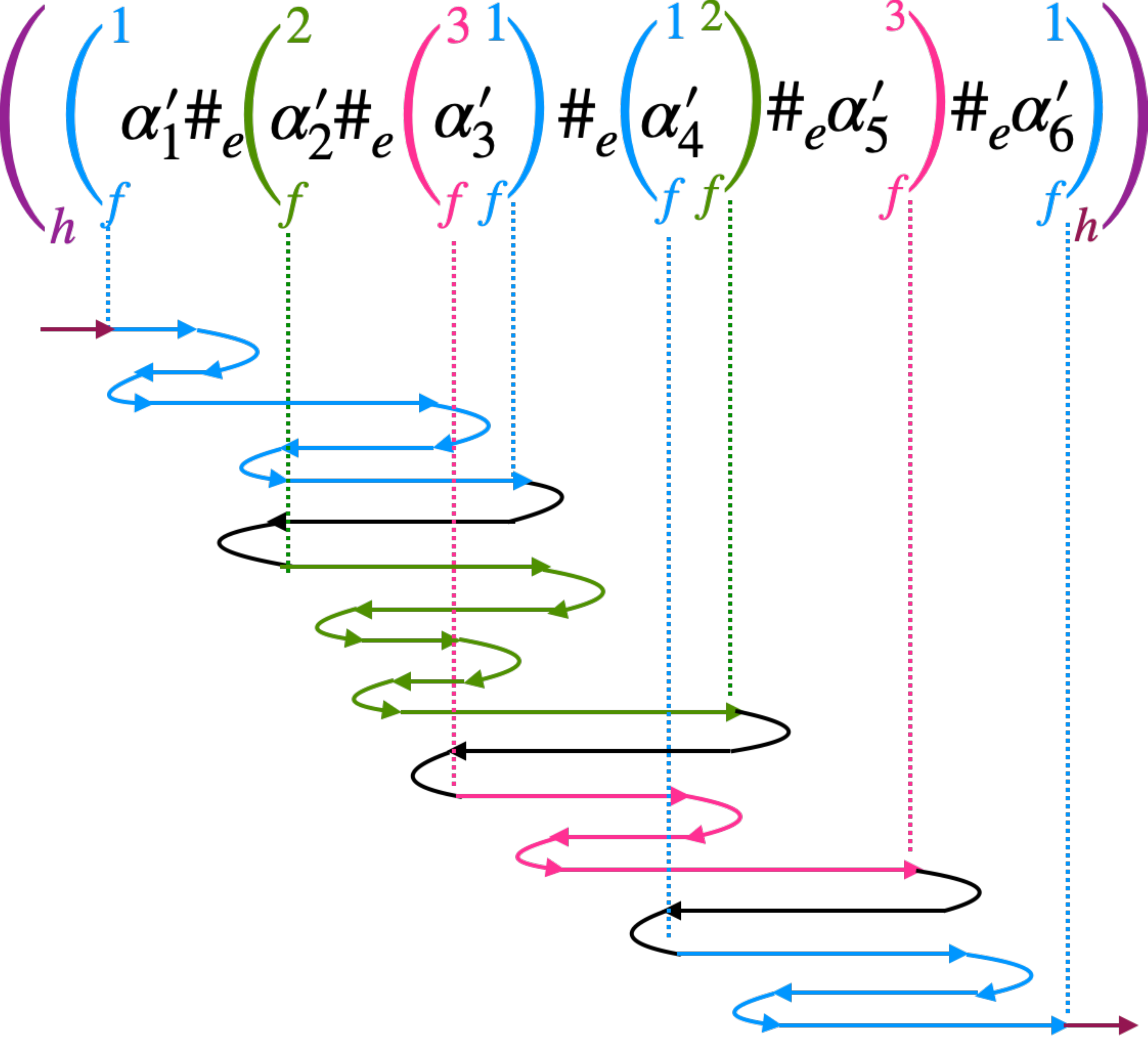}
    \captionof{figure}{Run of the transducer $\T_h$ on $\Parse_h(w)$, where $w = u_1u_2u_3u_4u_5u_6$, when $k=3$. The blue zig-zag arrows ignores all 
    paratheses not indexed by 1, and represents the run of $\T_f^1$, likewise, the green zig-zag arrows ignores all 
    paratheses not indexed by 2, and represents the run of $\T_f^2$, the 
    magenta  zig-zag arrows ignores all 
    paratheses not indexed by 3. The black arrow represents the backward movement of $\T_h$ from $)_f^i$ looking for 
    $(_f^{i+1}~mod~k$.}
    \label{fig:k-star-runs}
\end{figure}

\myparagraph{Evaluator for reverse $k$-star}
We turn to the description of the evaluator $\T_h$ for $h = \krstar{k}{e}{f}$.  The 2RFT
$\T_h$ is depicted in Figure~\ref{fig:evaluator-k-r-star}.  We use the same $k$ copies
$\T^1_f, \T^2_f, \ldots, \T^k_f$ (described above) of the evaluator $\T_{f}$ for the \RTE
$f$.  Recall that the parsing relation is the same for $\kstar{k}{e}{f}$ or
$\krstar{k}{e}{f}$. Hence, we use the observations made above for a parsing 
$\parent{h}{\alpha_{1}\#_{e}\alpha_{2}\cdots\#_{e}\alpha_{n}}\in\Parse_{h}(w)$.

$\T_h$ starts by reading the opening parenthesis $\lparent{h}$ that indicates that domain
of $h$ starts.  It moves to a $+$ state, where it scans the parsed word 
without producing anything, until a closing parenthesis $\rparent{h}$ is reached.  On
reading an $\rparent{h}$, $\T_h$ knows that the domain of $h$ has been read completely and
goes to a $-$ state.  $\T_h$ then looks at the letter on the left.
If the letter is from $\Sigma \cup \{\#_e\}$, it means that the parsing of $w$ contains
strictly less than $k$ $L$-factors.  In this case, $\T_h$ reads the closing parenthesis
$\rparent{h}$ and exits.  Otherwise the letter is a parenthesis $\rparent{f}^{i}$, which
means that the parsing of $w$ has $n\geq k$ $L$-factors.  Moreover, we know that the last
block $\alpha_{n-k+1}\cdots\alpha_{n}$ starts with $\lparent{f}^{i}$.  So, $\T_h$ moves to
the left until it sees $\lparent{f}^{i}$.  When $\T_h$ sees the $\lparent{f}^{i}$, it is
at the beginning of the last block, and on reading $\lparent{f}^{i}$, it moves to the
initial state of $\T^i_f$, which processes the block (ignoring all parsing symbols that
are not indexed by $i$).  The run goes on until we see a closing parenthesis
$\rparent{f}^{i}$ which means that $\T_f^{i}$ has finished reading the block.

Now, $\T_h$ should go back and read the block on the left, if it exists.  
It first goes back until it sees $\lparent{f}^{i}$.
The decision whether there is another block on the left (in other words, whether the block
just read is not the leftmost one) is taken depending on whether we see on the left
$\#_{e}$ or $\lparent{h}$.
If the opening parenthesis $\lparent{h}$ is seen,
then it means that the block just read is the first (leftmost) block.  In this case,
$\T_h$ knows that it has finished processing the domain of $h$, and therefore does a
rightward run until it sees a closing parenthesis $\rparent{h}$, upon seeing which $\T_h$
exits.  In this case, $\T_h$ goes on a rightward run until it sees the closing parenthesis
$\rparent{h}$, and exits.
Otherwise, $\T_h$ sees $\#_{e}$ and
realises that there is at least one more block on the left to be processed.
It moves to the beginning of this block and repeats the above process 
by going to the initial sate of $\T^{i-1 \bmod k}_f$.

\begin{figure}[!h]
  \centering
    \gusepicture{Tk-r-star}
  \caption{Evaluator for $h = \krstar{k}{e}{f}$.  Here, $x_i\in\{\kaparent{j}\mid j\neq
  i\}\cup\{\#_{e}\}$, $y_i \neq \lparent{f}^i, \rparent{f}^i, \lparent{f}^{i+1 \bmod k}$,
  $z_i \neq \rparent{f}^{i}, \lparent{f}^{i}$, and $\alpha \neq \lparent{h}, \rparent{h}$.
  Note that $\kaparent{}$ denotes any parenthesis.}
  \label{fig:evaluator-k-r-star}
\end{figure}

\begin{lemma}\label{lem:evaluator kstar}
  For all RTEs $h$, the number of states of the evaluator for $h$ is $\norm{\T_h}\leq
  5|h| \cdot \w{h}$.
\end{lemma}

\begin{proof}
  The proof is by structural induction on the expression $h$.  
  We have already seen that, when $h$ does not use $k$-star or reverse $k$-star, then
  $\norm{\T_h}\leq 5|h|$.

  We will now consider the case where $h$ is a $k$-chained Kleene-star expression
  ($h=\kstar{k}{e}{f}$) or a reverse $k$-chained Kleene-star expression
  ($h=\krstar{k}{e}{f}$).
  For both cases, we easily see that $\norm{\T_h} = k\norm{\T_f} +3k +8	$.
  By induction hypothesis, we have $\norm{T_f}\leq 5|f| \cdot {\w{f}}$.
  We get $\norm{\T_h} \leq 5k|f|\cdot\w{f}+3k+8\leq 5(|f|+k+2)(k\cdot\w{f}+1)$.
  Therefore, we get $\norm{\T_h}\leq 5|h| \cdot \w{h}$. 
\end{proof}

\subsection{Parser for $k$-star and reverse $k$-star}\label{sec:k-star-parsing-trasnducer}

In this section, we propose the parser $\TParse_h$ for $h = \kstar{k}{e}{f}$. 
The idea we use to construct $\Parse_h(w)$ is that on the input word, whenever we finish
reading an $L$-factor $u_{i}$, we mark this by adding a $\#_{e}$ indicating the end of an
$L$-factor, and we start an instance of the transducer $\TParse_{f}$ in which whenever a
parenthesis is output, it will be indexed by $i+1 \bmod k$.
Reading an $L$-factor can be detected by running in parallel, the automaton $\aut_{e}$ for $e$ obtained via the Glushkov algorithm.
We also employ a counter that keeps track of how many factors of $L$ we have seen so far in the current factorization of $w$ being considered - 
the counter stores $i \bmod k$ when we are reading the $i$th $L$-factor in the factorization of $w$. 
Then, whenever we reach an accepting state of $\aut_{e}$ with counter value $i$, the parser guesses whether or not there are at least $k$ $L$-factors left in the factorization of $w$.
If the parser guesses yes, an instance of the transducer $\TParse_f$ which adds the index $i+1 \bmod k$ to its parentheses is initialized and run on the next $k$ $L$-factors of $w$.
If the parser guesses that only fewer than $k$ $L$-factors are remaining, then no new instance of the transducer $\TParse_f$ is initialized.
We will show that $k$ copies of $\TParse_{f}$ suffice to implement the above idea. 
Further, we employ a variable that ensures the order of parenthesisation required by the
definition of the parsing relation.

We turn to the formal definition of the parsing transducer $\TParse_h$.
Let $\aut_{e} = (Q, \Sigma, q_I, F, \Delta)$ be the Glushkov automaton constructed 
from the regular expression $e$.
The parser $\TParse_h$ has two main components that we define separately.  
The first and the more involved one $\TParse_h'$ is used in the generic case for decomposition of words having at least $k$ factors in $L(e)$.  The second one $\TParse_h''$, defined afterwards, handles the decompositions having less than $k$ $L(e)$-factors.

Let $\TParse_f = (Q_f, \Sigma, B, q^f_I, q^f_F, \Delta_f, \mu_f)$ be the 1-way transducer
that produces the parsing \wrt expression $f$.
We define the 1NFT $\TParse'_h = (Q_h,\Sigma,B_h,q^h_I,F^h,\Delta_h,\mu_h)$ by
\begin{itemize}[nosep]
  \item $Q_h = \{1,2,\cdots,k\}\times Q\times (Q_f \cup \{q_{\bot}\})^k \times\{1,2,\cdots,k\}$

  \item The input alphabet is $\Sigma$, the same as that of $\TParse_f$ and $\aut_{e}$.
  
  \item The output alphabet $B_h = B_1 \cup B_2 \cup \cdots B_k\cup\{\#_e\}$, where $B_i$ is the
  output alphabet $B$ of $\TParse_f$ where each parenthesis is indexed by $i$.
  Note that $\Sigma$ is contained in each $B_i$.

  \item The initial state is $q_{I}^{h}=(1, q_I, q^f_I, q_{\bot}, \cdots, q_{\bot}, 1)$.

  \item The set of final states is
  $F^h = \{(i,q,q_1,\ldots,q_k,j) \mid q\in F, q_{i+1\bmod k}=q_F^{f}, \text{ and }
  q_{\ell}=q_{\bot} \text{ for } \ell\neq i+1\bmod k\}$.
  
  \item The transition relation $\Delta_{h}$ of $\TParse'_h$ is defined below.

  Let $(i,q,q_1,\ldots,q_k,j)$ be a state of $\TParse'_h$ and let $m=i+1\bmod k$.
  \begin{enumerate}[nosep]
    \item if $q \xrightarrow{a} q'$ in $\aut_{e}$ and for $1 \leq \ell \leq k$ either 
    $q_\ell\xrightarrow{a \mid a} q'_\ell$ in $\TParse_f$ or $q_{\ell}=q_{\bot}=q'_{\ell}$, then
    $(i,q,q_1,\ldots,q_k,j) \xrightarrow{a \mid a} (i,q',q'_1,\ldots,q'_k,m)$ in $\TParse'_{h}$.
    Note that the last component is reset to $m=i+1\bmod k$ which is the least element \wrt 
    $\leq_{i}$.

    \item if there is an $\varepsilon$-transition
    $q_\ell\xrightarrow{\varepsilon\mid\aparent{g}}q'_\ell\neq q_{F}^{f}$ 
    in $\TParse_f$, and $j\leq_{i}\ell$ and $q_{m}\neq q_{F}^{f}$, then we have 
    $(i,q,q_1,\ldots,q_k,j)
    \xrightarrow{\varepsilon\mid\kaparentf{\ell}{g}}
    (i,q,q_1,\ldots,q'_{\ell},\ldots,q_k,\ell)$ in $\TParse'_{h}$. 
    Note that the last component is set to $\ell$ which prevents executing next an
    $\varepsilon$-transition (case 2) in a component $\ell'<_{i}\ell$ since this would
    produce a parenthesis indexed $\ell$ followed by a parenthesis indexed $\ell'$ in the
    wrong order.

    \item if $q\in F$ is an accepting state of $\aut_{e}$ and $q_{i}\neq q_{I}^{f}$
    and $q_{m}\xrightarrow{\varepsilon\mid\rparent{f}}q'_{m}=q_{F}^{f}$ is a
    transition in $\TParse_{f}$, then 
    $(i,q,q_1,\ldots,q_k,j) \xrightarrow{\varepsilon\mid \krparentf{m}{f}}
    (i,q,q_1,\ldots,q'_{m}=q_{F}^{f},\ldots,q_k,j)$ in $\TParse'_{h}$.
    
    Notice that a transition of case 2 cannot be applied after a transition of case 3 
    since the state of the component $m$ is now $q_{F}^{f}$.
    
    \item if $q\in F$ and $q_{i}\neq q_{I}^{f}$ and $q_{m}\in\{q_{F}^{f},q_{\bot}\}$ 
    and $(i,q,q_1,\ldots,q_k,j)\notin F^{h}$, then we have
    $(i,q,q_1,\ldots,q_k,j)
    \xrightarrow{\varepsilon\mid \#_{e}}
    (m,q_{I},q_1,\ldots,q'_{m},\ldots,q_k,m)$ in $\TParse'_{h}$
    where $q'_{m}=q_{\bot}$ if $q_{i}=q_{\bot}$ and 
    $q'_{m}\in\{q_{\bot},q^{f}_{I}\}$ otherwise.
    
    Note that the last component is set to $m$ which is the maximal element \wrt
    $\leq_{m}$.  Hence, only component $m$ may perform $\varepsilon$-transitions
    producing parentheses indexed by $m$ (case 2) until a letter $a\in\Sigma$ is read
    (case 1) or until we use again a switching transition of case 3 or case 4, which is only
    possible if the initial state $q_{I}$ of $\aut_{e}$ is also accepting ($q_{I}\in F$),
    i.e., when $\varepsilon\in L(e)$.
  \end{enumerate}
\end{itemize}
Note that, even if all three conditions of case 3 are satisfied, $\TParse'_h$ may choose
not to execute the corresponding transition; it may still execute transitions from cases 1
or 2.  Also, a transition from case 3 is either the last one in the run or it must be
followed by a transition of case 4.  

We will now define a transducer $\TParse''_h$ that takes care of words in $\bigcup_{n<k}L(e)^{n}$.
Recall that the automaton $\aut_{e} = (Q, \Sigma, q_I, F, \Delta)$ recognizes $L(e)$.
$\TParse''_h$ is defined as the 1-way transducer whose
\begin{itemize}[nosep]
  \item set of states is $Q \times \{1,2,\cdots,k-1\}$,
  \item input alphabet is $\Sigma$ and output alphabet is $\Sigma\cup\{\#_{e}\}$,
  \item initial state is $(q_I, 1)$,
  \item set of final states is $F'' = F \times \{1,2,\cdots,k-1\}$,
  \item transition relation is defined as follows:
  \begin{itemize}[nosep]
    \item $(q,i) \xrightarrow{a \mid a} (q',i)$ if $q \xrightarrow{a} q'$ in $\Delta$,
    \item $(q,i) \xrightarrow{\varepsilon \mid \#_e} (q_I,i+1)$ if $q\in F$ and $i+1<k$.
  \end{itemize}
\end{itemize}
$\TParse''_h$ just counts the number of $L(e)$-factors read upto $k-1$ and adds the
corresponding separators between them.  If $w\in L(e)^{n}$ with $n<k$, then $\TParse''_h$
has a run from the initial state $(q_I,1)$ to a final state $(q,n)$ where $q \in F$.  The
output function of $\TParse''_h$ is then the identity with $\#_e$ symbols inserted.

The 1-way transducer for $h = \kstar{k}{e}{f}$ is $\TParse_h$ given in
Figure~\ref{fig:parser-k-star}.

\begin{figure}[!h]
  \centering
    \gusepicture{k-star}
  \caption{Parser for $h = \kstar{k}{e}{f}$}
  \label{fig:parser-k-star}
\end{figure}

\begin{lemma}\label{lem:correctness parser k-star}
  Let $h$ be an \RTE. 
  The parser $\TParse_{h}$ computes the parsing relation $\Parse_{h}$.
  The number of states of the parser is $\norm{\TParse_{h}}\leq|h|^{\w{h}}$.
\end{lemma}

\begin{proof}
  As before, the proof is by structural induction and the only new case is when
  $h=\kstar{k}{e}{f}$. We first show that $\rsem{\TParse_{h}}\subseteq\Parse_{h}$.
  
  It is easy to see that the relation defined by $\TParse''_{h}$ is the set of pairs 
  $(w,\alpha)$ where $w=u_{1}u_{2}\cdots u_{n}$ with $n<k$ and $u_{i}\in L(e)$, and
  $\alpha=u_{1}\#_{e}u_{2}\cdots\#_{e}u_{n}$.  We get $\parent{h}{\alpha}\in\Parse_{h}(w)$.
  
  Let $\rho$ be an accepting run of $\TParse'_{h}$ reading an input word 
  $w\in\Sigma^{\star}$. We factorize the run as 
  $\rho=\rho_{1}\delta_{1}\rho_{2}\cdots\delta_{n-1}\rho_{n}$ where $\delta_{i}$ are the 
  transitions of case 4, i.e., labelled $\varepsilon\mid\#_{e}$. Let $u_{i}$ and 
  $\alpha_{i}$ be respectively the input read by $\rho_{i}$ and the output produced by 
  $\rho_{i}$. We have $w=u_{1}\cdots u_{n}$ and the output produced by $\rho$ is 
  $\alpha=\alpha_{1}\#_{e}\alpha_{2}\cdots\#_{e}\alpha_{n}$. We will show that 
  $\parent{h}{\alpha}\in\Parse_{h}(w)$.
  
  The counter modulo $k$ which is the first component of states of $\TParse'_{h}$ is
  incremented only by transitions of case 4.  Hence, during $\rho_{i}$ the counter is
  constantly $i\bmod k$ and the transition $\delta_{i}$ increments it to $i+1\bmod k$.
  From the condition $q\in F$ in case 4 or by definition of $F_{h}$, we deduce that the
  projection of $\rho_{i}$ on the second component is an accepting run of $\aut_{e}$ for
  the input word $u_{i}$.  Hence $u_{i}\in L(e)$.  It is also easy to see that
  $\pi_{\Sigma}(\alpha_{i})=u_{i}$.
  Before $\rho$ can reach an accepting state of $F_{h}$, its third component which started 
  initially with $q_{I}^{f}$ has to reach $q_{F}^{f}$ which is possible only by a 
  transition of case 3 when the first component which counts modulo $k$ has the value 
  $k$. We deduce that $n\geq k$. 
  
  We show below that conditions (1-6) on page~\pageref{def:parsing k-star} defining the
  parsing relation are satisfied.  We deduce that $\parent{h}{\alpha}\in\Parse_{h}(w)$ as
  desired.
  \begin{enumerate}[nosep]
    \item Let $0\leq i\leq n-k$ and $m=i+1\bmod k$.  We consider the projection $\rho'$ on
    the component $2+m$ of the subrun $\rho_{i+1}\cdots\rho_{i+k}$ reading $u_{i+1}\cdots
    u_{i+k}$ and producing $\alpha_{i+1}\#_{e}\cdots\#_{e}\alpha_{i+k}$.  We can check
    that $\rho'$ is an accepting run of $\TParse_{f}$ reading $u_{i+1}\cdots u_{i+k}$ and
    producing the projection of $\alpha_{i+1}\cdots\alpha_{i+k}$ on $\Sigma\cup B_{m}$.
    We deduce that $\pi_{m}(\alpha_{i+1}\cdots\alpha_{i+k})\in\Parse_{f}(u_{i+1}\cdots
    u_{i+k})$ and condition (1) holds.
  
    \item Let $1\leq i<j\leq k$.  During the run $\rho_{i}$, the component $2+j$ of the
    states is constantly $q_{\bot}$ and we deduce that $\alpha_{i}$ does not contain a 
    parenthesis $\kaparent{j}$.
  
    \item Similarly, if $n-k+1\leq i<j\leq n$, then during the run $\rho_{i}$, the
    component $2+(j\bmod k)$ of the states is constantly $q_{\bot}$ and we deduce that
    $\alpha_{i}$ does not contain a parenthesis $\kaparent{j\bmod k}$.
    
    \item Let $i\geq k$ and $m=i+1\bmod k$.  As above, consider the projection $\rho'$ on
    the component $2+m$ of the subrun $\rho_{i-k+1}\cdots\rho_{i}$ reading
    $u_{i-k+1}\cdots u_{i}$ and producing $\alpha_{i-k+1}\#_{e}\cdots\#_{e}\alpha_{i}$.
    Since $\rho'$ is an accepting run of $\TParse_{f}$, it ends with some 
    $q_{2+m}\xrightarrow{\epsilon\mid\krparentf{m}{f}}q_{F}^{f}$, which must be the 
    projection of a transition of case 3. Now, a transition of 
    case 3 may only be followed by a transition of case 4. Hence, this is the last 
    transition of $\rho_{i}$ and we deduce that $\alpha_{i}$ ends with $\krparentf{m}{f}$.
    
    \item Let $i\leq n-k+1$ and $m=i\bmod k$.  We can see that $\rho_{i}$ starts from some
    state $(m,q_{I},q_{1},\ldots,q_{k},m)$ with $q_{m}=q_{I}^{f}$.  We have a transition
    $q_{m}=q_{I}^{f}\xrightarrow{\varepsilon\mid\lparent{f}}q'_{m}$ in $\TParse_{f}$.
    Hence the first transition of $\rho_{i}$ must be with case 2 and since $m$ is the
    maximal element \wrt the order $\leq_{m}$ it must be induced by the transition
    $q_{m}=q_{I}^{f}\xrightarrow{\varepsilon\mid\lparent{f}}q'_{m}$ of $\TParse_{f}$.  We
    deduce that it is labelled $\varepsilon\mid\klparentf{m}{f}$ and that $\alpha_{i}$
    starts with $\klparentf{m}{f}$.
    
    Let $i>n-k+1$ and $m=i\bmod k$.  We can see that $\rho_{i}$ starts from some state
    $(m,q_{I},q_{1},\ldots,q_{k},m)$ with $q_{m}=q_{\bot}$.  Since $m$ is the maximal
    element \wrt the order $\leq_{m}$, the first transition of $\rho_{2}$ cannot be from
    case 2.  Either it is from case 1 and $\alpha_{i}$ starts with a letter from $\Sigma$,
    or it is from case 3 and $\alpha_{i}=\krparentf{i+1\bmod k}{f}$, or it is from case 4
    and $\alpha_{i}=\varepsilon$.
    
    \item Assume that $\alpha_{i}$ has two consecutive parentheses 
    $\kaparent{j}\kaparent{\ell}$ with $\kaparent{\ell}\neq\krparentf{i+1\bmod k}{f}$. The 
    two parentheses have been produced by consecutive transitions of case 2. We get 
    $j\leq_{i}\ell$.
  \end{enumerate}
  
  \medskip
  
  Conversely, we prove that $\Parse_{h}\subseteq\rsem{\TParse_{h}}$.
  Let $\parent{h}\alpha\in\Parse_{h}(w)$.  We write $w=u_{1}u_{2}\cdots u_{n}$ with
  $n\geq0$, $u_{i}\in L(e)$ for $1\leq i\leq n$, and
  $\alpha=\alpha_{1}\#_{e}\alpha_{2}\cdots\#_{e}\alpha_{n}$ such that either $n<k$ and
  $\alpha_{i}=u_{i}$, or $n\geq k$ and conditions (1-6) on page~\pageref{def:parsing
  k-star} defining the parsing relation are satisfied. In the first case, it is clear 
  that $(w,\alpha)$ is in $\rsem{\TParse''_{h}}$. So we assume that $n\geq k$ and we will 
  show that $(w,\alpha)$ is in $\rsem{\TParse'_{h}}$.
  
  For each $0\leq i\leq n-k$, with $m=i+1\bmod k$, condition (1) implies that we have
  $\pi_{m}(\alpha_{i+1}\cdots\alpha_{i+k})\in\Parse_{f}(u_{i+1}\cdots u_{i+k})$.  We
  choose a corresponding accepting run $\sigma_{i}$ of $\TParse_{f}$.  We write
  $\sigma_{i}=\sigma_{i}^{1}\cdots\sigma_{i}^{k}$ where $\sigma_{i}^{\ell}$ reads
  $u_{i+\ell}$ and produces $\pi_{m}(\alpha_{i+\ell})$ (this factorization is unique since
  each transition of $\TParse_{f}$ produces a single symbol).  
  
  Now, let $1\leq j\leq n$, and consider an accepting run $\rho'_{j}$ for $u_{j}$ in
  $\aut_{e}$.  The output word $\alpha_{j}$ determines a unique way to order transitions
  of $\rho'_{j}$ and transitions of the runs $\sigma_{i}^{\ell}$ with $i+\ell=j$,
  synchronizing transitions reading letters from $\Sigma$ and interleaving transitions
  reading $\varepsilon$ and producing parentheses.  Using conditions (4,5,6) on
  $\alpha_{j}$, we can check that following the above order we obtain the run $\rho_{j}$
  using transitions of types (1,2,3) and such that the projection of $\rho_{j}$ on the
  second component (resp.\ on component $2+(j-\ell+1\bmod k)$) is $\rho'_{j}$ (resp.\
  $\sigma_{j-\ell}^{\ell}$).  Notice that, during the run $\rho_{j}$, the first component
  is constantly $j$ and the last component starts with $j$ and is then deterministically
  determined by each transition.
  
  \medskip
  
  We conclude the proof by showing the upper bound on the number of states of 
  $\TParse_{h}$. By Lemma~\ref{lem:correctness parser Hadamard} we already know that the 
  upper bound is valid when the expression does not use $k$-star or reverse $k$-star. So, 
  again, the only new cases to consider in the induction is when $h=\kstar{k}{e}{f}$ or 
  $h=\krstar{k}{e}{f}$. In both cases, the parser is the same and its number of states is
  $$
  \norm{\TParse_h} = k^{2} (\mathsf{nl}(e) + 1) (\norm{\TParse_f} + 1)^k 
  + (k-1) (\mathsf{nl}(e) + 1) + 2 \,.
  $$
  Recall that, in both cases, $|h|= 1 + \mathsf{nl}(e) + |f| + k + 1$ and 
  $\w{h}= k \times \w{f} + 2$. Using induction hypothesis 
  $\norm{\TParse_{f}}\leq|f|^{\w{f}}$ and the fact $(a^b+1) \leq (a+1)^{b}$ when $a,b>0$, 
  we get
  \begin{align*}
    \norm{\TParse_h} 
    &= k^2 (\mathsf{nl}(e) +1) (\norm{\TParse_f} + 1)^{k} + (\mathsf{nl}(e) +1) (k - 1) + 2\\
    &\leq k^2 (\mathsf{nl}(e) + 1) (|f|^{\w{f}} + 1)^{k} + (\mathsf{nl}(e) +1) (k - 1) + 2\\
    &\leq k^2 (\mathsf{nl}(e) + 1) (|f| + 1)^{k\cdot \w{f}} + (\mathsf{nl}(e) +1) (k - 1) + 2
    \,.
  \end{align*}
  On the other hand, considering only three terms of the binomial expansion for the first
  inequality, we have
  \begin{align*}
    |h|^{\w{h}} &= (k + (\mathsf{nl}(e) +1) + (|f| + 1))^{k \cdot \w{f} + 2} \\
    &\geq
    (k \cdot \w{f}+2) \cdot k \cdot (\mathsf{nl}(e) +1) \cdot (|f| + 1)^{k \cdot \w{f}} 
    + 1 + 1 \\
    &\geq  
    k^2 (\mathsf{nl}(e) + 1) (|f|^{\w{f}} + 1)^{k} + (\mathsf{nl}(e) +1) (k - 1) + 2 \,.
  \end{align*}
  We deduce that $\norm{\TParse_{h}}\leq|h|^{\w{h}}$.
\end{proof}

\section{Reversible transducer for the unambiguous semantics of \RTEs}\label{sec:unamb}

In this section, we will discuss how to check if a word is in the unambiguous domain of an
\RTE. As already discussed in Section~\ref{sec:fsem}, the unambiguous domain $\udom{h}$ of
an \RTE $h$ is defined as the set of words $w \in \dom{h}$ such that parsing $w$ according
to $h$ is unambiguous.  Further, from Theorem~\ref{thm:main-relational}, we know that
$\udom{h}$ coincides with $\fdom{\TParse_{h}}$, which is the set of words $w$ such that
$\rsem{\TParse_h}(w)$ is a singleton.  Making use of this observation, we will check if a
word $w$ is in the unambiguous domain of $h$ by checking whether $\TParse_h$ is functional
on $w$.

Let $h$ be an \RTE. Let $T_h^o$ denote the automaton obtained from the parser $\TParse_h$
by erasing the inputs on transitions and reading the output instead.  Recall that from
Theorem~\ref{thm:main-relational}, for each parsing $\alpha$ of $w$ \wrt $h$, the
projection of $\alpha$ on $\Sigma$ is $w$.  Now, we claim that in order to check for
functionality of $\TParse_h$ on a word $w$, it is sufficient to check whether $T_h^o$
accepts two words $\alpha\neq\beta$ having the same projection $w$ on $\Sigma$.  In the
rest of this section, we will give a construction that checks this and show its
correctness.  Specifically, we will compute an automaton $B_{h}$ from $T^o_h$, such that
$B_{h}$ accepts the language
$\dom{h}\setminus\udom{h}=\dom{\TParse_{h}}\setminus\fdom{\TParse_{h}}$.

Let $\TParse_h = (Q_h, \Sigma, A, q^h_I, q^h_F, \Delta_h, \mu_h)$ be the 1-way transducer
that produces the parsing \wrt the expression $h$.
Then, $B = (Q, \Sigma, q_I, F, \Delta)$ where $Q = Q_h \times Q_h \times \{0,1\}$,
$q_I = (q^h_I, q^h_I, 0)$.
A state $(p,q,\nu)$ is accepting, i.e., $(p,q,\nu)\in F$, if we find two runs
$p\xrightarrow[+]{\varepsilon\mid x}q^h_F$ and $q \xrightarrow[+]{\varepsilon\mid y}
q^h_F$ in $\TParse_{h}$ with $\Pi_{\Sigma}(x)=\Pi_{\Sigma}(y)=\varepsilon$, and in
addition, $x\neq y$ if $\nu=0$.
The transition relation $\Delta$ is given by the following rules, where
$p \xrightarrow[+]{a\mid xa} p'$ in the premises denotes that there is a sequence
of transitions in $\TParse_{h}$ that reads $a$ and produces $xa$.

\begin{prooftree}
  \def\defaultHypSeparation{\hskip .05in}
  \AxiomC{$p \xrightarrow[+]{a\mid xa} p'$}
  \AxiomC{$q \xrightarrow[+]{a\mid xa} q'$}
  \AxiomC{$\Pi_{\Sigma}(x) = \varepsilon$}
  \RightLabel{\scriptsize(1)}
  \TrinaryInfC{$(p,q,0) \xrightarrow{a} (p',q',0)$}
\end{prooftree}
\begin{prooftree}
  \def\defaultHypSeparation{\hskip .05in}
  \AxiomC{$p \xrightarrow[+]{a\mid xa} p'$}
  \AxiomC{$q \xrightarrow[+]{a\mid ya} q'$}
  \AxiomC{$x \neq y$\quad $\Pi_{\Sigma}(x) = \Pi_{\Sigma}(y) = \varepsilon$}
  \RightLabel{\scriptsize(2)}
  \TrinaryInfC{$(p,q,0) \xrightarrow{a} (p',q',1)$}
\end{prooftree}
\begin{prooftree}
  \def\defaultHypSeparation{\hskip .05in}
  \AxiomC{$p \xrightarrow[+]{a\mid xa} p'$}
  \AxiomC{$q \xrightarrow[+]{a\mid ya} q'$}
  \AxiomC{$\Pi_{\Sigma}(x) = \Pi_{\Sigma}(y) = \varepsilon$}
  \RightLabel{\scriptsize(3)}
  \TrinaryInfC{$(p,q,1) \xrightarrow{a} (p',q',1)$}
\end{prooftree}

We will now show that $B$ accepts precisely the set of words that have multiple parsings
in $\TParse_h$.

\begin{lemma}
  $w \in \dom{h}\setminus\udom{h}$ iff $B$ has an accepting run on $w$. 
\end{lemma}

\begin{proof}
  Suppose that $B$ has an accepting run on $w$. 
  Then, by our construction, $T_h^o$ has two accepting runs $\alpha$ and $\beta$, such
  that $\alpha \neq \beta$ and $\Pi_{\Sigma}(\alpha) = \Pi_{\Sigma}(\beta) = w$.
  This in turn means that $\TParse_h$ has two runs reading $w$ and producing $\alpha$ and
  $\beta$ respectively. Therefore, $\alpha,\beta\in\Parse_{h}(w)$ and we get
  $w\in\dom{\Parse_{h}}\setminus\fdom{\Parse_{h}}=\dom{h}\setminus\udom{h}$.

  Conversely, suppose that $w\in\dom{h}\setminus\udom{h}
  =\dom{\Parse_{h}}\setminus\fdom{\Parse_{h}}$.  Let $\alpha,\beta\in\Parse_{h}(w)$ with
  $\alpha\neq\beta$.  Then, $\TParse_h$ has two runs reading $w$ and producing $\alpha$
  and $\beta$ respectively.  By Theorem~\ref{thm:main-relational} we have
  $\Pi_{\Sigma}(\alpha) =\Pi_{\Sigma}(\beta) =w$.  Hence, we can write $w = a_1 a_2 \cdots
  a_n$, $\alpha = u_0 a_1 u_1 \cdots a_n u_n$, and $\beta = v_0 a_1 v_1 \cdots a_n v_n$,
  where $\Pi_{\Sigma}(u_i) = \Pi_{\Sigma}(v_i) = \varepsilon$, for $0 \leq i \leq n$.
  Further, let $i$ be the least index such that $u_i \neq v_i$.  From the construction, we
  know that after reading $a_{1}\cdots a_{i}$,
  the automaton $B$ reaches a state $(p,q,0)$ (by using rule (1) repeatedly).  
  If $i = n$, then $(p,q,0)\in F$ and $B$ accepts $w$.
  If $i<n$, then $B$ uses the rule (2) to go to a state $(p',q',1)$ when reading $a_{i+1}$. 
  Then, we can use repeatedly rule (3) to read $a_{i+2}\cdots a_n$
  and reach a state in $F$.
  In both cases, $B$ has an accepting run on $w$.
\end{proof}

Note that the number of states of $B$ is $2\norm{\TParse_h}^2$, where $\norm{\TParse_h}$
denotes the number of states of $\TParse_h$.  From Theorem~\ref{thm:main-relational}, we
know that $\norm{\TParse_h}\leq|h|^{\w{h}}$ for general \RTEs and
$\norm{\TParse_h}\leq|h|$ when $h$ does not use Hadamard product, $k$-star or reverse $k$-star.  Thus, the number of states of $B$ is at most $2|h|^{2\cdot\w{h}}$ in general,
and $2|h|^2$ when $h$ does not use Hadamard product, $k$-star or reverse $k$-star.
Moreover, the construction takes time 
$\poly{(\norm{\TParse_h})}$, which is
$\poly{(|h|^{\w{h}})}$ 
for general \RTEs and $\poly{(|h|)}$
when $h$ does not use Hadamard product, $k$-star or reverse $k$-star.

Using Lemma~\ref{lem-Aut to Rev} on $B$, we obtain a reversible automaton $B'$ with number
of states $2\cdot 2^{|h|^{2\cdot\w{h}})}+6$ that accepts the complement of $L(B)$, i.e.,
$B'$ accepts the set of words having at most one parsing \wrt $h$.
Using Lemma~\ref{lem-1w to Rev} on $\TParse_h$, we obtain a uniformizing 
2RFT $\TParse'_h$ of size $2^{\mathcal{O}(|h|^{\w{h}})}$.

Let $\TParse^{U}_{h}$ be defined as the 2RFT that first runs the automaton $B'$ without
producing anything, then runs $\TParse'_h$ if $B'$ has accepted.  This transducer is
reversible since it is the sequential composition of reversible transducers, and computes
the parsing relation on words belonging to the intersection of the domain of the two
machines, i.e., $\udom{h}$ the set of words having exactly one parsing.  Its number of
states is the sum of the number of states of the two machines ($+1$ dummy state making the
transition), hence $\norm{\TParse^{U}_{h}}=2^{\mathcal{O}(|h|^{2\cdot\w{h}})}$.

Thanks to Theorem~\ref{thm:main-relational}, for any given \RTE $h$, we have the evaluator
2RFT $\T_h$ which, when composed with the parsing relation $\Parse_{h}$, computes the
relational semantics of $h$: $\rsem{h}=\sem{\T_h}\circ\Parse_h$.  Let $\T^{U}_{h}$ be the
2RFT obtained by the composition of $\T_h$ and $\TParse^{U}_{h}$ (which restricts the
parser to the unambiguous domain of $h$), i.e., $\T^{U}_{h} = \T_{h}\circ\TParse^{U}_{h}$.
Then, $\T^{U}_{h}$ computes the unambiguous semantics of $h$: $\usem{h}=\sem{\T^{U}_h}$.

Recall that the number of states of the evaluator is
$\norm{\T_{h}}=\mathcal{O}(|h|\cdot\w{h})$.  Then, as the composition of 2RFTs can be done
with polynomial blowup, we get $\T^{U}_{h}$ whose number of states is
$2^{\mathcal{O}(|h|^{2\cdot\w{h}})}$.
Note that, if $h$ does not use Hadamard, $k$-star or reverse $k$-star, then both
$\norm{\TParse_h}$ and $\norm{\T_h}$ are linear in $|h|$.  As a consequence, for an RTE
$h$ belonging to this fragment, we get $\T^{U}_{h}$ whose number of states is
$2^{\mathcal{O}(|h|^{2})}$.

\section{Conclusion}

We conclude with some interesting avenues for future work.  An immediate future work is to
adapt our parser-evaluator construction to work for \SDRTE. We believe that this can be
done with some effort, preserving our complexity bounds.  Note that in this paper, we have
already done the work to handle $\kstar{k}{e}{f}$ even though 2-star was sufficient for
\RTE; it remains to ensure the aperiodicity of the parser-evaluator, preserving our
complexity bounds, to make our construction work for \SDRTE. Another interesting question
is to see if our approach and construction can be made amenable for extended \RTE, which
also allows function composition as an operator of the \RTE syntax.  Note that this does not
increase expressiveness, but is a useful shorthand.  
Already in this paper, we have
considered some useful functions like duplicate and reverse along with the fragment
\RTE[\Rat, $\reverse$]; allowing function composition makes the use of these shorthands
meaningful.
As we construct reversible machines, composing them is effective and hence having composition as the topmost operator is straight-forward. However, using composition as an operator would require the parsing of intermediate outputs which we leave as an open problem.

\bibliography{APrev}

\end{document}

%% file: macros.tex
\usepackage{mathabx}

\usepackage{comment}
\usepackage{mathtools}
\usepackage{amsthm}
\usepackage{amssymb}
\usepackage{cite}
\usepackage{tikz}
\usetikzlibrary{fadings}
\usepackage{graphicx}
\usepackage[basic,classfont=caps]{complexity}
\usepackage[T1]{fontenc}
\usepackage{enumerate}
\usepackage{paralist}
\usepackage{algorithm}
\usepackage{algorithmic}
\usepackage{xspace}
\usepackage{accents}
\usepackage{stmaryrd}
\usepackage{lmodern}
\usepackage{enumitem}
\usepackage{todonotes}
\usepackage{microtype}
\usepackage{xxcolor}
\usepackage{thm-restate}
\newcommand{\RTE}{\textsf{RTE}\xspace}
\newcommand{\Rat}{\textsf{Rat}}

\newcommand{\twoplus}[2]{[#1,#2]^{2\scriptstyle{\boxplus}}}

\newcommand{\RTEs}{\textsf{RTE}s\xspace}

\newcommand{\SDRTE}{\textsf{SDRTE}\xspace}

\usepackage{bussproofs}
\usepackage{logicproof}[2014/03/20]
\usetikzlibrary{chains,matrix,scopes,decorations.shapes,arrows,shapes,trees}

\usepackage{wrapfig}
\usepackage{scrextend}

\newcommand{\myparagraph}[1]{\par\smallskip\noindent\textbf{#1.}}

\let\epsilon\varepsilon
\let\phi\varphi
\let\emptyset\varnothing
\let\rho\varrho
\renewclass{\P}{PTime}
\renewclass{\DTIME}{DTime}
\renewclass{\DSPACE}{DSpace}
\renewclass{\EXP}{ExpTime}
\newclass{\TWOEXP}{2ExpTime}
\renewclass{\EXPSPACE}{ExpSpace}
\newclass{\ACK}{Ack}
\newclass{\FDTIME}{FTtime}
\renewclass{\AP}{APTime}
\renewclass{\PSPACE}{PSpace}

\usetikzlibrary{positioning,fit,arrows,automata,calc,shapes,decorations.pathmorphing} 
\tikzset{->,>=stealth',
shorten >=1pt,shorten <=1pt,
auto,node distance=1.5cm,
every loop/.style={looseness=6},
initial text={},
every state/.style={inner sep=0.2mm, minimum size=0.5cm},
el/.style={font=\scriptsize},
inner sep=1mm,
loopright/.style={loop,looseness=6,out=30, in=-30},
loopleft/.style={loop,looseness=6,out=210, in=150},
loopabove/.style={loop,looseness=6,out=120, in=60},
loopbelow/.style={loop,looseness=6,out=300, in=240},
ve/.style={rectangle,draw,inner sep=1mm},
va/.style={circle,draw,inner sep=0.6mm},
}

\newcommand{\aut}{\mathcal{A}}

\newcommand{\tra}{\mathcal{T}}

\providecommand\lang{}
\renewcommand{\lang}{\mathcal{L}}
\newcommand{\trans}{\mathcal{R}}

\newcommand{\alp}{\Sigma}
\newcommand{\alpo}{\Gamma}

\newcommand{\vdashv}[1]{#1_{\vdash \dashv}}

\newcommand{\dom}{\mathsf{dom}}

\newcommand\sem[1]{[\![ #1 ]\!]}

\newcommand{\norm}[1]{\lVert{#1}\rVert}
\newcommand{\w}[1]{\mathsf{width}({#1})}
\newcommand{\rev}{._r,\star_r}

\newcommand{\reverse}{\rev}
\newcommand{\revfn}{\mathsf{rev}}
\newcommand{\duplicate}{\mathsf{dup}}

\newcommand{\wrt}{w.r.t.\ }

\newcommand{\fdom}[1]{\mathsf{fdom}(#1)}
\newcommand{\udom}[1]{\mathsf{udom}(#1)}
\newcommand{\rsem}[1]{\sem{#1}^R}

\newcommand{\usem}[1]{\sem{#1}^U}

\newcommand{\Parse}{P}
\newcommand{\TParse}{\mathcal{P}}
\newcommand{\T}{\mathcal{T}}

\newcommand\Ifthenelse[3]{{#1}\,?\,{#2}:{#3}}
\newcommand\SimpleFun[2]{{#1} \triangleright {#2}}
\newcommand{\kstar}[3]{[#2,#3]^{#1\star}}
\newcommand{\krstar}[3]{[#2,#3]_r^{#1\star}}
\newcommand{\rstar}[1]{#1_r^{\star}}
\renewcommand\dom[1]{\mathsf{dom}(#1)}

\newcommand{\lparent}[1]{{(}\!{\scriptscriptstyle #1}}
\newcommand{\rparent}[1]{{\scriptscriptstyle #1}\!{)}}
\newcommand{\aparent}[1]{{|}\!{\scriptscriptstyle #1}}

\newcommand{\parent}[2]{\lparent{#1}\, {#2} \,\rparent{#1}}
\newcommand{\restrict}[2]{\pi_{\text{-}{#1}}(#2)}

\newcommand{\klparent}[1]{{(}^{#1}}
\newcommand{\krparent}[1]{{)}^{#1}}
\newcommand{\kaparent}[1]{{|}^{#1}}

\newcommand{\klparentf}[2]{\lparent{#2}^{#1}}
\newcommand{\krparentf}[2]{\rparent{#2}^{#1}}
\newcommand{\kaparentf}[2]{{|}_{#2}^{#1}}

\newcommand{\Blue}[1]{\textcolor{blue}{#1}}

\newcommand{\gu}{\mathsf{gu}}
\newcommand{\lu}{\mathsf{lu}}